\documentclass{article}
\usepackage[letterpaper,margin=1in]{geometry}
\usepackage{iftex}
\ifXeTeX
  \usepackage[OT1]{fontenc}

\fi
\usepackage{natbib}
\usepackage{enumitem}
\usepackage{amsmath,amssymb,amsthm}
\usepackage{booktabs}
\usepackage{algorithm}
\usepackage{algorithmic}
\usepackage[hidelinks]{hyperref}
\newtheorem{theorem}{Theorem}
\newtheorem{lemma}{Lemma}
\newtheorem{proposition}{Proposition}

\theoremstyle{definition}
\newtheorem{definition}{Definition}
\newtheorem{remark}{Remark}
\newenvironment{paperparagraph}{\par\begingroup}{\par\endgroup}

\title{EFX for Additive Chores: Nonexistence, Pareto Incompatibility, and Bi-Valued Existence}
\author{%
  Wentao He \\
  Shanghai Jiao Tong University \\
  \href{mailto:dlpt13@sjtu.edu.cn}{\texttt{dlpt13@sjtu.edu.cn}} \\
  \and
  Biaoshuai Tao\thanks{Corresponding author. This paper was accepted at SODA 2027.} \\
  Shanghai Jiao Tong University \\
  \href{mailto:bstao@sjtu.edu.cn}{\texttt{bstao@sjtu.edu.cn}} \\
}
\date{}
\hypersetup{
  pdftitle={EFX for Additive Chores: Nonexistence, Pareto Incompatibility, and Bi-Valued Existence},
  pdfauthor={Wentao He and Biaoshuai Tao}
}

\begin{document}

\maketitle

\begin{abstract}
    We consider the fair division problem of indivisible chores and resolve the long-standing open problem for the existence of EFX (envy-free up to any item) allocations with additive cost functions. We show that, even for tri-valued additive cost functions, for every $n\geq 4$, there exists an instance with $n$ agents where no EFX allocation exists. Our counterexample only uses three types of chores and two types of agents. The numbers of types for chores and agents are both tight: an EFX allocation is known to exist for one type of agents (i.e., with identical cost functions) or two types of chores.

    We then consider bi-valued instances. We show that, for every $n\geq 4$, there exists an instance with $n$ agents where every EFX allocation is not Pareto-optimal. This is also the first example showing the incompatibility of EFX and Pareto-optimality when the costs of items are positive: existing examples showing the incompatibility of EFX and Pareto-optimality exploit items with $0$ costs. Our result shows such an example exists even for bi-valued instances. The number of agents $n$ is also tight: for $n\leq 3$, it is known that EFX is compatible with Pareto-optimality for bi-valued additive chores. Finally, for bi-valued additive chores, we also show that an EFX allocation is guaranteed to exist for $n=4$ via a constructive proof with an $O(m)$ time algorithm.
\end{abstract}

\section{Introduction}
\label{sec:intro}

Fair division studies how to allocate resources among agents with heterogeneous
preferences. 
The items to be allocated may be desirable resources, usually called
\emph{goods}, or undesirable tasks, usually called \emph{chores}. 
In the former
case agents derive utility from receiving items, whereas in the latter case agents
incur disutility or cost.  
Although goods and chores are similar in some
models, the direction of preferences changes the structure of many fairness
questions.  
This paper focuses on the allocation of indivisible chores among
agents with additive cost functions.

A central fairness principle is \emph{envy-freeness} (EF): no agent should prefer
another agent's bundle to her own.
For indivisible items, exact envy-freeness is
often impossible, even in very small instances (e.g., when the number of items is less than the number of agents).  
This has motivated the study of relaxations that require envy to disappear after removing a small amount of
responsibility.  The most widely studied such relaxation is \emph{envy-freeness
up to one item} (EF1)~\citep{Lipton04onapproximately,budish2011combinatorial}.  In the chore setting, an allocation is EF1 if whenever an
agent envies another agent, this envy can be eliminated by removing a single chore
from the envious agent's own bundle.  A stronger and more robust relaxation is
\emph{envy-freeness up to any item} (EFX)~\citep{caragiannis2019unreasonable}: for every pair of agents $i,j$, and
for every chore $g$ assigned to agent $i$, agent $i$ should not envy agent $j$
after $g$ is removed from her own bundle.  
Thus, EFX asks that no single chore in an agent's bundle is responsible for hiding
a violation of envy-freeness.

EFX has become one of the central open problems in discrete fair division.  
For goods with additive valuations, the existence of complete EFX allocations remains
open for four or more agents, despite substantial progress for special cases such
as three agents, a bounded number of valuation types, and allocations with a small
number of unallocated goods
\citep{caragiannis2019envy,chaudhury2021little,mahara2023existence,hv2025efx,berger2022almost,akrami2025efx}.
For chores, the situation has been even less understood.  Early work established EFX existence for two agents and for identical-ordering cost functions, and subsequent work gave approximation guarantees for additive chores
\citep{garg2026existence,garg2025constant,zhou2024approximately}.  
Exact EFX allocations are known for several
restricted additive domains, including instances with at most twice as many chores as agents, instances in which all but one agent have identical-ordering
cost functions, and three-agent personalized bi-valued instances~\citep{kobayashi2025efx}.  
EFX allocations are also known for two types of chores~\citep{aziz2023fair}, and EFX together with fractional Pareto-optimality is known
for three agents with bi-valued additive costs~\citep{garg2023new}.  
For binary additive costs, where item costs are in $\{0,1\}$, EFX and Pareto-optimal
allocations can be computed for any number of agents~\citep{tao2025existence}.
Recent work has also obtained EFX guarantees for restricted additive cost
functions, where each item has an inherent cost and an agent either incurs this
cost or zero~\citep{lin2026allocating}.

Despite this progress, two basic gaps remained.  First, for chores with additive costs, all
known exact EFX existence results relied on strong restrictions such as small
numbers of agents, binary or restricted costs, identical-ordering structure, or
bounds on the number of chores.  Negative results for EFX had been obtained for
more general non-additive cost functions, such as superadditive costs
\citep{christoforidis2024pursuit} and submodular costs~\citep{christoforidis2026note}, but they did not rule out EFX for additive
costs.  Thus, before this work, it was still possible that chores with additive costs
always admitted EFX allocations, perhaps at least when all item costs came from a
small finite set.

Second, the relationship between EFX and efficiency was poorly understood for
positive bi-valued chores.  We use Pareto-optimality (PO) as our efficiency
benchmark: an allocation is PO if there is no other allocation that weakly
decreases every agent's cost and strictly decreases at least one agent's cost.
For goods, the compatibility of fairness and efficiency has been studied
extensively; for example, maximum Nash welfare allocations satisfy EF1 and PO for goods \citep{caragiannis2019unreasonable}, and an EF1+PO allocation can be computed in pseudopolynomial time~\citep{barman2018finding}.
Under bi-valued valuation functions, maximum Nash welfare allocations satisfy EFX and PO~\citep{amanatidis2021maximum}, and an EFX+PO allocation can be computed in polynomial-time~\citep{bu2024best}.
For chores, EF1 and PO are now known to be compatible in general additive instances
\citep{mahara2026existence}, and earlier work proved EF1+PO for important
restricted classes such as bi-valued chores
\citep{ebadian2022fairly,garg2022fair}.  However, EF1 is substantially weaker
than EFX.  The strongest positive result for EFX and efficiency in positive bi-valued chores applied to three agents \citep{garg2023new}; for four or more
agents, it was not known whether EFX and PO are compatible.  Existing incompatibility results closest to this setting used binary-marginal or more general domains in which zero marginal costs may play a crucial role
\citep{tao2025existence}.  
This left open whether incompatibility persists for strictly positive bi-valued additive costs.

\subsection{Our Results}
\label{subsec:our-results}

We address both gaps.

Our first result gives a negative answer to the general existence question for
additive chores.  We show that for every number of agents $n\geq 4$, there is an
additive tri-valued chore instance with no EFX allocation (Theorem~\ref{thm:tri}).  In our construction,
all item costs are positive and take only three possible values.  Thus, EFX
allocations for chores are not guaranteed even under a very severe restriction on
the cost functions.  %
The existence question for additive goods remains open; our result settles the corresponding question for additive chores.
Our counterexample only uses three types of chores, which is tight, as an EFX allocation is known to exist for two types of chores~\citep{aziz2023fair}.
In addition, our counterexample only uses two types of agents, which is also tight: for one type of agents, all agents have the same cost functions, and a leximin solution gives an EFX allocation~\citep{plaut2020almost,chen2020fairness}.
The precise contrast with goods concerns the number of agent types: two types already suffice for nonexistence with additive chores, whereas additive goods admit EFX allocations with up to three agent types~\citep{hv2025efx}.

Our second result concerns the compatibility of EFX and Pareto-optimality for
bi-valued chores.  For every $n\geq 4$, we construct a strictly positive
bi-valued additive instance, with costs in $\{1,r\}$, such that every EFX
allocation is Pareto dominated (Theorem~\ref{thm:EFXPO}).
In particular, EFX and PO are incompatible even
when every chore has positive cost for every agent and only two cost levels are
allowed.
{To the best of our knowledge, incompatibility of EFX and PO had not previously been established for strictly positive additive chore costs.}
Since EFX+fPO allocations are known to exist for three bi-valued agents
\citep{garg2023new}, the number of agents in this incompatibility result is
tight.

Our third result complements the previous impossibility theorem by proving an
existence guarantee for the first unresolved number of agents in the bi-valued
setting. 
We show that every four-agent instance with bi-valued additive costs
{admits an EFX allocation (Theorem~\ref{thm:four}). For strictly positive costs, the allocation can be computed in $O(m)$ time under the unit-cost arithmetic model.}  
Together with the second result, this shows that for
four positive bi-valued agents, EFX existence and EFX+PO compatibility diverge:
an EFX allocation always exists, but in some instances no EFX allocation can be
Pareto-optimal.

\subsection{Related Work}
\label{subsec:related-work}

The fair division literature contains several parallel lines of work on goods,
chores, and mixed items.  For goods, EFX was introduced and studied as a
particularly compelling relaxation of envy-freeness.  The leximin solution yields
EFX allocations when agents' valuation functions are identical~\citep{plaut2020almost}.  Complete EFX allocations are known
for three agents with additive valuations \citep{chaudhury2024efx} and beyond additive valuations~\citep{akrami2025efx}, for any number of agents with general monotone valuation functions ({that need not be additive}) that have binary marginals~\citep{bu2023efx}, for any
number of agents when there are at most two additive valuation types
\citep{mahara2023existence}, and, more recently, for at most three additive
valuation types \citep{hv2025efx}.  Other work has shown that almost-complete EFX
allocations exist: one can obtain EFX after leaving a small number of goods
unallocated \citep{chaudhury2021little,caragiannis2019envy,berger2022almost}.
There is also a large literature on approximate EFX for goods
\citep{amanatidis2020multiple,amanatidis2024pushing,akrami2025efx}, on the
relationship between EFX and maximum Nash welfare
\citep{amanatidis2021maximum}, and on extensions from additive to more general
valuation classes \citep{mahara2024extension}.
It was only very recently that a first counterexample showing the non-existence of EFX allocations is known~\citep{akrami2026counterexample,mackenzie2026counterexamples}.
However, the example constructed by \citet{akrami2026counterexample,mackenzie2026counterexamples} uses submodular valuation functions.
The existence of EFX allocations for additive valuations is still open for the setting with goods.
See the survey~\citep{amanatidis2023fair} on recent progress on fair division with indivisible goods.

For chores, other than those aforementioned papers on EFX, EF1 has been studied extensively in combination with efficiency.
Bi-valued chores admit EF1+PO allocations computable in polynomial time
\citep{ebadian2022fairly,garg2022fair}; EF1+PO is known for three agents and
for instances with at most two disutility functions \citep{garg2023new}; and
EF1+PO is now known to exist for general additive chores
\citep{mahara2026existence}.  Other works study connections between fairness and
efficiency \citep{sun2021connections,barman2025introspectively,barman2026compatibility,barman2025fair,feldman2024optimal}, leximin and related rules
\citep{chen2020fairness}, (normalized) $p$-means
\citep{barman2020tight,eckart2024fairness}, proportionality-type guarantees
\citep{li2022almost}, and binary or supermodular cost domains
\citep{barman2023fair,tao2025existence}.

\begin{paperparagraph}
\paragraph{Graphical allocations and orientations.}
Graphical models represent agents as vertices and items as edges, with zero marginal value or cost to nonendpoints; orientations restrict assignments to endpoints.
For goods, \citet{christodoulou2023graphs} establish EFX existence on simple graphs under monotone valuations, whereas deciding EFX-orientation existence is NP-complete.
\citet{afshinmehr2025bipartite} give polynomial-time EFX allocation algorithms for additive valuations on bipartite multigraphs and multi-cycles; \citet{afshinmehr2026multigraphs} give such algorithms for cancelable valuations, including additive valuations, on arbitrary multigraphs.
\citet{zengmehta2025orientations} establish tight chromatic-number bounds: every bipartite simple graph admits EFX orientations for every monotone graphical valuation profile, whereas no graph of chromatic number at least four has this guarantee.
\citet{hsu2025multigraphs} proves that deciding EFX-orientation existence is NP-complete even on symmetric bipartite multigraphs with two edge weights.
For mixed manna, \citet{zhou2024complete} guarantee polynomial-time computable allocations under three of four EFX variants; deciding existence is NP-complete for the strongest variant and for all four orientation variants.
For chores, \citet{hsuking2025chores} give polynomial-time decision and construction algorithms for EF1 and $\mathrm{EFX}_0$ orientations on graphs without parallel edges, with NP-completeness on multigraphs.
In contrast, our $M_2$ edges cost $1$ to endpoints and $r>1$ to nonendpoints, and EFX may require nonendpoint assignments.
\end{paperparagraph}

Finally, several papers study models that combine goods and chores or give
unified treatments of positive and negative items
\citep{hosseini2023fairly,gafni2023unified,aziz2025fair,li2025truthfully,zhou2024complete}. 
Most relevant to our work, \citet{hosseini2023fairly} present an example showing non-existence of EFX allocations for the setting with mixed goods and chores, even for lexicographic additive valuations. 
These models highlight that results for
goods do not automatically transfer to chores, even when the mathematical
definitions look formally symmetric.  The present paper further illustrates this
asymmetry: exact EFX remains unresolved for additive goods, whereas for additive
chores we show that EFX allocations may fail to exist even with only three
positive cost levels.

Less relevant to our work, the model with mixed divisible and indivisible resources has also been studied extensively~\citep{bei2021fair,liu2024mixed,aziz2025fair}.
We will not elaborate on them here.

\section{Preliminaries}
A set $M$ of $m$ chores $g_1,\ldots,g_m$ is allocated to a set $N$ of $n$ agents $1,\ldots,n$.
We will use ``item'' and ``chore'' interchangeably in this paper to refer to an element in $M$.
Each agent $i$'s cost function $c_i:2^M\to\mathbb{R}_{\geq0}$ is additive:
$$c_i(S)=\sum_{g\in S}c_i(\{g\}).$$
We write $c_i(g_j)$ in short for $c_i(\{g_j\})$ for notation simplicity.
We will use ``cost function'' and ``disutility function'' interchangeably in this paper.
We say agent $i$ values item $g_j$ at $t$ if $c_i(\{g_j\})=t$.

\begin{paperparagraph}
An \emph{allocation} $X=(X_1,\ldots,X_n)$ is a partition of $M$ into pairwise disjoint bundles. When only a subset of $M$ is allocated, we call $X$ a \emph{partial allocation}; the same fairness inequalities apply to the allocated bundles.

Two agents $i,j$ have the same \emph{agent type} if $c_i(g)=c_j(g)$ for every $g\in M$. Two chores $g,h$ have the same \emph{chore type} if $c_i(g)=c_i(h)$ for every $i\in N$, that is, if their cost vectors $(c_1(g),\ldots,c_n(g))$ coincide. These notions count identical cost functions and identical cost vectors, respectively. They are distinct from the number of cost levels: a bi-valued instance with four agents can have up to $2^4=16$ chore types.
\end{paperparagraph}

\begin{definition}[EF for chores]
An allocation $X=(X_1,\ldots,X_n)$ is \emph{envy-free} if for every pair of agents $i,j$ we have $c_i(X_i)\leq c_i(X_j)$.
\end{definition}

Given an allocation $X=(X_1,\ldots,X_n)$, we say that agent $i$ does not envy agent $j$ if $c_i(X_i)\leq c_i(X_j)$.
If $c_i(X_i)\leq c_i(X_j)$ holds for a particular agent $i$ and every agent $j$, we say that $X$ is \emph{envy-free for agent $i$}, or \emph{agent $i$ is envy-free}.

\begin{definition}[EFX for chores]
An allocation $X=(X_1,\ldots,X_n)$ is \emph{EFX} if for every pair of agents $i,j$ we have either
\begin{itemize}
    \item $X_i=\emptyset$, or
    \item for every item $g\in X_i$, $c_i(X_i\setminus\{g\})\le c_i(X_j).$
\end{itemize}
Equivalently, if
\[
    \tau_i(X):=\tau_i(X_i):=
    \begin{cases}
        0, & X_i=\emptyset,\\
        c_i(X_i)-\min_{g\in X_i}c_i(g), & X_i\ne\emptyset,
    \end{cases}
\]
then $X$ is EFX if and only if $\tau_i(X)\le c_i(X_j)$ for all $i,j$.
\end{definition}

Given an allocation $X=(X_1,\ldots,X_n)$, we say that agent $i$ strongly envies agent $j$ if the EFX condition fails from $i$ to $j$: $\tau_i(X)>c_i(X_j)$.
If $\tau_i(X)\leq c_i(X_j)$ holds for a particular agent $i$ and every agent $j$, we say that $X$ is \emph{EFX for agent $i$}, or \emph{the EFX condition holds for agent $i$}.

\begin{definition}[Pareto-Optimality]
    An allocation $Y=(Y_1,\ldots,Y_n)$ \emph{Pareto-dominates} an allocation $X=(X_1,\ldots,X_n)$ if
    \begin{itemize}
        \item $c_i(Y_i)\leq c_i(X_i)$ holds for all $i\in\{1,\ldots,n\}$, and
        \item there exists $i\in\{1,\ldots,n\},\;$ $c_i(Y_i)<c_i(X_i)$.
    \end{itemize}
    An allocation $X=(X_1,\ldots,X_n)$ is \emph{Pareto-optimal} if it is not Pareto-dominated by any allocation.
\end{definition}

We will occasionally mention \emph{fractional Pareto-optimality}.
An allocation $X$ is fractionally Pareto-optimal if it is not Pareto-dominated by any allocation $Y$ even when $Y$ is allowed to be a fractional allocation (where items are divisible).
Fractional Pareto-optimality is stronger than Pareto-optimality, and we will not focus on this concept in this paper.

An instance is \emph{tri-valued} if there are three common cost levels $p,q,r\geq0$ such that $c_i(g)\in\{p,q,r\}$ for every agent $i$ and chore $g$. It is \emph{bi-valued} if there are two common levels $p,q\geq0$ such that $c_i(g)\in\{p,q\}$ for every $i,g$. We focus on $0<p<q$; the case $p=0<q$ reduces by scaling to binary costs, for which EFX+PO existence is known~\citep{tao2025existence}, and the case $p=q$ is immediate. Scaling all costs by $1/p$, we assume that $c_i(g)\in\{1,r\}$ for a common ratio $r=q/p>1$. A chore is \emph{small} for agent $i$ if $c_i(g)=1$ and \emph{large} for her if $c_i(g)=r$.

In contrast, in a \emph{personalized bi-valued} instance, each agent has her own two positive cost levels. Scaling each agent's costs separately, which preserves EFX, gives $c_i(g)\in\{1,r_i\}$ with potentially different ratios $r_i>1$. Our four-agent theorem concerns a common ratio $r$; we discuss the personalized case in Sect.~\ref{subsec:personalized}.

\section{Tri-Valued Instances without EFX Allocations}
\label{sec:tri}

In this section, we present our main result, which shows the non-existence of EFX allocations for some tri-valued instances.

\begin{theorem}\label{thm:tri}
    For any $n\geq 4$, there exists a tri-valued instance with $n$ agents where no EFX allocation exists.
\end{theorem}

{For ease of exposition}, we will present the example with $n=4$, which contains the main ideas.
The example can be naturally generalized to every fixed $n\geq 4$.
The proof of Theorem~\ref{thm:tri} for arbitrary $n\geq4$ is available in Appendix~\ref{append:tri-n}.

The instance consists of $4$ agents and $13$ items.
The items are partitioned into three groups:
$$A=\{a_1,a_2,a_3\},\quad B=\{b_1,b_2,b_3,b_4,b_5\},\quad\mbox{and}\quad C=\{c_1,c_2,c_3,c_4,c_5\}.$$
{The item group $A$ consists of} ``large'' items where each agent values each of them at $20$.
For each item in $B$, agents $1$ and $2$ value it at $1$ and agents $3$ and $4$ value it at $7$.
For each item in $C$, agents $1$ and $2$ value it at $7$ and agents $3$ and $4$ value it at $1$.
Table~\ref{tab:tri} illustrates the instance.

\begin{table}[h]
    \centering
    \caption{A tri-valued instance with $n=4$ and $m=13$ where no EFX allocation exists.  Each entry is the cost of every item in the corresponding item class.}

    \begin{tabular}{cccc}
    \hline
         & $A=\{a_1,a_2,a_3\}$ & $B=\{b_1,b_2,b_3,b_4,b_5\}$ & $C=\{c_1,c_2,c_3,c_4,c_5\}$\\
    \hline
       Agents $1$ and $2$  & $20$ & $1$ & $7$\\
       Agents $3$ and $4$  & $20$ & $7$ & $1$\\
    \hline
    \end{tabular}
    
    \label{tab:tri}
\end{table}

In the remaining part of this section, we will show that no EFX allocation exists in this instance.
{Suppose for contradiction that an EFX allocation} $X=(X_1,X_2,X_3,X_4)$ exists.
Firstly, we will show that no one can receive more than one large item in $A$.

\begin{proposition}
    No agent can receive more than one item in $A$.
\end{proposition}
\begin{proof}
    Each agent's total disutility for the entire item set is $100$. Therefore, for each agent $i$, we have $\min_{j=1,\ldots,4}c_i(X_j)\leq 25$.
    If we have $\tau_i(X_i)=c_i(X_i\setminus\{g\})>25$ for some $g\in X_i$ with minimum cost, the allocation cannot be EFX, as agent $i$ will envy agent $j$ with the minimum $c_i(X_j)$ even when $g$ is removed from $i$'s bundle $X_i$.

    {Suppose for contradiction that some agent} receives more than one item in $A$, and suppose this is agent $1$ without loss of generality.
    If $X_1$ contains three items in $A$, then $c_1(X_1)\geq 60$, and agent $1$'s disutility is still at least $40$ after removing one large item.
    We have seen that the disutility must be at most $25$ up to removing one item.
    This is a contradiction.

    If $X_1$ contains two items in $A$, then $c_1(X_1)\geq 40$ and $\tau_1(X_1)\geq20$.
    Since the total disutility is $100$, for some $i\in\{2,3,4\}$, we have $c_1(X_i)\leq (100-40)/3=20$.
    To guarantee EFX, $X_1$ must contain exactly two items from $A$.
    Moreover, we must have $c_1(X_2)=c_1(X_3)=c_1(X_4)=20$.
    It is easy to verify that this is impossible: one of the three bundles must contain exactly the remaining item from $A$; then each of the two remaining bundles can contain at most $2$ items from $C$ to avoid disutility exceeding $20$, but there are $5$ items in $C$.
\end{proof}

With the above proposition, it must be that one agent does not receive any large item in $A$ and each of the remaining three agents receives exactly one item in $A$.
Suppose without loss of generality that agent $1$ does not receive any item from $A$.

The next proposition shows that every agent's disutility for agent $1$'s bundle $X_1$ is at least $20$.
\begin{proposition}
    For any agent $i$, we have $c_i(X_1)\geq 20$.
\end{proposition}
\begin{proof}
    Suppose some agent $i\in \{2,3,4\}$ believes $c_i(X_1)<20$.
    Since agent $i$ has received a large item from $A$ by our assumption,  to maintain EFX, $X_i$ must contain exactly one item from $A$ and nothing else.
    Each of the remaining two agents in $\{2,3,4\}\setminus\{i\}$ has also received a large item. 
    To prevent her from strongly envying $i$, this agent can receive at most one more item other than the large item.
    The total number of items received by agents $2$, $3$, and $4$ is therefore at most $5$, including at least $3$ large items.
    Agent $1$ then has received at least $8$ items from $B\cup C$, with at least $3$ items from $B$ and at least $3$ items from $C$.
    Up to removing one item from $B$, we have $\tau_1(X_1)\geq 2\times 1+3\times 7=23$.
    This violates EFX since agent $i$ has received only one large item so $c_1(X_i)=20<\tau_1(X_1)$.
    Therefore, $c_i(X_1)\geq 20$ for each of $i=2,3,4$.
    Since $c_1(\cdot)$ is identical to $c_2(\cdot)$, the proposition holds for all agents.
\end{proof}

Now, we are ready to derive a contradiction.
Let $x=|X_1\cap B|$ and $y=|X_1\cap C|$.
Since agent $1$ receives no large item, we have $|X_1|=x+y$.
By applying our proposition $c_i(X_1)\geq 20$ for $i=1,2,3,4$, we must have
$$7x+y\geq 20\qquad\mbox{and}\qquad x+7y\geq20.$$
Since $x,y\leq 5$, it is easy to check that we must have
$$x\geq 3\qquad\mbox{and}\qquad y\geq3.$$
Therefore, $c_1(X_1)\geq 3\times 7+3=24$, and, up to removing an item with cost $1$, we have $\tau_1(X_1)\geq 23$.
To guarantee EFX, we must then have $c_1(X_i)\geq23$ for each $i\in\{2,3,4\}$.

Since {at least $3$ out of $5$ items from group $C$ have been taken} by agent $1$, at most two agents in $\{2,3,4\}$ can take items in $C$, and so at least one agent $i\in \{2,3,4\}$ does not take any item from $C$.
Moreover, agent $i$ can take at most two items from $B$ (since agent $1$ has taken at least three items from $B$).
Thus, $c_1(X_i)\leq 20+1+1=22$, and agent $1$ strongly envies agent $i$.

We conclude that $X$ cannot be EFX, and no EFX allocation exists in this instance.

\begin{remark}
    For general $n$, agents are partitioned into two groups with $\lfloor n/2\rfloor$ and $\lceil n/2\rceil$ agents respectively. Items are also partitioned into three groups $A$, $B$, and $C$ where items in $A$ are large for all agents and items in $B$ and $C$ have similar structures in the example with $n=4$.
{    Let $s=2\lceil n/2\rceil+1$. Then $A$ contains $n-1$ items, and each of $B$ and $C$ contains $s$ items.
    The three cost levels $1,7$, and $20$ are replaced by $p=1$, $q=s+2$, and $r=s(q+1)/2$, respectively.}
    The analysis is mostly the same as it is in the case $n=4$.
    See Appendix~\ref{append:tri-n} for details.
\end{remark}

\begin{remark}
    Our counterexample uses only three types of chores. The number of types is also tight: for two types of chores, \citet{aziz2023fair} show that there always exists an EFX allocation. In addition, our counterexample only uses two types of agents, which is also tight: for one type of agents, all agents have the same cost functions, and a leximin++ solution gives an EFX allocation~\citep{plaut2020almost,chen2020fairness}.
{    Thus, the contrast with additive goods is established for the number of agent types: EFX exists for goods with up to three agent types~\citep{hv2025efx}, whereas two types suffice for our chore counterexample.}
\end{remark}

\subsection{Intuitions for the Counterexample}
Below, we describe the intuition behind our counterexample.

Consider $n-1$ large items where everyone has a large cost $T$ ($T=20$ in our example with $n=4$). We need to let $n-1$ agents take the large items such that each of them takes one. Then, a ``gap-filling'' process is required for the remaining agent (who did not receive a large item): we need to allocate those small items to this agent such that \emph{each of those $n-1$ agents} believes the bundle's cost is at least $T$ (so that those $n-1$ agents can continue to take the remaining small items from the pool without violating EFX). However, if we want a bundle that is large enough for \emph{all the $n-1$ agents}, this bundle can be too large for the agent receiving it.

This idea yields a search among the tri-valued instances: one large value $T$, and two smaller values to create heterogeneous views of gap-filling among every combination of $n-1$ agents. This also gives the natural idea of two agent types, where each type contains at least two agents: for each selected gap-filling agent, we can find two other agents with different types (so that they have different views on gap-filling).

Importantly, this idea of constructing EFX non-existence example does not work for goods. In the goods setting, we want the gap-filling agent to receive a set of small items that, in her view, is just enough to exceed $T$, while it is less than $T$ or just enough to exceed $T$ in the view of each of the remaining $n-1$ agents. This is always possible: starting from the entire item set $M$; we iteratively remove some item $g$ if at least one agent believes that the bundle after removing $g$ is still at least $T$; at the point we reach a "minimal bundle" where no item can be removed, this is a gap-filling bundle. In the setting with chores, however, we need a bundle that is larger than $T$ for every agent while it is just enough to exceed $T$ for one agent. This can be impossible. In particular, for the above-mentioned iterative-item-removal process, we may reach a bundle where everyone believes some extra items can still be removed, but agents have different opinions on which item to be removed next.

\section{{Bi-Valued Instances with No EFX+PO Allocation}}
From now on, we focus on bi-valued instances.
In this section, we show that EFX is incompatible with PO for at least four agents.
This is a complement to the result of \citet{garg2023new} showing that EFX is compatible with fractional Pareto-optimality for $n\leq 3$.
{To the best of our knowledge, this is the first example showing that EFX and PO are incompatible for strictly positive additive chore costs.}

Recall that $c_i(g)\in\{1,r\}$ for some $r>1$ in bi-valued instances.

\begin{theorem}\label{thm:EFXPO}
    For every $n\geq 4$ and $r>\lceil n/2\rceil+1$, there exists a $(1,r)$-bi-valued instance with $n$ agents where every EFX allocation is not Pareto-optimal.
\end{theorem}

{The instance has a structure similar to the construction in the proof} of Theorem~\ref{thm:tri}.

Let $t_1=\lfloor n/2\rfloor$ and $t_2=\lceil n/2\rceil$.
Agents are partitioned into two groups $T_1$ and $T_2$ with $|T_1|=t_1$ and $|T_2|=t_2$.
There are $m=2n+1$ items which are partitioned into three groups $A,B$, and $C$.
$A$ is the set of large items containing $n-1$ items which have cost $r$ to every agent.
$B$ contains $t_1+1$ items.
Each agent in $T_1$ values each item in $B$ at $1$ and each item in $C$ at $r$.
$C$ contains $t_2+1$ items.
Each agent in $T_2$ values each item in $B$ at $r$ and each item in $C$ at $1$.
Recall that $r$ can be any real number greater than $t_2+1$.
Table~\ref{tab:po} illustrates the instance.

\begin{table}[h]
    \centering
    \caption{A bi-valued instance where every EFX allocation fails to be Pareto-optimal.  Each entry is the cost of every item in the corresponding item class.}
    \begin{tabular}{cccc}
    \hline
         & $A=\{a_1,\ldots,a_{n-1}\}$ & $B=\{b_1,\ldots,b_{t_1+1}\}$ & $C=\{c_1,\ldots,c_{t_2+1}\}$\\
    \hline
       Agents in $T_1$ & $r$ & $1$ & $r$\\
       Agents in $T_2$ & $r$ & $r$ & $1$\\
    \hline
    \end{tabular}
    
    \label{tab:po}
\end{table}

Notice that an EFX allocation exists in this instance.
For each agent $i\in T_1$, she receives a large item from $A$ and a small item from $B$.
For all but one agent in $T_2$, each of them receives a large item from $A$ and a small item from $C$.
The remaining agent in $T_2$ receives one item from $B$ and two items from $C$.
The readers can easily verify this is an EFX allocation.

Let $X=(X_1,\ldots,X_n)$ be an arbitrary EFX allocation.
We will show that $X$ cannot be Pareto-optimal.

In the following proposition, we show that everyone must receive at least one item that is ``large'' for her.
This large item can either be an item in $A$, or an ``across-group'' item (an item in $C$ for some agent in $T_1$, or an item in $B$ for some agent in $T_2$).
Given that there are $n-1$ items in $A$, this proposition implies there is at least one item that is allocated ``across-group''.

\begin{proposition}\label{prop:PO-exist-large}
    For every $X_i$, there exists $g\in X_i$ with $c_i(g)=r$.
\end{proposition}
\begin{proof}
    Let $U_1\subseteq T_1$ be the set of agents such that, for each $i\in U_1$, $X_i\subseteq B$.
    That is, $U_1$ is the set of agents in $T_1$ that only receive small items.
    Similarly, let $U_2\subseteq T_2$ be the set of agents such that, for each $i\in U_2$, $X_i\subseteq C$.
    This proposition says both $U_1$ and $U_2$ are empty sets.
    {Suppose for contradiction that this is not the case.}

    If we have both $U_1\neq\emptyset$ and $U_2\neq\emptyset$, then items in $A$ must be allocated among $N\setminus(U_1\cup U_2)$ (which means $N\setminus(U_1\cup U_2)\neq\emptyset$).
    Moreover, each agent in $N\setminus(U_1\cup U_2)$  must receive at most one item in $A$.
    Suppose not, if an agent in $T_1\setminus U_1$ (in the case $T_1\setminus U_1\neq\emptyset$) {receives more than one item in $A$}, then, after removing one item, the disutility for this agent is at least $r$.
    On the other hand, by our choice of $r$, this agent values the bundle of each agent in $U_1$ strictly less than $r$.
    This is a violation of EFX.
    Similarly, each agent in $T_2\setminus U_2$ (in the case $T_2\setminus U_2\neq\emptyset$) can only receive one item from $A$.
    Since $|N\setminus(U_1\cup U_2)|\leq n-2$ for $U_1\neq\emptyset$ and $U_2\neq\emptyset$ and $|A|=n-1$, this is a contradiction.

    If one of $U_1$ and $U_2$ is nonempty, suppose $U_1\neq\emptyset$ and $U_2=\emptyset$. The case with $U_2\neq\emptyset$ and $U_1=\emptyset$ is analyzed similarly and is omitted.
    We further discuss two sub-cases.

    If $U_1=T_1$, then agents in $T_1$ only receive items from $B$. Since $|B|=t_1+1<2t_1$ for $n\geq 4$, at least one agent in $T_1$ receives at most one item.
    Now consider $T_2$.
    Firstly, since $|A|=n-1>t_2$ for $n\geq 4$, at least one agent in $T_2$ receives two large items from $A$.
    It then follows that everyone in $T_2$ must receive at least one large item from $A\cup B$: supposing not, the agent receiving two items from $A$ will strongly envy the agent receiving no item from $A\cup B$ even if the latter agent receives all items from $C$ (since $r>t_2+1$).
    Secondly, since $|C|>t_2$, at least one agent in $T_2$ must also receive two items from $C$.
    Therefore, we know an agent in $T_2$ will receive one large item from $A\cup B$ and two small items from $C$.
    This agent strongly envies the agent in $T_1$ who receives at most one item.

    If $U_1\subsetneq T_1$, to guarantee EFX, each agent $i\in T_1\setminus U_1$ must receive exactly one large item in $A\cup C$: she receives at least one large item by our definition of $U_1$, and she cannot receive more to avoid strongly envying agents in $U_1$.
    This item is valued either $1$ or $r$ for each agent in $T_2$.
    Since $U_2=\emptyset$, each agent in $T_2$ has received an item in $A\cup B$.
    To avoid strongly envying agents in $T_1\setminus U_1$, each agent in $T_2$ can then receive at most one more item other than the large item in $A\cup B$.
    The total number of items allocated among $T_2$ and $T_1\setminus U_1$ is at most $|T_1\setminus U_1|+2|T_2|\leq t_1-1+2t_2=n-1+t_2$.
    On the other hand, $A\cup C$ are allocated only to agents in $(T_1\setminus U_1)\cup T_2$ and $|A\cup C|=n-1+t_2+1$.
    This leads to a contradiction.
\end{proof}

Next, we show that every agent's disutility is at least $r+1$ and one of them is at least $r+2$.

\begin{proposition}\label{prop:r+1r+2}
    For every agent $i$, we have $c_i(X_i)\geq r+1$. Moreover, there exists an agent $i$ with $c_i(X_i)\geq r+2$.
\end{proposition}
\begin{proof}
    For the first part, if there exists $i$ with $c_i(X_i)<r+1$, then it must be that $X_i$ contains exactly one large item (either an item in $A$ or an ``across-group'' item) by Proposition~\ref{prop:PO-exist-large}.
    By Proposition~\ref{prop:PO-exist-large} again, each remaining agent has already received one large item.
    To guarantee EFX, each remaining agent can at most receive one more item.
    The total number of items allocated is then bounded by $2n-1$.
    However, we have $m=2n+1$ items.
    This is a contradiction.

    For the second part, since $m=2n+1$, at least one agent receives at least $3$ items.
    By Proposition~\ref{prop:PO-exist-large}, at least one of these items costs $r$.
    The total cost is therefore at least $r+2$.
\end{proof}

We now construct an allocation $Y$ that Pareto-dominates $X$. Choose an agent $i$ with $c_i(X_i)\geq r+2$, whose existence follows from Proposition~\ref{prop:r+1r+2}. Suppose first that $i\in T_1$. Give $i$ one chore from $A$ and two chores from $B$, for cost $c_i(Y_i)=r+2$. Give every agent in $T_1\setminus\{i\}$ one chore from $A$ and one from $B$. Give all but one agent in $T_2$ one chore from $A$ and one from $C$, and give the remaining agent the last two chores from $C$. This allocates every chore. All agents other than $i$ incur cost at most $r+1$, and the last agent incurs only $2<r+1$. Proposition~\ref{prop:r+1r+2} therefore implies $c_j(Y_j)\leq c_j(X_j)$ for every $j$, with a strict inequality for the last agent. The construction for $i\in T_2$ is symmetric. Thus, every EFX allocation is Pareto dominated.

\section{EFX Allocations for Four Agents with Bi-Valued Disutilities}
We continue to study bi-valued instances.
In this section, we show that an EFX allocation always exists for four agents.

\begin{theorem}\label{thm:four}
{    For four agents, there exists an EFX allocation in every bi-valued instance. For strictly positive costs, such an allocation can be computed in $O(m)$ time under the unit-cost arithmetic model.}
\end{theorem}

For $r=2$, \citet{lin2025approximately} shows that EFX is compatible with Pareto-optimality for any number of agents.
If we only consider EFX, the case with $r\leq 2$ can be handled by the following simple variant of the round-robin algorithm for any number of agents.
Let $m=an+b$ for $b<n$.
If $b=0$, the standard round-robin algorithm outputs an EFX allocation $X=(X_1,\ldots,X_n)$: for any pair of agents $i,j$ with $i<j$, $i$ will not envy $j$ as $i$ always picks first; we also have $c_j(X_j)-c_j(X_i)\leq r-1$ (if in a round agent $j$ begins to receive a large item whereas agent $i$'s item in this round is small for agent $j$, then every item in the remaining process of this algorithm is large for agent $j$), $j$ will not envy $i$ up to removing any item from $j$'s bundle since $r-1\leq1$.
If $b>0$, we introduce a special ``first reverse round'' of the round-robin algorithm where agents $b,b-1,\ldots,1$ iteratively pick favorite items in the order.
After that, we perform the standard round-robin algorithm in the standard order $1,\ldots,n$.
For each pair of agents $i,j$ with $i<j$, agent $i$ will not envy agent $j$ if $i>b$, and $i$ will not envy $j$ after removing the item picked in the first special round in the case $i\leq b$. To see the latter, if this item is a small item for agent $i$, the EFX condition holds; if this item is a large item, then $X_i$ only contains large items, and the EFX condition also holds.
The reason agent $j$ will not strongly envy agent $i$ is similar as the $b=0$ case: agent $j$ has an advantage to agent $i$ in the first round, and after that the disadvantage is at most $r-1$.

We will assume $r>2$ in the remaining part of this section.
For a chore $g$, let
\[
 S(g)=\{i\in N:c_i(g)=1\}
\]
be the set of agents for whom $g$ is small.  We partition the chores into
\[
M_{01}=\{g:|S(g)|\le 1\},\qquad
M_2=\{g:|S(g)|=2\},\qquad
M_{34}=\{g:|S(g)|\ge 3\}.
\]
{A chore in $M_2$ has type $(i,j)$ if it is small exactly for agents $i$ and $j$; we also call it an $(i,j)$-item.}  Thus, the items in $M_2$ form a multigraph on the vertex set
$N$, where an edge $(i,j)$ is a chore small for precisely agents $i$ and $j$.

{\paragraph{Overview.}
We first prove that every chore in $M_{34}$ can be inserted into an EFX partial allocation while preserving EFX (Sect.~\ref{subsec:34}). It therefore suffices to allocate $M_{01}\cup M_2$. We construct a canonical allocation of $M_{01}$, use selected $M_2$ chores to fill the gaps of agents with smaller prefix bundles, and append an allocation of the remaining $M_2$ chores. Sect.~\ref{subsec:012over} explains how the multigraph construction and its residual guarantees make this composition possible; Sects.~\ref{subsec:01}--\ref{subsec:combine} give the details. Algorithm~\ref{alg:four} summarizes the resulting procedure, and Sect.~\ref{subsec:runtime} gives its running time.}

\subsection{Adding Items in \texorpdfstring{$M_{34}$}{M34}}
\label{subsec:34}
The first step shows that items small for at least three agents can be added
one by one after an EFX allocation has already been found for the remaining
items.

\begin{lemma}[Insertion of an $M_{34}$ item]\label{lem:insertion}
Let $X$ be an EFX allocation of a set of chores $T$.  Let $g\in M_{34}$ be a
chore that is small for at least three agents.  Then there is an EFX
allocation of $T\cup\{g\}$.
\end{lemma}

\begin{proof}
Given an allocation $X$, define its envy graph $G(X)$ by putting an edge
$i\to j$ if $c_i(X_i)>c_i(X_j)$.  Define the \emph{most-envy graph} $H(X)$
by giving each envying agent $i$ only those outgoing edges $i\to j$, where
$X_j$ is a cheapest bundle from $i$'s perspective.  Agents that envy nobody
have no outgoing edge.  Hence $G(X)$ and $H(X)$ have exactly the same sinks.

First eliminate directed cycles in $H(X)$.  Suppose
\[
 i_1\to i_2\to\cdots\to i_k\to i_1
\]
is a directed cycle in $H(X)$.  Rotate bundles along the cycle: agent $i_t$
receives $X_{i_{t+1}}$, with indices modulo $k$.  This preserves EFX.  Indeed,
each agent on the cycle receives a bundle that was cheapest for her before
the rotation.  Therefore her new EFX threshold is at most the cost of her
new bundle, which is at most the cost of every old bundle from her
perspective.  Since the rotation merely permutes old bundles, all comparison
inequalities continue to hold.  Agents outside the cycle keep their own
bundles and only see comparison bundles permuted.  The rotated allocation is
therefore EFX.  Moreover, every agent on the cycle strictly decreases her
own cost, so repeated cycle elimination terminates.  We may therefore assume
that $H(X)$ is acyclic.

\begin{paperparagraph}
For four agents, the cycle-elimination step takes only a constant number of rotations. During this step the four bundles themselves are fixed and only their owners change, so there are at most $4!$ assignments. Every rotation strictly decreases the sum of the agents' own costs, so no assignment can recur.
\end{paperparagraph}

Let $g$ be the new item.  Since $H(X)$ is acyclic, it has at least one sink.
If there is a sink $i$ with $c_i(g)=1$, assign $g$ to $i$.  For every
$h\ne i$, the own bundle of $h$ is unchanged and bundle $i$ only becomes more costly from $h$'s perspective, so no EFX inequality of $h$ can
be violated.  For agent $i$, because $i$ was a sink, we had
\[
 c_i(X_i)\le c_i(X_j)\qquad\forall j.
\]
After receiving the small item $g$, we have
\[
 c_i(X_i\cup\{g\})-\min_{e\in X_i\cup\{g\}}c_i(e)
 \le c_i(X_i),
\]
so all EFX inequalities of $i$ hold.

It remains to consider the case in which no sink views $g$ as small.  Since
$g$ is small for at least three agents, this means that $H(X)$ has a unique
sink $i$, and $i$ is the unique agent for whom $g$ is large.
Tentatively assign $g$ to $i$.  If the resulting allocation is EFX, we are
done.  Otherwise, only agent $i$ can violate EFX, because every other agent's
own bundle is unchanged and $g$ is small for every other agent.  Let $j$ be
one of the cheapest bundles for $i$ after the tentative assignment.

Since $i$ was the unique sink of the old acyclic graph $H(X)$, there is a
path
\[
 j=v_0\to v_1\to\cdots\to v_k=i
\]
in the old graph $H(X)$.  Rotate along the cycle formed by the new edge
$i\to j$ and this old path:
\[
 i\to v_0\to v_1\to\cdots\to v_k=i.
\]
That is, set
\[
 X'_i=X_{v_0},\qquad
 X'_{v_t}=X_{v_{t+1}}\quad (t=0,\ldots,k-2),\qquad
 X'_{v_{k-1}}=X_i\cup\{g\},
\]
and leave all other bundles unchanged.

We verify EFX.  Agent $i$ receives $X_{v_0}=X_j$, which was cheapest for
$i$ after the tentative assignment, so all her EFX inequalities hold.
For every $v_t$ with $t<k-1$, the new bundle $X_{v_{t+1}}$ was cheapest for
$v_t$ in the old allocation.  Thus, her new threshold is at most the cost of
any old bundle.  The only comparison bundle that is not an old bundle is
$X_i\cup\{g\}$; but $g$ costs positively for $v_t$, so this bundle is more costly than $X_i$ from $v_t$'s perspective.
Hence, $v_t$ remains EFX.
Finally, the predecessor $v_{k-1}$ receives $X_i\cup\{g\}$.  Since $g$ is
small for $v_{k-1}$,
\[
 c_{v_{k-1}}(X_i\cup\{g\})-1=c_{v_{k-1}}(X_i).
\]
The bundle $X_i$ was cheapest for $v_{k-1}$ in the old allocation, because
$v_{k-1}\to i$ was an edge of $H(X)$.  Therefore $v_{k-1}$ also satisfies
EFX.
All agents not on the path keep their own bundle, and their comparison
bundles are old bundles except that $X_i$ is replaced by the weakly more
costly bundle $X_i\cup\{g\}$.
Hence their EFX inequalities remain valid.
\end{proof}

Consequently, it is enough to prove the theorem for the items in
$M_{01}\cup M_2$.  Once an EFX allocation for $M_{01}\cup M_2$ has been
constructed, the items in $M_{34}$ can be inserted one by one using the
lemma above.  From now on we therefore assume that all items belong to
$M_{01}\cup M_2$.

\begin{remark}
    This idea also works for general numbers of agents.
    If we have a partial EFX allocation such that each of the remaining unallocated items is valued as large for at most one agent, then we can always allocate the remaining items one by one until we get a complete EFX allocation.
    Therefore, under bi-valued setting, items valued as large by at most one agent are never a real obstacle for EFX allocations.
\end{remark}

\subsection{Overview of the Proof for \texorpdfstring{$M_{01}\cup M_2$}{M01 union M2}}
\label{subsec:012over}
{The proof combines two partial allocations whose fairness certificates must be compatible. First write
\[
 |M_{01}|=4a+b,\qquad b\in\{0,1,2,3\}.
\]
A canonical $M_{01}$ allocation gives each agent either $a$ or $a+1$ chores, prioritizing chores uniquely small for her. We may choose which $b$ agents receive the larger bundles. This choice lets us place subsequent gap-filling chores on suitable short agents (Sect.~\ref{subsec:01}).

\paragraph{The multigraph step.}
After selecting the gap-filling chores, represent the residual $M_2$ chores by a multigraph $R$ on the four agents. Assigning an edge to an endpoint makes it small for its owner. If every set $U$ of agents is incident to at least $|U||R|/4$ edges, the Hall-type conditions permit an endpoint assignment with bundle sizes differing by at most one. Otherwise, choose a set of maximum \emph{deficiency}, $|U||R|/4-\operatorname{inc}_R(U)$. Only singletons and pairs can be deficient, because every edge meets every three-agent set.

For a maximally deficient singleton, give that agent all her incident edges and enough nonincident edges to reach $\lfloor |R|/4\rfloor$ chores; orient the remaining triangle edges to their endpoints. For a maximally deficient pair, the edge type on its complement is abundant. Give cross edges to their endpoints on the deficient side and split the complementary edges to balance the bundles. When no cross edges exist, only two disjoint edge types remain, and explicit integer splits handle their multiplicities. These are the three cases of Lemma~\ref{lem:M2}; the detailed inequalities in Sect.~\ref{subsec:2} certify EFX.

\paragraph{Why the residual properties matter.}
Two EFX allocations do not automatically concatenate: a large-removal certificate in the residual can fail once a small chore is present in the prefix. Lemma~\ref{lem:M2properties} therefore strengthens the multigraph construction. Outside two explicit exceptional configurations, every agent with a small residual chore satisfies $c_i(B_i)-1\leq c_i(B_j)$, and every agent with only large residual chores is envy-free. Thus, an envy-free gap-filled prefix can be combined with the residual using Lemma~\ref{lem:composition}.

The exceptions consist, up to relabeling, of $4q+3$ chores of one edge type, possibly together with one chore of the disjoint type. Exactly one agent needs a large-removal certificate, and we can choose which of the two relevant agents this is. We either keep her final bundle all-large or use an $r-1$ prefix advantage to compensate for the smaller removal threshold. For two disjoint edge types, we can also prescribe an envy-free agent on the less numerous side; this handles prefixes that retain a small-item EFX certificate after gap filling. Sect.~\ref{subsec:combine} selects the prefix and the gap-filling chores according to $b$ and verifies these conditions.}

The following simple composition observation will be used repeatedly.

{\begin{lemma}[Composition]\label{lem:composition}
Let $P=(P_1,\ldots,P_n)$ and $B=(B_1,\ldots,B_n)$ allocate disjoint sets of chores, each having cost at least $1$ to every agent. If
\[
 c_i(P_i)\leq c_i(P_j)\qquad\forall i,j
\]
and
\[
 c_i(B_i)-1\leq c_i(B_j)\qquad\forall i,j,
\]
then $X_i=P_i\cup B_i$ is EFX. More generally, for each agent $i$ with $X_i\neq\emptyset$, the EFX comparison from $i$ to $j$ is equivalent to
\[
 c_i(P_i)-c_i(P_j)
 \leq c_i(B_j)-c_i(B_i)+\min_{g\in X_i}c_i(g).
\]
Agents with $X_i=\emptyset$ satisfy all their EFX comparisons automatically.
\end{lemma}

\begin{proof}
For a nonempty $X_i$, additivity gives
\[
 \tau_i(X)-c_i(X_j)
 =c_i(P_i)-c_i(P_j)+c_i(B_i)-c_i(B_j)
   -\min_{g\in X_i}c_i(g).
\]
This proves the general statement. Under the two sufficient conditions, the first cost difference is at most $0$, the second is at most $1$, and the minimum item cost is at least $1$, so the displayed expression is nonpositive. If $X_i$ is empty, its threshold is $0$ and no minimum over an empty set is needed.
\end{proof}}

\subsection{The \texorpdfstring{$M_{01}$}{M01} Prefix}
\label{subsec:01}
For each agent $i$, let
\[
 n_i=\left|\{g\in M_{01}:S(g)=\{i\}\}\right|
\]
be the number of $M_{01}$ items uniquely small for $i$.  Items in $M_{01}$
with $S(g)=\emptyset$ are all-large items.
Let $|M_{01}|=4a+b$ with $b\in\{0,1,2,3\}$.
As mentioned before, we will make sure each agent receives either $a$ items or $a+1$ items.

Fix quotas $q_i\in\{a,a+1\}$ with $\sum_i q_i=|M_{01}|$. 
We consider allocations of $M_{01}$ with quotas $q=(q_1,q_2,q_3,q_4)$ in which each agent $i$ first receives as many of her uniquely-small items as possible, up to quota $q_i$, and the remaining slots are filled with the remaining items.
We will show that all such allocations satisfy EFX.
It is also important that we have the freedom to choose which agents receive $a$ items and which agents receive $a+1$ items.
This freedom enables us to append the allocation of $M_2$ after the allocation of $M_{01}$.

We begin by defining this type of allocations.

\begin{definition}
    An allocation $P=(P_1,P_2,P_3,P_4)$ of $M_{01}$ is \emph{canonical} if, for each agent $i$, $|P_i|\in\{a,a+1\}$ and $P_i$ contains $\min\{|P_i|, n_i\}$ items that are small for agent $i$.
    In a canonical allocation with $b>0$, agents receiving $a$ items are called \emph{short} agents, and agents receiving $a+1$ items are called \emph{long} agents.
    For $b=0$ and $a>0$, {every agent in a canonical allocation is both long and short.}
    For $b>0$, an allocation $P=(P_1,P_2,P_3,P_4)$ is \emph{super-canonical} if it is canonical and satisfies $c_i(P_i)\leq c_i(P_j)-r$ for every short agent $i$ and every long agent $j$.
\end{definition}

We will need the following lemma, which will be repeatedly used later.

\begin{lemma}\label{lem:canonical}
    We have the following properties regarding canonical allocations.
    \begin{enumerate}[label=(\alph*)]
        \item Every canonical allocation is EFX. For $b=0$, every canonical allocation is envy-free.
        \item For any $b>0$, there always exists a super-canonical allocation.
        \item If $a>0$ and $b=2$, for any partition $N_1,N_2$ of $N$ with $|N_1|=|N_2|=2$, there exist $i\in N_1$ and $j\in N_2$ such that there exists a super-canonical allocation $P$ with $|P_i|=|P_j|=a$.
    \end{enumerate}
\end{lemma}
\begin{proof}
For a canonical allocation $P$, write
\[
    q_i:=|P_i|\in\{a,a+1\}.
\]
By definition, $P_i$ contains exactly $\min\{q_i,n_i\}$ items that
are small for agent $i$.

\paragraph{Proof of (a).}

\begin{paperparagraph}
If $q_i=0$, then $P_i=\emptyset$ and $\tau_i(P)=0$, so every EFX comparison from $i$ holds. In the EFX verification below, we therefore assume $q_i>0$; all displayed formulas involving removal of an item are then well-defined.
\end{paperparagraph}

{First suppose} that $n_i\ge q_i$.  Then all items in $P_i$ are small
for agent $i$, and hence
\[
    c_i(P_i)=q_i
    \qquad\text{and}\qquad
    \tau_i(P_i)=q_i-1.
\]
Every item has cost at least $1$ to agent $i$, so
\[
    c_i(P_j)\ge |P_j|=q_j.
\]
Since $q_i,q_j\in\{a,a+1\}$, we have $q_i-1\le q_j$.  Therefore,
\[
    \tau_i(P_i)=q_i-1\le q_j\le c_i(P_j).
\]

Now suppose that $n_i<q_i$.  Since $P$ is canonical, all $n_i$
items that are small for agent $i$ belong to $P_i$.  Consequently,
every item in $P_j$ (with $j\neq i$) is large for agent $i$.

If $n_i=0$, then all items in $P_i$ are large for agent $i$, and
\[
    \tau_i(P_i)=r(q_i-1)
    \le rq_j
    =c_i(P_j).
\]

If $n_i>0$, then the minimum-cost item in $P_i$ has cost $1$, and
\[
\begin{aligned}
    \tau_i(P_i)
        &=n_i+r(q_i-n_i)-1\\
        &=r(q_i-1)-(r-1)(n_i-1)\\
        &\le r(q_i-1)\le rq_j=c_i(P_j).
\end{aligned}
\]

Thus, $P$ is EFX.
It remains to prove the stronger statement for $b=0$.  In this case,
$q_i=q_j=a$ for every pair of agents.

If $n_i\ge a$, then $P_i$ consists entirely of items that are small
for agent $i$, and therefore,
\[
    c_i(P_i)=a\le c_i(P_j).
\]

If $n_i<a$, then all items that are small for agent $i$ belong to
$P_i$, while every item in $P_j$ is large for agent $i$.  Hence,
\[
    c_i(P_i)=n_i+r(a-n_i)\le ra=c_i(P_j).
\]

Therefore, $P$ is envy-free when $b=0$.

\paragraph{Proof of (b).}
Call an agent $i$ \emph{heavy} if $n_i\ge 2a+1$, and \emph{light} otherwise, i.e., if $n_i\le 2a$.

We first show that there are at most $b$ heavy agents.  Indeed, if
there were at least $b+1$ heavy agents, then
\[
    \sum_{i\in N} n_i\ge (b+1)(2a+1).
\]
For each $b\in\{1,2,3\}$,
\[
    (b+1)(2a+1)>4a+b=|M_{01}|,
\]
contradicting
\[
    \sum_{i\in N}n_i\le |M_{01}|.
\]

Thus, there are at least $4-b$ light agents.

Choose arbitrarily $4-b$ light agents to be short and the remaining $b$ agents
to be long, and construct a canonical allocation with these quotas.
We show that this allocation is super-canonical.

Fix a short agent $i$ and a long agent $j$.  Suppose first that $n_i<a$.  All items that are small for agent $i$
belong to $P_i$, so every item in $P_j$ is large for $i$.  Therefore,
\[
    c_i(P_j)=r(a+1),
\]
whereas
\[
    c_i(P_i)=n_i+r(a-n_i)\le ra=c_i(P_j)-r.
\]

Now suppose that
\[
    a\le n_i\le 2a.
\]
Then $P_i$ consists of $a$ items that are small for agent $i$, and
hence
\[
    c_i(P_i)=a.
\]
After these $a$ items are assigned to $P_i$, at most $$n_i-a\le a$$ items that are small for agent $i$ remain.  Since $P_j$ contains
$a+1$ items, at least one item in $P_j$ must be large for agent $i$.
Consequently,
\[
    c_i(P_j)\ge a+r,
\]
and therefore
\[
    c_i(P_i)=a\le c_i(P_j)-r.
\]

Thus, the allocation is super-canonical.

\paragraph{Proof of (c).}
Let
\[
    N=N_1\mathbin{\dot\cup}N_2,
    \qquad
    |N_1|=|N_2|=2.
\]

Since $b=2$, a canonical allocation has exactly two short agents. {There are at least two light agents when $b=2$.}
 
If each of $N_1$ and $N_2$ contains a light agent, choose a light
agent $i\in N_1$ and a light agent $j\in N_2$ as the two short
agents.  The construction and proof in part~(b) then give a
super-canonical allocation satisfying
\[
    |P_i|=|P_j|=a.
\]

It remains to consider the case in which one of the two sets contains
no light agent.  Without loss of generality, write $N_1=\{i,k\}$ where both $i$ and $k$ are heavy.  Then
\[
    n_i,n_k\ge 2a+1.
\]
Since
\[
    n_i+n_k\le\sum_{h\in N}n_h
    \le |M_{01}|=4a+2,
\]
we must have
\[
    n_i=n_k=2a+1.
\]

Moreover, equality must hold throughout, so the two agents in
$N_2$ have no uniquely-small items, and every item in $M_{01}$ is
uniquely small for either $i$ or $k$.

Write $N_2=\{j,\ell\}$. We choose $i$ and $j$ as the short agents and $k$ and $\ell$ as the
long agents.  Construct the allocation as follows:
\begin{itemize}
    \item[1.] give agent $i$ exactly $a$ items that are small for $i$;
    \item[2.] give agent $k$ exactly $a+1$ items that are small for $k$;
    \item[3.] give agent $\ell$ exactly $a$ of the remaining items that
    are small for $i$, together with one of the remaining items that
    is small for $k$;
    \item[4.] give agent $j$ the last remaining item that is small for
    $i$, together with the remaining $a-1$ items that are small for
    $k$.
\end{itemize}
The last step is well defined because $a>0$.

The resulting bundle sizes are
\[
    |P_i|=|P_j|=a,
    \qquad
    |P_k|=|P_\ell|=a+1.
\]
Moreover, agents $i$ and $k$ receive as many of their own uniquely-small
items as their quotas allow, while agents $j$ and $\ell$ have no
uniquely-small items.  Hence, the allocation is canonical.

We verify the super-canonical inequalities.  Since $n_j=0$, every
item is large for agent $j$.  Therefore,
\[
    c_j(P_j)=ra \le c_j(P_k)-r=c_j(P_\ell)-r=r(a+1)-r = ra.
\]

For agent $i$, all $a$ items in $P_i$ are small, and therefore,
\[
    c_i(P_i)=a.
\]
Every item in $P_k$ is large for agent $i$, so
\[
    c_i(P_k)=r(a+1)\ge a+r.
\]
The bundle $P_\ell$ contains $a$ items that are small for agent $i$
and one item that is large for agent $i$, and hence
\[
    c_i(P_\ell)=a+r.
\]
Consequently,
\[
    c_i(P_i)\le c_i(P_k)-r
    \qquad\text{and}\qquad
    c_i(P_i)=c_i(P_\ell)-r.
\]

Thus, the constructed allocation is super-canonical, with one short
agent in each of $N_1$ and $N_2$.
\end{proof}

\subsection{The \texorpdfstring{$M_2$}{M2} Multigraph Algorithm} 
\label{subsec:2}
We now give a complete allocation procedure for the items in $M_2$.
Each item is small for some agents $i$ and $j$ and
large for the other two agents, and we say that the type of this item is $(i,j)$, or $ij$ for short.
In this section, we consider the case where the instance only has $M_{2}$ items.
We represent the instance by a
multigraph $R$ on $N=\{1,2,3,4\}$, where an item of type $(i,j)$ is an edge $(i,j)$.
We slightly abuse the notation by letting $R$ also represent the edge set, i.e., $R$ exactly corresponds to items in $M_2$.
Let $x_{ij}$ denote the multiplicity of edge $(i,j)$ and let $m_R=|R|$.

For $U\subseteq N$, define
\[
 \operatorname{inc}_R(U)=\left|\{e\in R:e\cap U\neq\emptyset\}\right|
\]
and
\[
 \Delta_R(U)=\frac{|U|m_R}{4}-\operatorname{inc}_R(U).
\]
We call $U$ \emph{deficient} if $\Delta_R(U)>0$.  
Intuitively, we will make the allocation mostly balanced (except for one special case).
If there is no deficient set, we can ``almost evenly'' allocate items in $M_2$ to the four agents such that each agent receives only small items.
If $U$ is deficient, there are not that many items that are small for some agents in $U$, and we need to fill the bundles of agents in $U$ by some large items (that are small for some agents in $N\setminus U$). 

Since every edge is
incident to every three-agent set, no set of size three is deficient.  Hence,
only singleton and pair deficiencies need to be considered.  In particular,
if $U=\{i,j\}$ and $N\setminus U=\{k,\ell\}$, then
\[
 \Delta_R(U)=x_{k\ell}-\frac{m_R}{2}.
\]

We shall repeatedly use the following standard orientation fact.

\begin{lemma}[Balanced orientation]\label{lem:balancedorientation}
Let $G=(V,E)$ be a multigraph and let $(b_i)_{i\in V}$ be nonnegative
integers satisfying $\sum_{i\in V}b_i=|E|$.  There exists an assignment of
every edge to one of its endpoints such that vertex $i$ receives exactly
$b_i$ edges if and only if
\begin{equation}\label{eqn:orientation}
    \sum_{i\in U}b_i\le \operatorname{inc}_G(U)
 \qquad\text{for every }U\subseteq V.
\end{equation}

If we have \(\sum_{i\in V}b_i\le |E|\) instead, and \eqref{eqn:orientation} is satisfied, then one can assign some edges to their endpoints so that vertex \(i\)
receives exactly \(b_i\) edges, leaving the remaining edges unassigned.
\end{lemma}

\begin{proof}
For the first part, necessity is immediate.  For sufficiency, construct a bipartite flow network
with one node for every edge and one node for every vertex.  The source is
connected to each edge-node with capacity $1$, every edge-node is connected
to its two endpoints with capacity $1$, and vertex $i$ is connected to the
sink with capacity $b_i$.  The displayed inequalities are exactly the
Hall--type cut conditions for a flow of value $|E|$.  Integrality of the
network yields an integral flow, and hence the desired orientation.

The second part of the lemma holds for the same arguments above (the sufficiency part).
\end{proof}

\begin{lemma}[EFX Allocation of $M_2$]\label{lem:M2}
There is an EFX allocation of $M_2$.
\end{lemma}

\begin{proof}
\begin{paperparagraph}
If $m_R=0$, the allocation with four empty bundles is EFX. We may therefore assume $m_R>0$.
\end{paperparagraph}

Write $m=m_R$ and let $S\subseteq N$ be a maximum-deficiency set (i.e., with the maximum $\Delta_R(S)$), with ties broken in favor of smaller cardinality.  If $\Delta_R(S)\le 0$, then no set is deficient.  Otherwise, $S$ is either a singleton or a pair.

\paragraph{Case 1: no deficient set.}
Write $m=4q+s$, where $s\in\{0,1,2,3\}$.  We choose integers
$b_i\in\{q,q+1\}$ with $\sum_i b_i=m$ and then orient every edge to one of
its endpoints so that agent $i$ receives exactly $b_i$ items.

When $s\neq 2$, any choice of the $s$ agents with quota $q+1$ satisfies
\[
 \sum_{i\in U}b_i\le
 \left\lceil\frac{|U|m}{4}\right\rceil
 \le \operatorname{inc}_R(U)
 \qquad\text{for every }U\subseteq N.
\]
When $s=2$, the only possible obstruction is a pair $U$ satisfying
$\operatorname{inc}_R(U)=2q+1$.  Such a tight pair corresponds to an edge
type of multiplicity $2q+1$ on the complementary pair.  Since
$m=4q+2$, there are at most two tight pairs.  We therefore choose the two
agents with demand $q+1$ so that they do not form a tight pair.  Lemma~\ref{lem:balancedorientation} then applies.

Every item is assigned to an endpoint and is therefore small for its owner.
Moreover, every bundle has size $q$ or $q+1$.  %
{Thus, for every agent $i$,
\[
 \tau_i(B)=\max\{0,|B_i|-1\}\leq q,
\]}
whereas every other bundle contains at least $q$ items, each of cost at least
$1$ to agent $i$.  Hence, $c_i(B_j)\ge q\ge \tau_i(B)$ for all $i,j$, and the
allocation is EFX.

\paragraph{Case 2: a singleton is maximally deficient.}
Assume, after relabeling, that $S=\{1\}$.  Let
\[
 d_1=x_{12}+x_{13}+x_{14}
\]
be the number of edges incident to agent $1$, and let $q=\lfloor m/4\rfloor$.
Since $\{1\}$ is deficient, $d_1<m/4$, and hence $d_1\le q$.

Assign all $d_1$ incident edges to agent $1$.  These items are small for her.
We shall additionally assign exactly $q-d_1$ edges from the triangle on
$H=\{2,3,4\}$ to agent $1$; these items are large for her.  All remaining
triangle edges will be oriented to their endpoints, so agents $2,3,4$ receive
only small items.
We will first show that this is always possible by giving the explicit construction, and then verify the EFX property.

For $i\in H$, let $\deg_H(i)$ be the degree of $i$ in the triangle.  If
$\{j,k\}=H\setminus\{i\}$, maximality of the deficiency of $\{1\}$ gives
\[
 \frac m4-d_1=\Delta_R(\{1\})
 \ge \Delta_R(\{1,i\})
 =x_{jk}-\frac m2.
\]
Consequently,
\[
 x_{jk}\le \frac{3m}{4}-d_1
 \quad\text{and hence}\quad
 \deg_H(i)=(m-d_1)-x_{jk}\ge \frac m4.
\]
Choose targets $b_i\in\{q,q+1\}$ for $i\in H$ such that
\[
 b_2+b_3+b_4=m-q.
\]
The preceding degree bounds imply $b_i\le \deg_H(i)$.  Moreover, every pair
of vertices in the triangle is incident to all $m-d_1$ triangle edges, and
$m-d_1\ge m-q=b_2+b_3+b_4$.  Therefore, Lemma~\ref{lem:balancedorientation}
assigns exactly $b_i$ triangle edges to each $i\in H$.  The unassigned
triangle edges are given to agent $1$, whose total bundle size is
\[
 d_1+(m-d_1)-(m-q)=q.
\]

We verify EFX.  Each of agents $2,3,4$ receives only small items and at most
$q+1$ items, so her EFX threshold is at most $q$.  Every bundle has size at
least $q$, and hence costs every agent at least $q$.

For agent $1$, every bundle of an agent in $H$ consists entirely of triangle
edges and therefore entirely of items that are large for agent $1$.  Such a
bundle has size at least $q$.  If $d_1>0$, then agent $1$ owns a small item and
\[
 \tau_1(B)
 =q-1+(r-1)(q-d_1)
 \le rq-1
 <rq
 \le c_1(B_j)
 \qquad(j\in H).
\]
If \(d_1=0\) and \(q=0\), then \(B_1=\emptyset\), so agent \(1\)'s
EFX constraints are vacuous.  If \(d_1=0\) and \(q>0\), then all items
in \(B_1\) are large for agent \(1\), and
\[
    \tau_1(B)=r(q-1)\le rq\le c_1(B_j)
    \qquad(j\in H).
\]
Thus, the allocation is EFX.

\paragraph{Case 3: a pair is maximally deficient.}
Assume, after relabeling, that $S=\{1,2\}$.  Set
\[
 a=x_{12},\qquad
 b=x_{13}+x_{14},\qquad
 c=x_{23}+x_{24},\qquad
 f=x_{34}.
\]
Then $m=a+b+c+f$, and
\[
 \Delta_R(\{1,2\})=f-\frac m2>0.
\]
Since $\{1,2\}$ has maximum deficiency,
\[
 f-\frac m2\ge \Delta_R(\{1\})
 =\frac m4-(a+b)
\]
and similarly
\[
 f-\frac m2\ge \Delta_R(\{2\})
 =\frac m4-(a+c).
\]
Equivalently,
\begin{equation}\label{eqn:bcbound}
     b\le \frac m4,
 \qquad
 c\le \frac m4.
\end{equation}
We distinguish whether cross edges are present.

\subparagraph{Case 3.1: $b+c>0$.}
Write $m=4q+s$, with $s\in\{0,1,2,3\}$, and define
\[
 T=\left\lfloor\frac{m+1}{4}\right\rfloor
 =
 \begin{cases}
 q, & s=0,1,2,\\
 q+1, & s=3.
 \end{cases}
\]
{By \eqref{eqn:bcbound} and the integrality of $b,c$}, we have $b,c\le T$.  We will construct the allocation so that agents $1$ and $2$
receive exactly $T$ items each, and each of agents $3$ and $4$ receives either $q$ or $q+1$ items, depending on the value of $s$ (or $m=4q+s$).
Specifically,
\begin{itemize}
    \item if $s=0$, each of the four agents receives $q$ items;
    \item if $s=1$, each of agents $1$ and $2$ receives $q=T$ items; one of agents $3$ and $4$ receives $q$ items, and the other one receives $q+1$ items;
    \item if $s=2$, each of agents $1$ and $2$ receives $q=T$ items, and each of agents $3$ and $4$ receives $q+1$ items;
    \item if $s=3$, each of agents $1$ and $2$ receives $q+1=T$ items; one of agents $3$ and $4$ receives $q$ items, and the other one receives $q+1$ items.
\end{itemize}
In addition to the above requirements, we will assign all edges of types $13$ and $14$ to agent $1$, and all edges of types $23$ and $24$ to agent $2$.  
We then split the $a$ edges of type $12$ and some edges of type $34$ between agents $1$ and $2$.
In the next few paragraphs, we will describe how those $a$ edges and those type $34$ edges are split among agents $1$ and $2$.

Let
\[
 L=2T-a-b-c.
\]
Since $f>m/2$, the number $a+b+c=m-f$ is at most
$2q-1,2q,2q,2q+1$ for $s=0,1,2,3$, respectively.  In all cases this is at
most $2T$, and hence $L\ge 0$.

We choose nonnegative integers $a_1,a_2,\ell_1,\ell_2$ satisfying
\[
 a_1+a_2=a,
 \qquad
 \ell_1+\ell_2=L,
\]
\begin{equation}\label{eqn:eqT}
    b+a_1+\ell_1=T,
 \qquad
 c+a_2+\ell_2=T,
\end{equation}
and
\begin{equation}\label{eqn:ineqL}
    \ell_1\le c+\ell_2,
 \qquad
 \ell_2\le b+\ell_1.
\end{equation}
We will allocate $a_1$ type $12$ edges, $\ell_1$ type $34$ edges, and $b$ types $13$ and $14$ edges to agent $1$; we will allocate $a_2$ type $12$ edges, $\ell_2$ type $34$ edges, and $c$ types $23$ and $24$ edges to agent $2$.
Intuitively, Eqn.~(\ref{eqn:eqT}) guarantees the bundle size constraints, and Eqn.~(\ref{eqn:ineqL}) guarantees the EFX conditions between agents $1$ and $2$.
{Next, we first show} that a feasible solution to the above system of equations exists, and then show that the constructed allocation is EFX.

To show {the feasibility of this system}, it is enough to choose an integer $a_1$ in the interval
\begin{equation}\label{eqn:maxmin}
     \max\left\{0,\,a-T+c,\,
 \left\lceil\frac{a-b}{2}\right\rceil\right\}
 \le a_1\le
 \min\left\{a,\,T-b,\,
 \left\lfloor\frac{a+c}{2}\right\rfloor\right\}.
\end{equation}
The interval is nonempty.  Indeed, we need to verify that {each of the three upper bounds is at least each of the three lower bounds}, i.e., we need to prove nine inequalities. Here we just give those non-trivial ones:
\begin{equation}\label{eqn:maxmin1}
  T-b \ge a-T+c  
\end{equation} holds because $a+b+c \le 2T$ as proved above;

\begin{equation}\label{eqn:maxmin2}
 a-T+c
 \le 
 \left\lfloor\frac{a+c}{2}\right\rfloor
\end{equation}
holds because $a+c \le 2T$ ({by considering the parity of $a+c$ and using the integrality of $a,c,T$});

\begin{equation}\label{eqn:maxmin3}
 \left\lceil\frac{a-b}{2}\right\rceil
 \le 
 \left\lfloor\frac{a+c}{2}\right\rfloor
\end{equation}
holds because $b+c\ge 1$.

After choosing $a_1$, set
\[
 a_2=a-a_1,
 \qquad
 \ell_1=T-b-a_1,
 \qquad
 \ell_2=T-c-a_2=a_1-(a-T+c).
\]
Give agent $i$ exactly $a_i$ edges of type $12$ and $\ell_i$ edges of type
$34$, for $i=1,2$.  The remaining
\[
 f-L=m-2T
\]
edges of type $34$ are split as evenly as possible between agents $3$ and
$4$.  Agents $3$ and $4$ receive only small items, their bundle sizes differ
by at most one, every bundle has size at least $q$, and the bundles of agents
$3$ and $4$ have size at least $T-1$.

We then show that the allocation is EFX.
Agents $3$ and $4$ have EFX thresholds at most $q$, whereas every bundle has
size at least $q$.  Hence, they satisfy EFX.

Consider agent $1$.  The only items in $B_1$ that are large for her are the
$\ell_1$ edges of type $34$.  In $B_2$, exactly $c+\ell_2$ items are large
for agent $1$.  By \eqref{eqn:ineqL},
\[
 \tau_1(B)
 \le T-1+(r-1)\ell_1
 \le T+(r-1)(c+\ell_2)
 =c_1(B_2).
\]
Each of $B_3$ and $B_4$ consists only of type-$34$ items and has size at least
$T-1$.  If $B_1$ contains a small item for agent $1$, then
$\ell_1\le T-1$ and
\[
 \tau_1(B)=T-1+(r-1)\ell_1
 \le r(T-1)
 \le c_1(B_j)
 \qquad(j=3,4).
\]
If $B_1$ contains no small item, then $\ell_1=T$ and
\[
 \tau_1(B)=r(T-1)\le c_1(B_j)
 \qquad(j=3,4).
\]
The argument for agent $2$ is symmetric, using the second inequality in \eqref{eqn:ineqL}.
Thus, the allocation is EFX.

\subparagraph{Case 3.2: $b=c=0$.}
Only edge types $12$ and $34$ remain, and pair deficiency implies $f>a$.
We further distinguish three subcases.

\medskip
\noindent\emph{Subcase 3.2.1: $a=0$.}
All items have type $34$.  Allocate them so that the four bundle sizes differ
by at most one.  Agents $1$ and $2$ regard every item as large, whereas agents
$3$ and $4$ regard every item as small.  Removing one item from a largest
bundle leaves a cost no larger than that of any smallest bundle, so the
allocation is EFX.

\medskip
\noindent\emph{Subcase 3.2.2: $a=1$.}
Write $f=4p+t$, where $t\in\{0,1,2,3\}$.  Give the unique type-$12$ item to
agent $1$.  Give
\[
 \ell_1=p,
 \qquad
 \ell_2=
 \begin{cases}
 p, & t=0,1,\\
 p+1, & t=2,3
 \end{cases}
\]
type-$34$ items to agents $1$ and $2$, respectively.  Split the remaining
type-$34$ items as evenly as possible between agents $3$ and $4$.  Let
$t^-$ and $t^+$ denote the smaller and larger of the two resulting bundle
sizes.  A direct check of the four values of $t$ gives
\begin{equation}\label{eqn:tcheck}
    t^-\ge \ell_1,
 \qquad
 t^-\ge \ell_2-1,
 \qquad
 t^+-1\le \ell_2.
\end{equation}
Agent $1$ has threshold $r\ell_1$ and values $B_2$ at $r\ell_2$; agent $2$
has threshold $r(\ell_2-1)$ and values $B_1$ at $1+r\ell_1$.  The first two
inequalities in \eqref{eqn:tcheck} give their comparisons with agents $3$ and $4$.
Agents $3$ and $4$ receive only small items and have threshold at most
$t^+-1$.  The last inequality in \eqref{eqn:tcheck} gives their comparison with $B_2$, while
\[
 c_3(B_1)=c_4(B_1)=r+\ell_1\ge \ell_2.
\]
Their mutual comparison follows from the even split.  Hence, the allocation is
EFX.

\medskip
\noindent\emph{Subcase 3.2.3: $a\ge 2$.}
Let
\[
 p^- =\left\lfloor\frac a2\right\rfloor,
 \qquad
 p^+ =\left\lceil\frac a2\right\rceil.
\]
Give $p^+$ type-$12$ items to agent $1$ and $p^-$ type-$12$ items to agent
$2$.  Define
\[
 \eta=
 \begin{cases}
 2, & f\text{ is even},\\
 1, & f\text{ is odd}.
 \end{cases}
\]
Choose an integer $L\ge 0$ satisfying
\begin{equation}\label{eqn:intervalL}
    \left\lceil\frac{f-2rp^- -\eta}{4}\right\rceil
 \le L\le
 \left\lfloor\frac{f-2p^+ +2}{4}\right\rfloor.
\end{equation}
Such an integer exists.  If $a$ is even or $f$ is even, the underlying real
interval has length greater than $1$, and its upper endpoint is nonnegative
because $f>a$.  If both $a$ and $f$ are odd, write $a=2p+1$ and $f=2F+1$.
Then $p=p^-=p^+-1\ge 1$, the floor of the upper endpoint equals
$\lfloor(F-p)/2\rfloor$, and the lower endpoint is strictly smaller than
$(F-p)/2-1/2$.  Hence its ceiling is at most the same integer.

Give $L$ type-$34$ items to each of agents $1$ and $2$, and split the
remaining $f-2L$ type-$34$ items as evenly as possible between agents $3$ and
$4$.  Let
\[
 t^- =\left\lfloor\frac{f-2L}{2}\right\rfloor,
 \qquad
 t^+ =\left\lceil\frac{f-2L}{2}\right\rceil.
\]
The upper bound in \eqref{eqn:intervalL} implies
\begin{equation}\label{eqn:t-}
     t^-\ge L+p^+-1,
\end{equation}
whereas the lower bound, together with the definition of $\eta$, implies
\begin{equation}\label{eqn:t+}
     t^+\le L+rp^-+1.
\end{equation}
Agents $1$ and $2$ satisfy EFX with respect to each other because they receive
the same number of type-$34$ items and their numbers of type-$12$ items differ
by at most one.  By \eqref{eqn:t-},
\[
 \tau_1(B)=p^++rL-1\le rt^-\le c_1(B_j)
 \qquad(j=3,4),
\]
and the analogous inequality for agent $2$ is weaker because $p^-\le p^+$.
Agents $3$ and $4$ receive only small items and have threshold at most
$t^+-1$.  Each of $B_1$ and $B_2$ costs them at least $L+rp^-$.  Therefore
\eqref{eqn:t+} yields their EFX comparisons with agents $1$ and $2$, and their mutual
comparison follows from the even split.  This completes the final subcase.

The three main cases exhaust all possibilities, and therefore every multiset
of $M_2$ items admits an EFX allocation.
\end{proof}

{Lemma~\ref{lem:M2} suffices when the entire item set is $M_2$. For the full construction, however, some $M_2$ chores have already been used for gap filling, and the remaining allocation must be appended to the prefix. We therefore need the following stronger properties of the residual allocation.}

\begin{lemma}[Properties of the EFX allocations in Lemma~\ref{lem:M2}]\label{lem:M2properties}
The EFX allocation $B=(B_1,\ldots,B_4)$ ensured by Lemma~\ref{lem:M2} can be chosen to satisfy the
following additional properties.
\begin{enumerate}[label=(\roman*)]
\item Except in the two configurations described in~(ii), every possible envy is certified by removing a small item.  More precisely,
whenever $B_i$ contains an item that is small for agent $i$,
\[
 c_i(B_i)-1\le c_i(B_j)\qquad\forall j\in N,
\]
whereas every agent whose bundle contains only large items is envy-free:
\[
 c_i(B_i)\le c_i(B_j)\qquad\forall j\in N.
\]

\item Up to relabeling, the only configurations in which an agent whose
bundle contains only large items may fail to be envy-free, and hence
must use a large item as the EFX-removal item, are the following.
\begin{enumerate}[label=(\alph*)]
\item All items have type $(3,4)$ and
\[
 |R|=4q+3.
\]
Then agents $1,2$ receive $q$ and $q+1$ items, respectively, in either
order, while agents $3,4$ receive $q+1$ items each.  We are free to choose
which of agents $1,2$ receives $q+1$ items, and this agent is {the unique agent who uses a large item} as the EFX-removal item.

\item The multiset $R$ consists of one item of type $(1,2)$ and
$4q+3$ items of type $(3,4)$.  One of agents $1,2$ receives the unique
type-$(1,2)$ item together with $q$ type-$(3,4)$ items, the other receives
$q+1$ type-$(3,4)$ items, and agents $3,4$ receive $q+1$ type-$(3,4)$
items each.  We are free to choose which of agents $1,2$ receives the unique
type-$(1,2)$ item.  The other agent is the unique agent who must use a large
item as the EFX-removal item.
\end{enumerate}

\item If every edge of $R$ has type $(1,2)$ or $(3,4)$, with
$x_{34}\ge x_{12}\ge 0$, then the allocation can be chosen so that agent $1$
is envy-free:
\[
 c_1(B_1)\le c_1(B_j)\qquad\forall j\in N.
\]
The same statement holds with agent $2$ in place of agent $1$.
\end{enumerate}
\end{lemma}

\begin{proof}
We use the constructions in Lemma~\ref{lem:M2}, with the tie-breaking choices specified
below.

First observe that if $B_i$ contains an item that is small for agent $i$,
then
\[
 \tau_i(B)=c_i(B_i)-1.
\]
Since $B$ is EFX,
\begin{equation}\label{eqn:efxineq}
     c_i(B_i)-1\le c_i(B_j)\qquad\forall j\in N.
\end{equation}
It therefore remains only to examine agents whose bundles contain exclusively
large items.

\paragraph{No deficient set.}
In Case~1 of Lemma~\ref{lem:M2} every item is assigned to one of its endpoints.  Hence,
every nonempty bundle consists entirely of small items for its owner, and
\eqref{eqn:efxineq} applies.  An agent with an empty bundle is trivially
envy-free.

\paragraph{A maximally deficient singleton.}
Suppose that $\{1\}$ is maximally deficient.  Agents $2,3,4$ receive only
small items.  Agent $1$ receives all $d_1$ items incident to her and
$q-d_1$ triangle items.  If $d_1>0$, then $B_1$ contains a small item, so
\eqref{eqn:efxineq} applies.

If $d_1=0$, then $B_1$ consists of $q$ large items.  Every bundle
$B_j$, $j\in\{2,3,4\}$, consists of at least $q$ triangle items, all of
which are large for agent $1$.  Thus
\[
 c_1(B_1)=rq\le c_1(B_j)\qquad(j=2,3,4),
\]
and agent $1$ is envy-free.

\paragraph{A maximally deficient pair with cross edges.}
Suppose that $\{1,2\}$ is maximally deficient and $b+c>0$ (where $b$ and $c$ are as defined in the proof of Lemma~\ref{lem:M2}; the same holds for $c,f,T,\ell_1,\ell_2,a_1,a_2$ in the analysis below).  Agents $3$ and
$4$ receive only small items.  Consider agent $1$.  If $B_1$ contains a
small item, then \eqref{eqn:efxineq} applies.  Otherwise,
\[
 b=a_1=0,\qquad \ell_1=T.
\]
The inequality $\ell_1\le c+\ell_2$ in \eqref{eqn:ineqL}, together with
$c+a_2+\ell_2=T$, implies $a_2=0$, and hence $a=0$.  Thus, all cross edges
are incident to agent $2$.

If $m\not\equiv3\pmod 4$, then the bundles of agents $3$ and $4$ have size
at least $T$.  Since all their items have type $(3,4)$, they are large for
agent $1$, and
\[
 c_1(B_1)=rT\le c_1(B_3),c_1(B_4).
\]
Moreover, $B_2$ contains exactly $T$ items that are large for agent $1$, so
$c_1(B_2)=rT$.  Hence, agent $1$ is envy-free.

If $m=4q+3$, then $T=q+1$, and agents $3,4$ initially receive $q$ and
$q+1$ type-$(3,4)$ items.  Move one type-$(3,4)$ item from $B_1$ to the
smaller of $B_3,B_4$.  Agent $1$ now receives $q$ large items and is
envy-free, since every other bundle costs her at least $rq$.
Each of the agents $3$ and $4$ now receives $q+1$ small items, and the EFX threshold is $q$.
Since every agent receives at least $q$ items, EFX holds for both agents $3$ and $4$.
The envy from agent $2$ to agent $1$ is up to a small item: we have assumed $b+c>0$, and $b=0$ implies $c>0$; this means agent $2$'s bundle contains at least one item that is small for agent $2$; since $B_2$ contains $q+1$ items with at least one item that is small for agent $2$ and the updated $B_1$ contains $q$ items that are large for agent $2$, agent $2$ does not envy agent $1$ up to removing a small item from $B_2$.
The EFX comparisons from agent $2$ to agents $3$ and $4$ are weakly relaxed.  
Thus, the modified allocation remains EFX.  

The argument for agent $2$ is symmetric.

\paragraph{Two disjoint edge types.}
Suppose that $b=c=0$, so only types $(1,2)$ and $(3,4)$ occur.  Recall that
\[
 a=x_{12},\qquad f=x_{34}.
\]

\smallskip
\noindent\emph{The case $a=0$.}
Write $f=4q+s$, where $s\in\{0,1,2,3\}$.  If $s\neq3$, choose a balanced
allocation in which agents $1,2$ receive no more items than agents $3,4$.
Then agents $1,2$, who regard every item as large, are envy-free,
whereas agents $3,4$ own only small items.

If $s=3$, every balanced allocation has bundle-size multiset
\[
 \{q,q+1,q+1,q+1\}.
\]
Choose either agent $1$ or agent $2$ to receive the unique bundle of size
$q$.  The other receives $q+1$ large items, while agents $3,4$ receive
$q+1$ small items each.  The agent among $1,2$ with $q+1$ items is not
envy-free but becomes envy-free after removing one large item.  This
is exactly configuration~(ii)(a).

\smallskip
\noindent\emph{The case $a=1$.}
Write $f=4q+t$, where $t\in\{0,1,2,3\}$.

For $t=0,1,2$, give $q$ type-$(3,4)$ items to each of agents $1,2$, give
the unique type-$(1,2)$ item to either one of them, and split the remaining
type-$(3,4)$ items as evenly as possible between agents $3,4$.  The bundle
sizes of agents $3,4$ are, respectively,
\[
 (q,q),\qquad(q,q+1),\qquad(q+1,q+1)
\]
for $t=0,1,2$.  Hence, the agent among $1,2$ who does not receive the
type-$(1,2)$ item owns $q$ large items and is envy-free.  The other
agent owns a small item, and agents $3,4$ own only small items.  Thus, \eqref{eqn:efxineq}
applies to all remaining agents.

For $t=3$, give the unique type-$(1,2)$ item and $q$ type-$(3,4)$ items
to one of agents $1,2$, give $q+1$ type-$(3,4)$ items to the other, and
give $q+1$ type-$(3,4)$ items to each of agents $3,4$.  The agent who
receives only $q+1$ type-$(3,4)$ items is not envy-free, since she
values the other deficient-side bundle at
\[
 1+rq<r(q+1).
\]
After removing one large item, however, her cost becomes $rq$, which is no
larger than the cost of any other bundle.  This is exactly
configuration~(ii)(b).  The choice of which agent receives the unique
type-$(1,2)$ item is arbitrary.

\smallskip
\noindent\emph{The case $a\ge2$.}
In Subcase~3.2.3 of Lemma~\ref{lem:M2}, both agents $1,2$ receive at least one
type-$(1,2)$ item, while agents $3,4$ receive only type-$(3,4)$ items.
Thus, every agent owns a small item, and \eqref{eqn:efxineq} applies.

This proves parts~(i) and~(ii).

\paragraph{Prescribing an envy-free agent.}
We now prove part~(iii).  By symmetry, it is enough to prescribe agent $1$.
Let $f\ge a\geq0$.

If $f=a$, split the type-$(1,2)$ items between agents $1,2$ so that agent
$1$ receives $\lfloor a/2\rfloor$ and agent $2$ receives
$\lceil a/2\rceil$, and split the type-$(3,4)$ items evenly between agents
$3,4$.  Every item is assigned to a small endpoint and the bundle sizes
differ by at most one, so the allocation is EFX.  Moreover,
\[
 c_1(B_1)=\left\lfloor\frac a2\right\rfloor
 \le c_1(B_2),
\]
while each of $B_3,B_4$ contains at least $\lfloor a/2\rfloor$
type-$(3,4)$ items and therefore costs agent $1$ at least
$r\lfloor a/2\rfloor$.  Thus, agent $1$ is envy-free.

Assume henceforth that $f>a$.

If $a\le 1$, write $f=4q+t$.
For $t=0,1,2$, give agent $1$ exactly $q$ type-$(3,4)$ items and %
{give agent $2$ the type-$(1,2)$ item, if it exists.}  Give agent $2$ exactly $q$
type-$(3,4)$ items and split the remainder evenly between agents $3,4$.
Then
\[
 c_1(B_1)=rq,
\]
whereas each of $B_2,B_3,B_4$ costs agent $1$ at least $r q$.  Thus, agent $1$ is envy-free.

For $t=3$, give {the type-$(1,2)$ item, if it exists,} and $q$ type-$(3,4)$ items
to agent $1$, give $q+1$ type-$(3,4)$ items to agent $2$, and give
$q+1$ type-$(3,4)$ items to each of agents $3,4$.  Then
\[
 c_1(B_1)\le 1+rq\le r(q+1)
 =c_1(B_2)=c_1(B_3)=c_1(B_4),
\]
so agent $1$ is envy-free.  Notice that in this exceptional
configuration the unique large-removal agent is agent $2$, not agent $1$.

If $a=2$, give one type-$(1,2)$ item to each of agents $1,2$ and set
\[
 L=\left\lfloor\frac{f-2}{4}\right\rfloor.
\]
Give $L$ type-$(3,4)$ items to each of agents $1,2$, and split the
remaining $f-2L$ type-$(3,4)$ items evenly between agents $3,4$.  If
$t^-$ denotes the smaller of the latter two bundle sizes, then
$t^-\ge L+1$.  Consequently,
\[
 c_1(B_1)=1+rL=c_1(B_2)
 \le r(L+1)\le c_1(B_3),c_1(B_4),
\]
and agent $1$ is envy-free.  The same inequalities also give the
required EFX comparisons.

Finally, suppose that $a\ge3$.  In Subcase~3.2.3 of Lemma~\ref{lem:M2}, assign
\[
 p^-=\left\lfloor\frac a2\right\rfloor
\]
type-$(1,2)$ items to agent $1$ and
$p^+=\lceil a/2\rceil$ such items to agent $2$, and give both agents $L$
type-$(3,4)$ items, where $L$ satisfies \eqref{eqn:intervalL}.  The EFX proof from Lemma~\ref{lem:M2}
is unchanged.  Moreover,
\[
 c_1(B_1)=p^-+rL\le p^++rL=c_1(B_2).
\]
If $a$ is odd, then $p^+-1=p^-$, and \eqref{eqn:t-} yields
$t^-\ge L+p^-$.  Hence,
\[
 c_1(B_1)=p^-+rL
 \le r(L+p^-)
 \le rt^-
 \le c_1(B_3),c_1(B_4).
\]
If $a$ is even, write $p^-=p^+=p\ge2$.  By \eqref{eqn:t-},
$t^-\ge L+p-1$, and since $r>2$,
\[
 p\le r(p-1).
\]
Therefore,
\[
 c_1(B_1)=p+rL
 \le r(L+p-1)
 \le rt^-
 \le c_1(B_3),c_1(B_4).
\]
Thus, agent $1$ is envy-free.  Interchanging agents $1$ and $2$
proves the symmetric statement.
\end{proof}

\subsection{Concatenating the \texorpdfstring{$M_{01}$}{M01} prefix and the \texorpdfstring{$M_2$}{M2} allocation}
\label{subsec:combine}
We now finish the proof for $M_{01}\cup M_2$.  Write
\[
 |M_{01}|=4a+b,
 \qquad b\in\{0,1,2,3\}.
\]
\begin{paperparagraph}
If $M_2=\emptyset$, any canonical allocation of $M_{01}$ is already EFX by Lemma~\ref{lem:canonical}(a). Hence, in the remaining concatenation cases we assume $M_2\neq\emptyset$.
\end{paperparagraph}

We discuss the four values of $b$.
{In each case and subcase}, we will discuss how the concatenation is done and show that the resultant allocation is EFX.
We will exploit many features discussed earlier: the freedom to choose those agents that receive one more item when allocating $M_{01}$, the freedom to choose envy-free agents suggested in Lemma~\ref{lem:M2properties}, etc.

At a high level, we will try to use some items in $M_2$ for gap-filling such that the allocation of $M_{01}$ in the first phase becomes envy-free after gap-filling.
If the residual of $M_2$ is not of the two configurations in Lemma~\ref{lem:M2properties}(ii), we can allocate this residue so that every envy is up to a small item.
Lemma~\ref{lem:composition} then implies the combined allocation is EFX.
The tricky part is to try to make the residue of $M_2$ after gap-filling avoid the two configurations in Lemma~\ref{lem:M2properties}(ii).
However, there are hard cases where this is not always possible, and we need to handle these cases separately.

If, after the gap-filling, the residual $M_2$ items form one of the two configurations in Lemma~\ref{lem:M2properties}(ii), we will call this configuration \emph{exceptional}.
In an exceptional configuration, one edge $(i,j)$ has multiplicity at least $3$, and, other than those $(i,j)$-items, there can be at most one other item of type $(i',j')$ where $(i',j')$ is disjoint from $(i,j)$.
We will call the two agents $i',j'$ \emph{exceptional agents}.
By Lemma~\ref{lem:M2properties}(ii), for the two exceptional agents $i',j'$, one of them {receives only large items in the residual} of $M_2$ and envies the other by a large item.
We call this agent the \emph{disadvantageous agent}.
Lemma~\ref{lem:M2properties}(ii) tells us we are free to choose the disadvantageous agent from the two exceptional agents.

We first prove the following lemma that will be repeatedly used in this section.

\begin{lemma}\label{lem:easycombine}
    Suppose $P=(P_1,P_2,P_3,P_4)$ is a canonical allocation of $M_{01}$. Suppose some items of \(M_2\) are used for gap-filling, all assigned to short agents, such that \(P\) is extended to an envy-free allocation.
    If the residual of $M_2$ is of an exceptional configuration with exceptional agents $i$ and $j$ and both agents are long agents in the canonical allocation, then there exists an EFX allocation for $M_{01}\cup M_2$.
\end{lemma}
\begin{proof}
    We will carefully choose the disadvantageous agent from $i$ and $j$ so that the final allocation is EFX.
    Consider the allocation of the residual of $M_2$.
    Except for the envy between the disadvantageous agent in $\{i,j\}$ to the other agent in $\{i,j\}$, the remaining envy can only be up to a small item.
    Since the allocation of $M_{01}$ together with the gap-filling items is assumed to be envy-free, by Lemma~\ref{lem:composition}, we only need to check the EFX condition from the disadvantageous agent to the other agent in $\{i,j\}$.

    If one of agents $i,j$ has no small item in her $M_{01}$ bundle $P_i,P_j$,
    choose that agent as the disadvantageous agent.  Her final bundle then contains only large items, so removing a large residual item is compatible with EFX.  
    If both agents have small items in their prefix bundles $P_i,P_j$, then canonicity gives an $r-1$ prefix advantage in both directions (the small item for an agent must be large for the other three agents for items in $M_{01}$):
    \[
     c_i(P_j)\ge c_i(P_i)+r-1
     \qquad\text{and}\qquad
     c_j(P_i)\ge c_j(P_j)+r-1.
    \]
    In this case, we can choose either agent as the disadvantageous agent.
    The $r-1$ prefix advantage exactly compensates for replacing one unit of small-item residual slack by the removal of a large item.  Thus, the combined
    allocation is EFX.
\end{proof}

Below we give a sketch for each of the four cases $b=0,1,2,3$. Complete analyses are available in Appendix~\ref{append:concatenation}. Specifically, in the sketch below, the construction of the allocation is fully described in each (sub)case; the verifications of the EFX properties are routine for most parts, and they are deferred to the full version in Appendix~\ref{append:concatenation}.

\subsubsection{Case \texorpdfstring{$b=0$}{b = 0}}
\label{subsec:b=0}
If $a=0$, we have $M_{01}=\emptyset$, and we can choose any EFX allocation of $M_2$ ({whose existence is guaranteed by Lemma~\ref{lem:M2}) as the final allocation}.
In the case $a>0$, we directly append the allocation of $M_2$ to a canonical allocation of $M_{01}$.
Lemma~\ref{lem:canonical}(a) implies that the prefix allocation $M_{01}$ is envy-free.
If the configuration of $M_2$ is not exceptional, Lemma~\ref{lem:composition} implies the combined allocation is EFX.
If the configuration of $M_2$ is exceptional, by noting that every agent is a long agent in the prefix canonical allocation (for $a>0$), we can apply Lemma~\ref{lem:easycombine} to conclude an EFX allocation of $M_{01}\cup M_2$.

\subsubsection{Case \texorpdfstring{$b=1$}{b = 1}}
We discuss two sub-cases: (sub-case i) there exists an edge in $M_2$ with multiplicity at least three, and (sub-case ii) the multiplicities of all edges in $M_2$ are bounded by two.

\paragraph{Sub-Case i.}
Consider an edge with the maximum multiplicity, and assume this edge is $(1,2)$ without loss of generality. We further split into two sub-sub-cases.

i(a): If there is another edge that intersects $(1,2)$, say, it is $(2,3)$ without loss of generality, we choose a canonical allocation $P$ with $|P_4|=a+1$. 
{For gap filling, give one $(1,2)$-item to agent $1$, one $(1,2)$-item to agent $2$, and one $(2,3)$-item to agent $3$.}
The current allocation is envy-free since the three gap-filling items are all large items for agent $4$ (who receives one more item in the canonical allocation).

Let $M_2^-$ be the remaining items in $M_2$.
If $M_2^-$ is not an exceptional configuration, applying Lemma~\ref{lem:composition} yields an EFX allocation for $M_{01}\cup M_2$.
{Otherwise, at least one $(1,2)$-item remains in $M_2^-$, since the original multiplicity was at least three and only two such items were used for gap filling.}
Thus, the exceptional pair of agents can either be $(1,2)$ or $(3,4)$.

In the former case where the exceptional pair is $(1,2)$, it can only be that $x_{12}=x_{34}=3$ since $x_{12}$ was the maximum, and $M_2$ contains exactly three $(1,2)$-items, one $(2,3)$-item, and three $(3,4)$-items.
In this case, we revert the gap-filling of two $(1,2)$-items and one $(2,3)$-item described above, and {we directly construct} an EFX allocation of $M_{01}\cup M_2$.
Consider a super-canonical allocation $P$ of $M_{01}$ whose existence is guaranteed by Lemma~\ref{lem:canonical}(b).
Exactly one agent is long.
If this is agent $1$, the items in $M_2$ are allocated as follows: agent $2$ takes the three $(1,2)$-items, agent $3$ takes the $(2,3)$-item and one of the $(3,4)$-items, and agent $4$ takes the remaining two $(3,4)$ items.
The case where agent $4$ is long is handled symmetrically.
If agent $2$ is long, {we let agent $1$ take} the three $(1,2)$-items, agent $3$ takes the $(2,3)$-item and one of the $(3,4)$-items, and agent $4$ takes the remaining two $(3,4)$ items.
The case where agent $3$ is long is symmetric.
It is straightforward to verify that the allocations described above are EFX, by using the property of super-canonicity.

In the latter case where $(3,4)$ is the exceptional pair, we can let agent $4$ be the disadvantageous agent.
We only need to verify that agent $4$ does not strongly envy agent $3$ for EFX.
The reasoning is similar to the proof of Lemma~\ref{lem:easycombine}.
If $P_4$ contains only large items in the prefix allocation $P$, then agent $4$ only receives large items at the end, and EFX is satisfied.
If $P_4$ contains a small item, then agent $4$ has an advantage of at least $r-1$ over agent $3$ after gap filling has been performed on top of the canonical allocation, since the gap-filling item for agent $3$ is a $(2,3)$-item which is large for agent $4$.
EFX is again guaranteed.

i(b): If no other edge intersects $(1,2)$, then $M_2$ contains at least three $(1,2)$-items, possibly some $(3,4)$-items, and no other types of items.
Recall our assumptions $x_{12}\geq x_{34}$ and $x_{12}\geq3$.

We first find a canonical allocation $P$ of $M_{01}$ where one of agents $3$ and $4$ is long, and the short agent in $\{3,4\}$ has an advantage of at least $r$ towards the long agent in $\{3,4\}$.
Notice that $P$ needs not to be super-canonical, as we did not require the advantage of agent $1$ or $2$ over the long agent to be at least $r$.
This is always possible.
We can first tentatively let $|P_3|=a+1$.
If agent $4$'s advantage over agent $3$ is not at least $r$, it must be that all items in $P_3$ are small for agent $4$ and large for agent $3$.
Moving one item from $P_3$ to $P_4$ gives a canonical allocation satisfying our description.
We will then fix such an allocation and assume $|P_3|=a+1$ without loss of generality.

We first handle two very special cases with $x_{12}=x_{34}=3$ and $x_{12}=4,x_{34}=3$.
In the first case, we will let agent $1$ receive two $(1,2)$-items, agent $2$ receive one $(1,2)$-item, agent $3$ receive none of $M_2$, and agent $4$ receive all the three $(3,4)$-items.
We can easily check that the allocation is EFX by our property of $P$ and $r>2$.
In the second case, one just needs to ask agent $2$ to receive one more $(1,2)$-item.

For the remaining cases, we use three $(1,2)$-items for gap-filling such that each of agents $1,2,4$ receives exactly one $(1,2)$-item.
The current allocation is envy-free by our property of $P$.
If the remaining part $M_2^-$ is not an exceptional configuration, we are done (Lemma~\ref{lem:composition}).
Otherwise, the exceptional pair can only be $(3,4)$: since $x_{12}\geq x_{34}$ and {we have used three $(1,2)$-items} for gap-filling, and since $x_{34}\equiv 3\mod 4$ in order to make $(1,2)$ the exceptional pair, we must have $x_{12}=x_{34}=3$ or $x_{12}=4,x_{34}=3$ if $(1,2)$ is the exceptional pair. But we have handled these two cases before.
For the $(3,4)$-exceptional pair, we let agent $3$ be the disadvantageous agent.
The reasoning for EFX is similar as before, by noting that the gap-filling item for agent $4$ is an $(1,2)$-item, which is large for the long agent $3$.

\paragraph{Sub-Case ii.}
In this case, there can never be any exceptional configurations.
We only need to make sure we can use some items from $M_2$ for gap-filling so that the resultant allocation is envy-free.
After that, we know the $M_2^-$ cannot be of an exceptional configuration, and we can apply Lemma~\ref{lem:composition} for a final EFX allocation.
We consider two cases: ii(a) and ii(b).

ii(a): First suppose there is an agent, say agent $1$, such that at least three items of $M_2$ are large for her. 
Choose three such items and let $P$ be a canonical allocation with $|P_1|=a+1$.
If $c_i(P_i)\leq c_i(P_1)-r$ for all $i=2,3,4$, we use $P$ as the prefix allocation.  
Use the three chosen $M_2$ items for gap filling.  
Whenever possible, give each gap item to a short agent for whom it is small.  
To be specific, if the three chosen items all have the same edge type, then two
short agents receive small gap items and the third receives a large gap item; otherwise, we can always make sure each of agents $2,3,4$ receives a small gap item.
It is easy to verify that the current allocation is envy-free.

If $c_i(P_i)\leq c_i(P_1)-r$ fails for some $i\in\{2,3,4\}$, this agent $i$ must see only small items in $P_1$.  
Since each item in $M_{01}$ is small for at most one agent, the condition $c_i(P_i)\leq c_i(P_1)-r$ can only fail for one agent.
Suppose it is agent $4$.
If one of the three chosen items is small for agent $4$, give such an item to agent $4$ as her gap item; the prefix is again envy-free.  
If none of the
three chosen items is small for agent $4$, then the three items are large for both agents $1$ and $4$, so they all have the same edge type $(2,3)$.  In this case, make agent $4$, rather than agent $1$, the long
$M_{01}$ agent, i.e., use the prefix allocation $P'$ obtained by moving one item from $P_1$ to $P_4$.
Use the three chosen items to fill the gaps of the other
three agents.
The current allocation is still envy-free.

After the above-mentioned gap-filling process, allocate the remaining $M_2$ items using Lemma~\ref{lem:M2properties}(i). Lemma~\ref{lem:composition} implies the combined allocation is EFX since there cannot be any exceptional configurations.

ii(b): It remains to handle the corner in which no agent views three items of $M_2$ as large.  Since each $M_2$ item is large for exactly two agents, this means $|M_2|\le 4$.  
Recall also the sub-case ii assumption: the multiplicity of each edge is at most $2$.

If $|M_2|\leq 3$, by the multiplicity bound of $2$, we can make sure $M_2$ is allocated such that exactly $|M_2|$ agents receive one small item each.
We only need to choose the prefix allocation $P$ such that an agent not receiving items in $M_2$ is long.
This is then a complete EFX allocation for $M_{01}\cup M_2$.

If $|M_2|=4$, with the multiplicity bound, there can only be two configurations of $M_2$ up to relabeling: 1) two $(1,2)$-items and two $(3,4)$-items, and 2) one item from each of the four types $(1,2),(2,3),(3,4),(1,4)$.
We first find a super-canonical allocation $P$ of $M_{01}$.
{Each configuration is symmetric among its four vertices, so after relabeling we may assume that agent $1$ is long in the super-canonical allocation $P$.}
In the first case, we let agent $2$ {take the two $(1,2)$-items,} agent $3$ take one $(3,4)$-item, and agent $4$ take one $(3,4)$-item.
In the second case, we let agent $2$ take one $(1,2)$-item and one $(2,3)$-item, agent $3$ take one $(3,4)$-item, and agent $4$ take one $(1,4)$-item.
It can be easily verified that the resultant allocation is EFX for both cases.

\subsubsection{Case \texorpdfstring{$b=2$}{b = 2}}
We consider two sub-cases: (sub-case i) $M_2$ contains an edge type with multiplicity at least two, and (sub-case ii) all edge types in $M_2$ have multiplicities at most one.

\paragraph{Sub-Case i.} First suppose some edge type in $M_2$ has multiplicity at least two.  Choose
an edge type of maximum multiplicity, say $(1,2)$, and use a prefix canonical allocation $P$ where agents $1$ and $2$ are short.
Use two $(1,2)$ items for gap filling, one assigned to agent $1$ and one to agent $2$.  The prefix is envy-free: 
envy-freeness for agents $1,2$ is obvious;
for agents $3,4$, each of them views the gap-filling items as large items.

Allocate the residual $M_2$ items, $M_2^-$, using Lemma~\ref{lem:M2properties}.  
If $M_2^-$ is not exceptional, we are done by Lemma~\ref{lem:composition}.
Otherwise, the exceptional pair can only be $(1,2)$ or $(3,4)$: if for contradiction $(2,4)$ is exceptional after we have used two $(1,2)$-items for gap filling, the multiplicity of $(1,3)$ is at least $3$ while the multiplicity of $(1,2)$ is $2$ before the gap-filling; then we should use $(1,3)$-items for gap-filling instead (since we use the edge type with maximum multiplicity for gap-filling).  

{If $(3,4)$ is exceptional, we can} directly invoke Lemma~\ref{lem:easycombine}.
Then the only remaining case is when $(1,2)$ is exceptional. 
Since $(1,2)$ had maximum multiplicity before the two gap items were removed, this can occur only in the case $x_{12}=x_{34}=3$: the residual would consist of one type-$(1,2)$ item and three type-$(3,4)$ items.  
We avoid this residual directly by reverting our above gap-filling strategy and handling it separately.  
If $a>0$ for $M_{01}$, we find a super-canonical allocation by Lemma~\ref{lem:canonical}(c) with $N_1=(1,2)$ and $N_2=(3,4)$.
Suppose without loss of generality that agents $1$ and $3$ are long.
We let agent $2$ take the three $(1,2)$-items and agent $4$ take the three $(3,4)$-items.
The allocation is EFX by super-canonicity and $r>2$.
If $a=0$ (so $|M_{01}|=2$), the only ``hard'' case is when each of agents $1$ and $2$ has one small item in $M_{01}$ (or symmetrically, agents $3$ and $4$), as, in any other case, we can have a super-canonical prefix where exactly one of $1,2$ is long and exactly one of $3,4$ is long, {which reduces to the $a>0$ case}.
For this ``hard'' case, we can use the super-canonical prefix with agents $1$ and $2$ being long.
Then, allocate two of $(1,2)$-items to agent $1$ and {one of them to agent $2$}; allocate two of $(3,4)$-items to agent $3$ and the last $(3,4)$-item to agent $4$.
This allocation is still EFX for $r>2$ since each of agents $1$ and $2$ only takes at most three small items in total, and the items received by $3$ and $4$ are large for $1$ and $2$.

\paragraph{Sub-Case ii.}
Now suppose $M_2$ is a simple graph; equivalently, every edge type corresponds to at most one item.
We choose any super-canonical allocation $P$ guaranteed by Lemma~\ref{lem:canonical}(b).
Suppose $1$ and $2$ are the short agents without loss of generality.
{We consider four subcases.}
\begin{itemize}
    \item $x_{13}+x_{14}>0$ and $x_{23}+x_{24}>0$. In this case, let agent $1$ take $(1,3)$ and $(1,4)$ whenever the items exist, and let agent $2$ take $(2,3)$ and $(2,4)$ {whenever they exist.} 
    
    {If both $(1,2)$ and $(3,4)$ exist}, give $(1,2)$ {to the agent in $\{1,2\}$ who receives fewer items, breaking ties arbitrarily}. Then, for at least one of agents $3$ and $4$, her evaluation to the gap-filling items for both agents $1$ and $2$ is at least $r$. To see this, if agent $1$ takes the item $(1,2)$, then both $3$ and $4$ think this is a large item. For agent $2$, she takes at least one of $(2,3)$ and $(2,4)$, so at least one of agents $4$ and $3$ thinks this is a large item.
    We will allocate $(3,4)$ to this agent.

    If $(1,2)$ exists and $(3,4)$ does not, the item $(1,2)$ can be allocated to any of $1$ and $2$.

    If $(1,2)$ does not exist and $(3,4)$ exists, the item $(3,4)$ is allocated by this rule: if $x_{13}=x_{14}=x_{23}=x_{24}=1$, item $(3,4)$ is allocated to any of agents $3$ and $4$; otherwise, item $(3,4)$ is given to any agent $i\in\{1,2\}$ with $x_{i3}+x_{i4}\leq1$.
    \item $x_{13}+x_{14}>0$ and $x_{23}=x_{24}=0$. {We consider two subcases:}
    \begin{itemize}
        \item $x_{13}+x_{14}=2$, i.e., $x_{13}=x_{14}=1$. In this case, if the $(1,2)$-item exists, we let agent $1$ take $(1,3)$ and $(1,4)$, agent $2$ take $(1,2)$, and any of agents $3,4$ take $(3,4)$ (if exists). If the $(1,2)$-item does not exist but {the $(3,4)$-item exists,} let agent $1$ take $(1,3)$ and $(1,4)$, and let agent $2$ take $(3,4)$. If the $(1,2)$-item does not exist and the $(3,4)$-item does not exist, then agent $1$ takes $(1,3)$ and agent $2$ takes $(1,4)$.
        \item $x_{13}+x_{14}=1$. {Let agent $1$ take} the one incident edge from $\{(1,3),(1,4)\}$. If at most one of $(1,2)$ and $(3,4)$ exists, let agent $2$ take this item. Otherwise, if both $(1,2)$ and $(3,4)$ exist: give $(1,2)$ to agent $2$; for $(3,4)$, give it to agent $3$ if $x_{14}=1$, and give it to agent $4$ if $x_{13}=1$.
    \end{itemize}
    \item $x_{23}+x_{24}>0$ and $x_{13}=x_{14}=0$. This is symmetric to the second case, and is handled in the same way.
    \item $x_{13}=x_{14}=x_{23}=x_{24}=0$. In this case, only \((1,2)\)- and \((3,4)\)-items can exist. Allocate the \((3,4)\)-item, if it exists, to a short agent whose prefix bundle contains no small item, if such an agent exists; otherwise allocate it to either short agent.  If the \((1,2)\)-item also exists, allocate it to the other short agent.

    If the short agent receiving $(3,4)$ has no small prefix item, her final bundle contains only large items, so her EFX threshold is just her old prefix cost; the equal-quota canonical property handles the comparison with the other short agent. If both short agents have small prefix items, then for each short agent $i$, the other short bundle contains an item large for $i$, so
$$c_i(P_j)\geq c_i(P_i)+r-1.$$
Therefore, adding the large $(3,4)$-item to one short agent and the small $(1,2)$-item to the other preserves EFX.
\end{itemize}
It is routine to verify that the constructed allocation in each case above is EFX.

\subsubsection{Case \texorpdfstring{$b=3$}{b = 3}}
We discuss three different sub-cases: i) there exist two intersecting edges in $M_2$, and assume $(1,2)$ and $(1,3)$ exist without loss of generality; ii) no two edges intersect, and there are two types of edges; we will assume $(1,2)$ and $(3,4)$ exist without loss of generality; iii) only one type of items exists, and we assume only $(1,2)$ exists without loss of generality.

\paragraph{Sub-Case i.}
Consider the canonical allocation $P$ where agent $1$ is short.
If $P$ is super-canonical, we use $P$ as the prefix.
Let agent $1$ take one $(1,2)$-item and one $(1,3)$-item for gap filling.
Since both items are small for agent $1$ and $P$ is super-canonical, agent $1$ does not envy any other agent.
Since {each agent $i\in\{2,3,4\}$ views} at least one of the two gap-filling items as large, none of agents $(2,3,4)$ envy the others.
This allocation is thus envy-free.
We then apply Lemma~\ref{lem:M2properties} to find an EFX allocation for the residual $M_2^-$ of $M_2$.
If $M_2^-$ is not exceptional, we are done according to Lemma~\ref{lem:composition}.
If $M_2^-$ has an exceptional pair that avoids agent $1$, we can apply Lemma~\ref{lem:easycombine} to get an EFX allocation.
If the exceptional pair $(1,i)$ contains agent $1$, we can make $i$ disadvantageous.
Agent $i$ will not strongly envy agent $1$ as one of the two items with types $(1,2)$ and $(1,3)$ respectively, used for gap-filling, is large for agent $i$.

If $P$ is not super-canonical and there are items of types
$(1,2)$ and $(1,3)$, choose the weakly more numerous of these two types; say $x_{12}\ge x_{13}$.  
Make agent $2$, rather than agent $1$, the short $M_{01}$ agent, and use one $(1,2)$ item as the gap-filling item for agent $2$.  
This modified prefix (with gap-filling) is envy-free: agent $1$'s long bundle $P_1$ contains only small items (since the old prefix $P$ is not super-canonical), agent $2$'s gap-filling item $(1,2)$ is small for her and large for agents $3$ and $4$.
The residual $M_2^-$ contains at least one $(1,3)$-item.
If $M_2^-$ is not exceptional, we can find an EFX allocation by Lemma~\ref{lem:M2properties}(i) and Lemma~\ref{lem:composition}.
If $M_2^-$ is exceptional, since $x_{12}\ge x_{13}$ before gap-filling and $M_2^-$ contains at least one $(1,3)$-item, it must be that $M_2^-$ contains exactly one $(1,3)$-item and many $(2,4)$-items, making $(1,3)$ the exceptional pair.
If the current bundle for agent $3$ contains only large items for her, then we can make agent $3$ the disadvantageous agent.
Otherwise, agent $3$'s bundle contains a small item for her.
Moreover, agent $1$'s long bundle $P_1$ contains only small items for agent $1$, and these are large items for agent $3$. 
Agent $3$ has an advantage of $r-1$ over agent $1$.
It is again safe to let agent $3$ be the disadvantageous agent.

\paragraph{Sub-Case ii.}
Assume $x_{34}\ge x_{12}$ without loss of generality.  
Consider all super-canonical allocations of $M_{01}$.

If there exists a super-canonical allocation $P$ where either agent $1$ or $2$ is short, we use this as the prefix and assume agent $1$ is short without loss of generality.
Let agent $1$ take one $(1,2)$-item and one $(3,4)$-item for gap-filling.
In the current allocation, agents $2,3,4$ are not envious, and agent $1$ may envy some agent but only up to a small item.
We apply Lemma~\ref{lem:M2properties}(iii) to find an EFX allocation of the residual $M_2^-$ such that agent $1$ does not envy any other agent in $M_2^-$-allocation.
If $M_2^-$ is not exceptional, Lemma~\ref{lem:composition} implies the combined allocation is EFX.
If $M_2^-$ is exceptional, the exceptional pair can only be $(1,2)$ and $2$ is disadvantageous if we want to keep agent $1$ from envying any other agent in $M_2^-$-allocation.
We only need to show that agent $2$ does not strongly envy agent $1$ in the combined allocation.
Indeed, if agent $2$'s prefix bundle $P_2$ contains only large items, agent $2$'s final bundle also does, and EFX condition is met.
If agent $2$'s prefix bundle $P_2$ contains a small item, then agent $2$ has an advantage of at least $r$ over agent $1$ since, among the two gap-filling items $(1,2)$ and $(3,4)$ for agent $1$, one of them is large and one of them is small for agent $2$.
Agent $2$ will not envy agent $1$ in the combined allocation.

If both agents $1$ and $2$ are long in all super-canonical allocations, use any super-canonical allocation $P$ as the prefix, and it must be that agents $1$ and $2$ receive some small items in $P$ (actually, it can only be the case that $a=0$ and each of agents $1$ and $2$ receives a single small item, but we do not need this for the later part).
Let the short agent (either $3$ or $4$) take one $(1,2)$-item for gap-filling.
The current allocation is envy-free.
If the residual $M_2^-$ is not exceptional, we apply Lemma~\ref{lem:M2properties}(i) and Lemma~\ref{lem:composition} to get an EFX allocation.
If $M_2^-$ is exceptional, then the only possible exceptional pair is $(1,2)$, and we apply Lemma~\ref{lem:easycombine} to get an EFX allocation.

\paragraph{Sub-Case iii.}
Consider again all super-canonical allocations of $M_{01}$.
If there exists a super-canonical allocation $P$ where either agent $1$ or $2$ is short, we use this as the prefix and assume agent $1$ is short without loss of generality.
If $|M_2|\leq r$, we let agent $1$ take all items from $M_2$ (only $(1,2)$-type items).
The resultant allocation is EFX since $P$ is super-canonical.
Otherwise, we let agent $1$ take $\lfloor r\rfloor$ items from $M_2$ for gap-filling.
The remaining $(1,2)$-type items are allocated in the round-robin fashion with the order either $1,2,3,4$ or $1,2,4,3$.
If agent $3$'s prefix bundle $P_3$ contains only large items, we use the order $1,2,3,4$; otherwise, we use the order $1,2,4,3$.
Similar analysis in the proof of Lemma~\ref{lem:easycombine} shows that this allocation is EFX.

If both agents $1$ and $2$ are long in {every super-canonical allocation of $M_{01}$}, we choose an arbitrary super-canonical allocation $P$ as the prefix.
We use one $(1,2)$-item to fill the gap, i.e., assign it to the short agent in $\{3,4\}$.
The allocation is envy-free for the same reason as it is in the last paragraph of sub-case ii.
If $M_2^-$ is not exceptional (precisely when $|M_2^-|\not\equiv3\mod 4$), we apply Lemma~\ref{lem:M2properties}(i) and Lemma~\ref{lem:composition} to find an EFX allocation.
Otherwise, the only possible exceptional pair is $(3,4)$.
Let $4$ be the short agent who is filled with an item $(1,2)$.
We can make $3$ disadvantageous.
If $P_3$ contains only large items for agent $3$, {agent $3$ receives only large items,} and she will not strongly envy agent $4$.
Otherwise, agent $3$ has an advantage of $r-1$ towards agent $4$ since the item $(1,2)$ for gap-filling is large for agent $3$.
Thus, we obtain a combined EFX allocation in both cases.

This completes all cases for $M_{01}\cup M_2$.  Therefore an EFX allocation
of $M_{01}\cup M_2$ exists.  By the insertion lemma (Lemma~\ref{lem:insertion}) from Sect.~\ref{subsec:34}, the items
in $M_{34}$ can be added one by one while preserving EFX.  Hence the theorem
follows.

\begin{paperparagraph}
\subsection{Algorithm and Running Time}
\label{subsec:runtime}
Algorithm~\ref{alg:four} summarizes the constructive proof. The calls to the concatenation procedure use the explicit cases in Sect.~\ref{subsec:combine} and Appendix~\ref{append:concatenation}, including the tie-breaking choices of Lemma~\ref{lem:M2properties}.

\begin{algorithm}[H]

\caption{EFX allocation for four agents with costs in $\{1,r\}$}
\label{alg:four}
\begin{algorithmic}[1]
\REQUIRE Four additive cost functions with a common ratio $r>1$.
\ENSURE An EFX allocation of all chores.
\IF{$r\leq2$}
  \STATE Write $m=4a+b$ with $0\leq b<4$.
  \STATE If $b>0$, let agents $b,b-1,\ldots,1$ each pick a cheapest remaining chore.
  \STATE Perform $a$ complete picking rounds in the order $1,2,3,4$.
  \STATE \textbf{return} the resulting allocation.
\ENDIF
\STATE Group chores by $S(g)$ into $M_{01}$, $M_2$, and $M_{34}$.
\STATE Write $|M_{01}|=4a+b$ and store the six $M_2$ edge multiplicities.
\STATE Select the canonical prefix quotas and gap-filling chores using the case for $b$ in Sect.~\ref{subsec:combine}.
\STATE Allocate the residual $M_2$ chores by the deficiency cases of Lemma~\ref{lem:M2}, with the refinements of Lemma~\ref{lem:M2properties}; use the direct allocations in the exceptional concatenation cases.
\STATE Form the EFX allocation $X$ of $M_{01}\cup M_2$.
\FOR{each $g\in M_{34}$}
  \STATE Insert $g$ using Lemma~\ref{lem:insertion}, updating bundle owners and costs.
\ENDFOR
\STATE \textbf{return} $X$.
\end{algorithmic}
\end{algorithm}

\paragraph{Finding the prefix.}
A super-canonical allocation with prescribed short agents can be found using a constant-size flow network. For fixed quotas $q_i\in\{a,a+1\}$, use a flow network with one node for each of the five $M_{01}$ types (all-large or uniquely small for one agent) and one node for each agent. Source-to-type capacities are the type counts, and agent-to-sink demands are $q_i$. The arc from type $i$ to agent $i$ has lower and upper bounds both equal to $\min\{q_i,n_i\}$, enforcing canonicity.

For each short agent $i$ and long agent $j$, also bound the type-$i$ to agent-$j$ arc by $a$. If $n_i<a$, all chores small for $i$ already belong to $P_i$, and her super-canonical comparison holds automatically. If $n_i\geq a$, then $c_i(P_i)=a$, and the comparison $c_i(P_i)\leq c_i(P_j)-r$ is equivalent to $P_j$ containing at least one chore large for $i$; the arc bound enforces exactly this. Thus, a feasible integral flow is equivalent to a super-canonical allocation with the prescribed quotas. There are only $\binom{4}{b}$ quota choices, and each lower-bounded flow network has constant size.

\paragraph{Linear-time implementation.}
Group chores into lists indexed by their at most $16$ small-agent sets. After this grouping, constructing the $M_{01}$ prefix, allocating $M_2$, and selecting the gap-filling chores take $O(m)$ time. The prefix feasibility tests and the flow networks for $M_2$ operate on type multiplicities and have constant size for four agents. Each branch makes a constant number of these calls, and assigning the actual chores from the type lists takes $O(m)$ time.

Maintain the bundle-cost tables and bundle pointers throughout the algorithm. Each $M_{34}$ insertion takes $O(1)$ time: the most-envy graph has four vertices, cycle elimination requires only constantly many bundle rotations, and each rotation exchanges bundle pointers and their maintained cost columns. Adding a chore updates four costs and the corresponding small-chore indicators. The round-robin branch for $r\leq2$ also takes $O(m)$ time using the type lists. Thus, the overall running time is $O(m)$.

\subsection{Personalized Bi-Valued Costs}
\label{subsec:personalized}
The round-robin argument for $r\leq2$ extends to personalized costs $\{1,r_i\}$ for any number of agents when every $r_i\leq2$: each comparison uses only the viewing agent's ratio, and the disadvantage is at most $r_i-1\leq1$.

Beyond this regime, the current proof has two specific barriers. If some ratios exceed $2$ and others do not, round robin no longer gives the required one-unit bound for the former agents. In the alternative multigraph branch, Subcase~3.2.3 of Lemma~\ref{lem:M2} uses the common ratio to guarantee that the interval in \eqref{eqn:intervalL} contains an integer. Replacing that ratio by $\min\{r_3,r_4\}$ can make the interval empty: with $a=3$, $f=5$, and $\min\{r_3,r_4\}=1.1$, its integer lower and upper bounds become $1$ and $0$, respectively.

Even if all $r_i>2$, the concatenation in Appendix~\ref{append:single-type}, case~(a), uses a common-ratio comparison. When agent $1$ is short, choosing $k=\lfloor r_1\rfloor$ small gap-filling chores fits within her prefix advantage, but the comparison from agent $2$ toward agent $1$ needs $r_2-k\leq1$. This need not hold when $r_1$ and $r_2$ differ. These are limitations of the present construction and its certificates; they do not prove nonexistence for personalized bi-valued costs. The general four-agent personalized case remains open.
\end{paperparagraph}

\section{Open Problems}
We have shown that there exist instances for every $n\geq4$ where no EFX allocation exists.
On the other hand, an EFX allocation always exists for $n=2$ agents: the simple divide-and-choose algorithm can compute it.
It is open whether an EFX allocation always exists for three agents.

Previous work and our results also show that an EFX allocation always exists for bi-valued instances with up to four agents (although EFX is not compatible with PO for four agents).
Another natural future direction is to check if an EFX allocation always exists for more than four agents.

\citet{kobayashi2025efx} showed that an EFX allocation always exists for three agents with personalized bi-valued chores.
A natural question is whether this existence result holds for four agents.
For personalized bi-valued chores, the round-robin guarantee above applies when every normalized ratio is at most $2$. Whether EFX always exists for four personalized bi-valued agents with arbitrary ratios is another open question; Sect.~\ref{subsec:personalized} identifies the obstacles in our present proof.

The existence of EFX allocations for goods is also a prominent open problem.
Recently, \citet{akrami2026counterexample} proved that EFX allocation may not exist for submodular valuations, and more explicit counterexamples were provided by
\citet{mackenzie2026counterexamples}.
The existence of EFX allocations for additive valuations remains open.
We only know EFX allocations always exist for up to three agents~\citep{chaudhury2024efx,akrami2025efx} or three types of agents~\citep{hv2025efx}.

In contrast to the setting with chores, for the setting with goods, EFX allocations are known to exist for bi-valued valuations and EFX is compatible with fractional PO in the bi-valued setting~\citep{amanatidis2021maximum,bu2024best}.
The existence of EFX allocations extends to the setting with personalized bi-valued valuations~\citep{jin2025pareto,byrka2026probing}, but its compatibility with PO is unknown.

\section*{Acknowledgments}
The research of Biaoshuai Tao is supported by the National Natural Science Foundation of China (No.\@ 62472271).

The three main theorems and their supporting lemmas and propositions in this paper are formalized in Lean using the AppliedModelingLib Project~\cite{garg2026econcslib} (\protect\url{https://reservoir.lean-lang.org/@nikhgarg/AppliedModelingLib}) and the EconCSLib~\cite{bei2026econcslib} (\protect\url{https://github.com/gametheoryinlean/EconCSLib}). The
formalization is available at \protect\url{https://github.com/nikhgarg/AppliedModelingLib/tree/main/papers/HT26EFXChores}. The authors greatly thank Nikhil Garg and Xiaohui Bei for this.

\bibliographystyle{plainnat}
\bibliography{reference}

\begin{thebibliography}{62}
\providecommand{\natexlab}[1]{#1}
\providecommand{\url}[1]{\texttt{#1}}
\expandafter\ifx\csname urlstyle\endcsname\relax
  \providecommand{\doi}[1]{doi: #1}\else
  \providecommand{\doi}{doi: \begingroup \urlstyle{rm}\Url}\fi

\bibitem[Afshinmehr et~al.(2025)Afshinmehr, Danaei, Kazemi, Mehlhorn, and Rathi]{afshinmehr2025bipartite}
Mahyar Afshinmehr, Alireza Danaei, Mehrafarin Kazemi, Kurt Mehlhorn, and Nidhi Rathi.
\newblock {EFX} allocations and orientations on bipartite multi-graphs: A complete picture.
\newblock In \emph{Proceedings of the 24th International Conference on Autonomous Agents and Multiagent Systems}, pages 32--40, 2025.

\bibitem[Afshinmehr et~al.(2026)Afshinmehr, Ashuri, Mahmoudkhan, Mehlhorn, and Shahrezaei]{afshinmehr2026multigraphs}
Mahyar Afshinmehr, Arash Ashuri, Pouria Mahmoudkhan, Kurt Mehlhorn, and Amir~Mohammad Shahrezaei.
\newblock {EFX} allocations exist on multi-graphs.
\newblock \emph{arXiv preprint arXiv:2606.18665}, 2026.

\bibitem[Akrami et~al.(2025)Akrami, Alon, Chaudhury, Garg, Mehlhorn, and Mehta]{akrami2025efx}
Hannaneh Akrami, Noga Alon, Bhaskar~Ray Chaudhury, Jugal Garg, Kurt Mehlhorn, and Ruta Mehta.
\newblock {EFX}: a simpler approach and an (almost) optimal guarantee via rainbow cycle number.
\newblock \emph{Operations Research}, 73\penalty0 (2):\penalty0 738--751, 2025.

\bibitem[Akrami et~al.(2026)Akrami, Mayorov, Mehlhorn, Srinivas, and Weidenbach]{akrami2026counterexample}
Hannaneh Akrami, Alexander Mayorov, Kurt Mehlhorn, Shreyas Srinivas, and Christoph Weidenbach.
\newblock A counterexample to {EFX} $n\geq 3$ agents, $m\geq n+ 5$ items, submodular valuations via {SAT}-solving.
\newblock \emph{arXiv preprint arXiv:2604.18216}, 2026.

\bibitem[Amanatidis et~al.(2020)Amanatidis, Markakis, and Ntokos]{amanatidis2020multiple}
Georgios Amanatidis, Evangelos Markakis, and Apostolos Ntokos.
\newblock Multiple birds with one stone: Beating 1/2 for {EFX} and {GMMS} via envy cycle elimination.
\newblock \emph{Theoretical Computer Science}, 841:\penalty0 94--109, 2020.

\bibitem[Amanatidis et~al.(2021)Amanatidis, Birmpas, Filos-Ratsikas, Hollender, and Voudouris]{amanatidis2021maximum}
Georgios Amanatidis, Georgios Birmpas, Aris Filos-Ratsikas, Alexandros Hollender, and Alexandros~A Voudouris.
\newblock Maximum nash welfare and other stories about {EFX}.
\newblock \emph{Theoretical Computer Science}, 863:\penalty0 69--85, 2021.

\bibitem[Amanatidis et~al.(2023)Amanatidis, Aziz, Birmpas, Filos-Ratsikas, Li, Moulin, Voudouris, and Wu]{amanatidis2023fair}
Georgios Amanatidis, Haris Aziz, Georgios Birmpas, Aris Filos-Ratsikas, Bo~Li, Herv{\'e} Moulin, Alexandros~A. Voudouris, and Xiaowei Wu.
\newblock Fair division of indivisible goods: Recent progress and open questions.
\newblock \emph{Artificial Intelligence}, 322:\penalty0 103965, 2023.
\newblock \doi{10.1016/j.artint.2023.103965}.

\bibitem[Amanatidis et~al.(2024)Amanatidis, Filos-Ratsikas, and Sgouritsa]{amanatidis2024pushing}
Georgios Amanatidis, Aris Filos-Ratsikas, and Alkmini Sgouritsa.
\newblock Pushing the frontier on approximate {EFX} allocations.
\newblock In \emph{Proceedings of the 25th ACM Conference on Economics and Computation}, pages 1268--1286, 2024.

\bibitem[Aziz et~al.(2023)Aziz, Lindsay, Ritossa, and Suzuki]{aziz2023fair}
Haris Aziz, Jeremy Lindsay, Angus Ritossa, and Mashbat Suzuki.
\newblock Fair allocation of two types of chores.
\newblock In \emph{Proceedings of the 2023 International Conference on Autonomous Agents and Multiagent Systems}, pages 143--151, 2023.

\bibitem[Aziz et~al.(2025)Aziz, Lu, Mackenzie, and Suzuki]{aziz2025fair}
Haris Aziz, Xinhang Lu, Simon Mackenzie, and Mashbat Suzuki.
\newblock Fair division with indivisible goods, chores, and cake.
\newblock \emph{arXiv preprint arXiv:2511.04891}, 2025.

\bibitem[Barman and Suzuki(2026)]{barman2026compatibility}
Siddharth Barman and Mashbat Suzuki.
\newblock Compatibility of fairness and nash welfare under subadditive valuations.
\newblock In \emph{Proceedings of the 2026 Annual ACM-SIAM Symposium on Discrete Algorithms (SODA)}, pages 1724--1746. SIAM, 2026.

\bibitem[Barman and Verma(2025)]{barman2025introspectively}
Siddharth Barman and Paritosh Verma.
\newblock Introspectively envy-free and efficient allocation of indivisible mixed manna.
\newblock \emph{arXiv preprint arXiv:2509.18673}, 2025.

\bibitem[Barman et~al.(2018)Barman, Krishnamurthy, and Vaish]{barman2018finding}
Siddharth Barman, Sanath~Kumar Krishnamurthy, and Rohit Vaish.
\newblock Finding fair and efficient allocations.
\newblock In \emph{Proceedings of the 2018 ACM Conference on Economics and Computation}, pages 557--574, 2018.

\bibitem[Barman et~al.(2020)Barman, Bhaskar, Krishna, and Sundaram]{barman2020tight}
Siddharth Barman, Umang Bhaskar, Anand Krishna, and Ranjani~G Sundaram.
\newblock Tight approximation algorithms for p-mean welfare under subadditive valuations.
\newblock In \emph{28th Annual European Symposium on Algorithms (ESA 2020)}, pages 11:1--11:17. Schloss Dagstuhl--Leibniz-Zentrum f{\"u}r Informatik, 2020.

\bibitem[Barman et~al.(2023)Barman, Narayan, and Verma]{barman2023fair}
Siddharth Barman, Vishnu Narayan, and Paritosh Verma.
\newblock Fair chore division under binary supermodular costs.
\newblock In \emph{Proceedings of the 2023 International Conference on Autonomous Agents and Multiagent Systems}, pages 2863--2865, 2023.

\bibitem[Barman et~al.(2025)Barman, Vishwa~Prakash, Sethia, and Suzuki]{barman2025fair}
Siddharth Barman, HV~Vishwa~Prakash, Aditi Sethia, and Mashbat Suzuki.
\newblock Fair and efficient allocation of indivisible mixed manna.
\newblock In \emph{International Conference on Web and Internet Economics}, pages 467--483. Springer, 2025.

\bibitem[Bei et~al.(2021)Bei, Li, Liu, Liu, and Lu]{bei2021fair}
Xiaohui Bei, Zihao Li, Jinyan Liu, Shengxin Liu, and Xinhang Lu.
\newblock Fair division of mixed divisible and indivisible goods.
\newblock \emph{Artificial Intelligence}, 293:\penalty0 103436, 2021.
\newblock \doi{10.1016/j.artint.2020.103436}.

\bibitem[Bei et~al.(2026)Bei, Ma, Jing, Fu, and Tang]{bei2026econcslib}
Xiaohui Bei, Jiajun Ma, Zhan Jing, Hongfei Fu, and Zhihao~Gavin Tang.
\newblock {EconCSLib}: A {Lean} library for computational economics and {AI}-assisted research.
\newblock \emph{arXiv preprint arXiv:2606.16144}, 2026.

\bibitem[Berger et~al.(2022)Berger, Cohen, Feldman, and Fiat]{berger2022almost}
Ben Berger, Avi Cohen, Michal Feldman, and Amos Fiat.
\newblock Almost full {EFX} exists for four agents.
\newblock In \emph{Proceedings of the AAAI Conference on Artificial Intelligence}, volume~36, pages 4826--4833, 2022.

\bibitem[Bu et~al.(2023)Bu, Song, and Yu]{bu2023efx}
Xiaolin Bu, Jiaxin Song, and Ziqi Yu.
\newblock {EFX} allocations exist for binary valuations.
\newblock In \emph{International Workshop on Frontiers in Algorithmics}, pages 252--262. Springer, 2023.

\bibitem[Bu et~al.(2024)Bu, Li, Liu, Lu, and Tao]{bu2024best}
Xiaolin Bu, Zihao Li, Shengxin Liu, Xinhang Lu, and Biaoshuai Tao.
\newblock Best-of-both-worlds fair allocation of indivisible and mixed goods.
\newblock In \emph{International Conference on Web and Internet Economics}, pages 277--294. Springer, 2024.

\bibitem[Budish(2011)]{budish2011combinatorial}
Eric Budish.
\newblock The combinatorial assignment problem: Approximate competitive equilibrium from equal incomes.
\newblock \emph{Journal of Political Economy}, 119\penalty0 (6):\penalty0 1061--1103, 2011.

\bibitem[Byrka et~al.(2026)Byrka, Malinka, and Ponitka]{byrka2026probing}
Jaroslaw Byrka, Franciszek Malinka, and Tomasz Ponitka.
\newblock Probing {EFX} via {PMMS}:(non-) existence results in discrete fair division.
\newblock In \emph{Proceedings of the AAAI Conference on Artificial Intelligence}, volume~40, pages 16735--16742, 2026.

\bibitem[Caragiannis et~al.(2019{\natexlab{a}})Caragiannis, Gravin, and Huang]{caragiannis2019envy}
Ioannis Caragiannis, Nick Gravin, and Xin Huang.
\newblock Envy-freeness up to any item with high nash welfare: The virtue of donating items.
\newblock In \emph{Proceedings of the 2019 ACM Conference on Economics and Computation}, pages 527--545, 2019{\natexlab{a}}.

\bibitem[Caragiannis et~al.(2019{\natexlab{b}})Caragiannis, Kurokawa, Moulin, Procaccia, Shah, and Wang]{caragiannis2019unreasonable}
Ioannis Caragiannis, David Kurokawa, Herv{\'e} Moulin, Ariel~D Procaccia, Nisarg Shah, and Junxing Wang.
\newblock The unreasonable fairness of maximum nash welfare.
\newblock \emph{ACM Transactions on Economics and Computation (TEAC)}, 7\penalty0 (3):\penalty0 1--32, 2019{\natexlab{b}}.

\bibitem[Chaudhury et~al.(2021)Chaudhury, Kavitha, Mehlhorn, and Sgouritsa]{chaudhury2021little}
Bhaskar~Ray Chaudhury, Telikepalli Kavitha, Kurt Mehlhorn, and Alkmini Sgouritsa.
\newblock A little charity guarantees almost envy-freeness.
\newblock \emph{SIAM Journal on Computing}, 50\penalty0 (4):\penalty0 1336--1358, 2021.

\bibitem[Chaudhury et~al.(2024)Chaudhury, Garg, and Mehlhorn]{chaudhury2024efx}
Bhaskar~Ray Chaudhury, Jugal Garg, and Kurt Mehlhorn.
\newblock {EFX} exists for three agents.
\newblock \emph{Journal of the ACM}, 71\penalty0 (1):\penalty0 1--27, 2024.

\bibitem[Chen and Liu(2020)]{chen2020fairness}
Xingyu Chen and Zijie Liu.
\newblock The fairness of leximin in allocation of indivisible chores.
\newblock \emph{arXiv preprint arXiv:2005.04864}, 2020.

\bibitem[Christodoulou et~al.(2023)Christodoulou, Fiat, Koutsoupias, and Sgouritsa]{christodoulou2023graphs}
George Christodoulou, Amos Fiat, Elias Koutsoupias, and Alkmini Sgouritsa.
\newblock Fair allocation in graphs.
\newblock In \emph{Proceedings of the 24th ACM Conference on Economics and Computation}, pages 473--488, 2023.

\bibitem[Christoforidis(2026)]{christoforidis2026note}
Vasilis Christoforidis.
\newblock A note on {EFX} inapproximability for chores.
\newblock \emph{arXiv preprint arXiv:2605.21448}, 2026.

\bibitem[Christoforidis and Santorinaios(2024)]{christoforidis2024pursuit}
Vasilis Christoforidis and Christodoulos Santorinaios.
\newblock On the pursuit of {EFX} for chores: non-existence and approximations.
\newblock In \emph{Proceedings of the Thirty-Third International Joint Conference on Artificial Intelligence}, pages 2713--2721, 2024.

\bibitem[Ebadian et~al.(2022)Ebadian, Peters, and Shah]{ebadian2022fairly}
Soroush Ebadian, Dominik Peters, and Nisarg Shah.
\newblock How to fairly allocate easy and difficult chores.
\newblock In \emph{Proceedings of the 21st International Conference on Autonomous Agents and Multiagent Systems}, pages 372--380, 2022.

\bibitem[Eckart et~al.(2024)Eckart, Psomas, and Verma]{eckart2024fairness}
Owen Eckart, Alexandros Psomas, and Paritosh Verma.
\newblock On the fairness of normalized p-means for allocating goods and chores.
\newblock In \emph{Proceedings of the 25th ACM Conference on Economics and Computation}, pages 1267--1267, 2024.

\bibitem[Feldman et~al.(2024)Feldman, Mauras, and Ponitka]{feldman2024optimal}
Michal Feldman, Simon Mauras, and Tomasz Ponitka.
\newblock On optimal tradeoffs between {EFX} and nash welfare.
\newblock In \emph{Proceedings of the AAAI Conference on Artificial Intelligence}, volume~38, pages 9688--9695, 2024.

\bibitem[Gafni et~al.(2023)Gafni, Huang, Lavi, and Talgam-Cohen]{gafni2023unified}
Yotam Gafni, Xin Huang, Ron Lavi, and Inbal Talgam-Cohen.
\newblock Unified fair allocation of goods and chores via copies.
\newblock \emph{ACM Transactions on Economics and Computation}, 11\penalty0 (3-4):\penalty0 1--27, 2023.

\bibitem[Garg and Murhekar(2026)]{garg2026existence}
Jugal Garg and Aniket Murhekar.
\newblock Existence of 2-{EFX} allocations of chores.
\newblock In \emph{Proceedings of the AAAI Conference on Artificial Intelligence}, volume~40, pages 16922--16929, 2026.

\bibitem[Garg et~al.(2022)Garg, Murhekar, and Qin]{garg2022fair}
Jugal Garg, Aniket Murhekar, and John Qin.
\newblock Fair and efficient allocations of chores under bivalued preferences.
\newblock In \emph{Proceedings of the AAAI Conference on Artificial Intelligence}, volume~36, pages 5043--5050, 2022.

\bibitem[Garg et~al.(2023)Garg, Murhekar, and Qin]{garg2023new}
Jugal Garg, Aniket Murhekar, and John Qin.
\newblock New algorithms for the fair and efficient allocation of indivisible chores.
\newblock In \emph{32nd International Joint Conference on Artificial Intelligence, IJCAI 2023}, pages 2710--2718. International Joint Conferences on Artificial Intelligence, 2023.

\bibitem[Garg et~al.(2025)Garg, Murhekar, and Qin]{garg2025constant}
Jugal Garg, Aniket Murhekar, and John Qin.
\newblock Constant-factor {EFX} exists for chores.
\newblock In \emph{Proceedings of the 57th Annual ACM Symposium on Theory of Computing}, pages 1580--1589, 2025.

\bibitem[Garg(2026)]{garg2026econcslib}
Nikhil Garg.
\newblock {EconCSLib}: {AI}-assisted {Lean} formalization for economics \& computation research.
\newblock \emph{arXiv preprint arXiv:2606.13306}, 2026.

\bibitem[Hosseini et~al.(2023)Hosseini, Sikdar, Vaish, and Xia]{hosseini2023fairly}
Hadi Hosseini, Sujoy Sikdar, Rohit Vaish, and Lirong Xia.
\newblock Fairly dividing mixtures of goods and chores under lexicographic preferences.
\newblock In \emph{Proceedings of the International Joint Conference on Autonomous Agents and Multiagent Systems, AAMAS}, volume 2023, pages 152--160. International Foundation for Autonomous Agents and Multiagent Systems (IFAAMAS), 2023.

\bibitem[Hsu(2025)]{hsu2025multigraphs}
Kevin Hsu.
\newblock {EFX} orientations of multigraphs.
\newblock In \emph{Proceedings of the 28th European Conference on Artificial Intelligence}, pages 3783--3790. IOS Press, 2025.

\bibitem[Hsu and King(2025)]{hsuking2025chores}
Kevin Hsu and Valerie King.
\newblock Polynomial-time algorithms for fair orientations of chores.
\newblock In \emph{Proceedings of the 28th European Conference on Artificial Intelligence}, pages 3511--3518. IOS Press, 2025.

\bibitem[Hv et~al.(2025)Hv, Ghosal, Nimbhorkar, and Varma]{hv2025efx}
Vishwa~Prakash Hv, Pratik Ghosal, Prajakta Nimbhorkar, and Nithin Varma.
\newblock {EFX} exists for three types of agents.
\newblock In \emph{Proceedings of the 26th ACM Conference on Economics and Computation}, pages 101--128, 2025.

\bibitem[Jin and Tao(2025)]{jin2025pareto}
Jiarong Jin and Biaoshuai Tao.
\newblock On pareto-optimal and fair allocations with personalized bi-valued utilities.
\newblock In \emph{International Conference on Web and Internet Economics}, pages 521--537. Springer, 2025.

\bibitem[Kobayashi et~al.(2025)Kobayashi, Mahara, and Sakamoto]{kobayashi2025efx}
Yusuke Kobayashi, Ryoga Mahara, and Souta Sakamoto.
\newblock {EFX} allocations for indivisible chores: Matching-based approach.
\newblock \emph{Theoretical Computer Science}, 1026:\penalty0 115010, 2025.

\bibitem[Li et~al.(2022)Li, Li, and Wu]{li2022almost}
Bo~Li, Yingkai Li, and Xiaowei Wu.
\newblock Almost (weighted) proportional allocations for indivisible chores.
\newblock In \emph{Proceedings of the ACM Web Conference 2022}, pages 122--131, 2022.

\bibitem[Li et~al.(2025)Li, Tao, Wang, Wu, Yang, and Zhou]{li2025truthfully}
Bo~Li, Biaoshuai Tao, Fangxiao Wang, Xiaowei Wu, Mingwei Yang, and Shengwei Zhou.
\newblock When is truthfully allocating chores no harder than goods?
\newblock In \emph{International Symposium on Algorithmic Game Theory}, pages 247--264. Springer, 2025.

\bibitem[Lin et~al.(2025)Lin, Wu, and Zhou]{lin2025approximately}
Zehan Lin, Xiaowei Wu, and Shengwei Zhou.
\newblock Approximately {EFX} and f{PO} allocations for bivalued chores.
\newblock In \emph{Proceedings of the Thirty-Fourth International Joint Conference on Artificial Intelligence}, pages 3952--3960, 2025.

\bibitem[Lin et~al.(2026)Lin, Wu, and Zhou]{lin2026allocating}
Zehan Lin, Xiaowei Wu, and Shengwei Zhou.
\newblock Allocating chores with restricted additive costs: Achieving {EFX}, {MMS}, and efficiency simultaneously.
\newblock In \emph{Proceedings of the ACM Web Conference 2026}, pages 146--156, 2026.

\bibitem[Lipton et~al.(2004)Lipton, Markakis, Mossel, and Saberi]{Lipton04onapproximately}
Richard Lipton, Evangelos Markakis, Elchanan Mossel, and Amin Saberi.
\newblock On approximately fair allocations of indivisible goods.
\newblock In \emph{Proceedings of the ACM Conference on Electronic Commerce (EC)}, pages 125--131, 2004.

\bibitem[Liu et~al.(2024)Liu, Lu, Suzuki, and Walsh]{liu2024mixed}
Shengxin Liu, Xinhang Lu, Mashbat Suzuki, and Toby Walsh.
\newblock Mixed fair division: A survey.
\newblock \emph{Journal of Artificial Intelligence Research}, 80:\penalty0 1373--1406, 2024.
\newblock \doi{10.1613/jair.1.15800}.

\bibitem[Mackenzie and Suzuki(2026)]{mackenzie2026counterexamples}
Simon Mackenzie and Mashbat Suzuki.
\newblock Counterexamples to {EFX} for submodular and subadditive valuations.
\newblock \emph{arXiv preprint arXiv:2605.06451}, 2026.

\bibitem[Mahara(2023)]{mahara2023existence}
Ryoga Mahara.
\newblock Existence of {EFX} for two additive valuations.
\newblock \emph{Discrete Applied Mathematics}, 340:\penalty0 115--122, 2023.

\bibitem[Mahara(2024)]{mahara2024extension}
Ryoga Mahara.
\newblock Extension of additive valuations to general valuations on the existence of {EFX}.
\newblock \emph{Mathematics of operations research}, 49\penalty0 (2):\penalty0 1263--1277, 2024.

\bibitem[Mahara(2026)]{mahara2026existence}
Ryoga Mahara.
\newblock Existence of fair and efficient allocation of indivisible chores.
\newblock In \emph{Proceedings of the 2026 Annual ACM-SIAM Symposium on Discrete Algorithms (SODA)}, pages 6742--6766. SIAM, 2026.

\bibitem[Plaut and Roughgarden(2020)]{plaut2020almost}
Benjamin Plaut and Tim Roughgarden.
\newblock Almost envy-freeness with general valuations.
\newblock \emph{SIAM Journal on Discrete Mathematics}, 34\penalty0 (2):\penalty0 1039--1068, 2020.

\bibitem[Sun et~al.(2021)Sun, Chen, and Doan]{sun2021connections}
Ankang Sun, Bo~Chen, and Xuan~Vinh Doan.
\newblock Connections between fairness criteria and efficiency for allocating indivisible chores.
\newblock In \emph{Proceedings of the 20th International Conference on Autonomous Agents and Multiagent Systems (AAMAS 2021)}, pages 1281--1289. ACM, 2021.

\bibitem[Tao et~al.(2025)Tao, Wu, Yu, and Zhou]{tao2025existence}
Biaoshuai Tao, Xiaowei Wu, Ziqi Yu, and Shengwei Zhou.
\newblock On the existence of {EFX} (and pareto-optimal) allocations for binary chores.
\newblock \emph{Theoretical Computer Science}, 1042:\penalty0 115248, 2025.

\bibitem[Zeng and Mehta(2025)]{zengmehta2025orientations}
Jinghan~A. Zeng and Ruta Mehta.
\newblock On the structure of {EFX} orientations on graphs.
\newblock In \emph{Proceedings of the 24th International Conference on Autonomous Agents and Multiagent Systems}, pages 2309--2316, 2025.

\bibitem[Zhou and Wu(2024)]{zhou2024approximately}
Shengwei Zhou and Xiaowei Wu.
\newblock Approximately {EFX} allocations for indivisible chores.
\newblock \emph{Artificial Intelligence}, 326:\penalty0 104037, 2024.

\bibitem[Zhou et~al.(2024)Zhou, Wei, Li, and Li]{zhou2024complete}
Yu~Zhou, Tianze Wei, Minming Li, and Bo~Li.
\newblock A complete landscape of {EFX} allocations on graphs: Goods, chores and mixed manna.
\newblock In \emph{Proceedings of the Thirty-Third International Joint Conference on Artificial Intelligence}, pages 3049--3056, 2024.
\newblock \doi{10.24963/ijcai.2024/338}.

\end{thebibliography}

\newpage
\appendix

\section{Proof of Theorem~\ref{thm:tri} for General \texorpdfstring{$\boldsymbol{n}$}{n}}
\label{append:tri-n}
In this section, we prove Theorem~\ref{thm:tri} for an arbitrary fixed $n\geq 4$.

\subsection{The construction}
Fix an integer $n\ge 4$. Partition the agents into two groups $T_1$ and $T_2$ such that
\[
    |T_1|=\left\lfloor \frac n2\right\rfloor,
    \qquad
    |T_2|=\left\lceil \frac n2\right\rceil.
\]
Let $t_1=|T_1|$ and $t_2=|T_2|$.
Let $s=2t_2+1$.
The instance contains $n-1+2s$ items that are partitioned into three groups $A,B,C$ where $A$ contains $n-1$ items and each of $B$ and $C$ contains $s$ items.

The three values in the tri-valued instances are given by
$$p=1,\quad q=s+2=2t_2+3,\quad\mbox{and}\quad r=\frac{s(q+1)}2=(2t_2+1)(t_2+2).$$

The costs are defined as follows. For every agent $i\in T_1$,
\[
    c_i(a)=r \quad(a\in A),
    \qquad
    c_i(b)=p \quad(b\in B),
    \qquad
    c_i(c)=q \quad(c\in C),
\]
and for every agent $i\in T_2$,
\[
    c_i(a)=r \quad(a\in A),
    \qquad
    c_i(b)=q \quad(b\in B),
    \qquad
    c_i(c)=p \quad(c\in C).
\]

For $n=4$, we have $t_1=t_2=2$, $s=5$, $q=7$, and $r=20$.
{This is exactly our construction in Sect.}~\ref{sec:tri}.

\subsection{Non-existence of EFX allocations}
Let $X=(X_1,\ldots,X_n)$ be an arbitrary allocation. For a bundle $Y\subseteq A\cup B\cup C$, write
\[
    a(Y)=|Y\cap A|,
    \qquad
    b(Y)=|Y\cap B|,
    \qquad
    c(Y)=|Y\cap C|.
\]
All agents in $T_1$ have the same cost for a bundle $Y$, namely
\[
    P_1(Y)=r\cdot a(Y)+b(Y)+q\cdot c(Y),
\]
and all agents in $T_2$ have the same cost for $Y$, namely
\[
    P_2(Y)=r\cdot  a(Y)+q\cdot b(Y)+c(Y).
\]
Define
\[
    p_1=\min_{i\in \{1,\ldots,n\}} P_1(X_i),
    \qquad\mbox{and}\qquad
    p_2=\min_{i\in \{1,\ldots,n\}} P_2(X_i).
\]
For each agent, the total cost of all items is
\[
    (n-1)r+s(1+q)=(n-1)r+2r=(n+1)r,
\]
where the middle equality uses $r=s(q+1)/2$. 
Hence
\begin{equation}\label{eqn:2rbound}
    p_1\le \frac{(n+1)r}{n}<2r,
    \qquad\mbox{and}\qquad
    p_2\le \frac{(n+1)r}{n}<2r.
\end{equation}
We show that $X$ cannot be EFX.

\begin{proposition}
    In an EFX allocation, no bundle contains two or more items from $A$.
\end{proposition}
\begin{proof}
Suppose some agent receives a bundle containing at least two items from $A$. 
If the bundle contains at least three items from $A$, then after removing one of them, the remaining cost for that agent is at least $2r$, contradicting \eqref{eqn:2rbound} and EFX. 
If the bundle contains exactly two items from $A$ and also contains some other item, then after removing that other item, the remaining cost is again at least $2r$, contradicting \eqref{eqn:2rbound} and EFX.

Therefore, the only possible case is that some bundle is exactly $\{a,a'\}$ for two items $a,a'\in A$. 
Suppose first that this bundle is assigned to an agent in $T_1$, and suppose this agent is $1$ without loss of generality.
Removing one of the two items leaves cost $r$ (i.e., $\tau_1(X)=r$), so EFX implies $p_1\ge r$. 
Since $c_1(X_1)=2r$, and since the cost for all items is $(n+1)r$, the remaining $n-1$ bundles must all have cost exactly $r$ to maintain EFX. 
Among those $n-1$ bundles, there are only $n-3$ remaining large items from $A$; every bundle with one large item in $A$ with cost $r$ must be a singleton bundle. 
Consequently, two bundles must contain no large items from $A$ and have cost exactly $r$ to agent $1$. 
Thus, there must be nonnegative integers $y,z\le s$ such that
\begin{equation}\label{eqn:yqzr}
    y+q\cdot z=r.
\end{equation}
On the other hand, since $s=2t_2+1$ and $q=2t_2+3$,
\[
    r=(2t_2+1)(t_2+2)=qt_2+(s+1).
\]
Thus, $r\equiv s+1\pmod q$. The left-hand side of \eqref{eqn:yqzr} is congruent to $y$ modulo $q$, while $0\le y\le s$. This is impossible because $s+1$ is not in the set $\{0,1,\ldots,s\}$.

The case where the bundle $\{a,a'\}$ is assigned to an agent in $T_2$ is symmetric: EFX implies $p_2\geq r$, and the same reasoning gives nonnegative integers $y,z\le s$ satisfying $q\cdot y+z=r$, which is impossible modulo $q$ because $0\le z\le s$ but $r\equiv s+1\pmod q$.
\end{proof}

By the proposition above, since there are $n-1$ items in $A$ and $n$ agents, exactly one bundle contains no item in $A$ and every other bundle contains exactly one item in $A$.

\begin{proposition}
    In an EFX allocation, $p_1\ge r$ and $p_2\ge r$.
\end{proposition}
\begin{proof}
We prove $p_1\ge r$; the proof of $p_2\ge r$ is symmetric. Suppose toward a contradiction that $p_1<r$. 
Consider any agent in $T_1$ whose bundle contains a large item in $A$. 
By the previous proposition, this bundle contains exactly one large item in $A$. 
It cannot contain any additional item: if it did, then removing that additional item would leave cost at least $r>p_1$, violating EFX. 
Thus, every $T_1$-agent who receives a large item receives a singleton large item bundle.

There is at least one such agent, because $|T_1|=\lfloor n/2\rfloor\ge 2$ and only one bundle contains no large item. Therefore, $p_2\le r$. Now consider any agent in $T_2$ whose bundle contains a large item in $A$. Since $p_2\le r$, this agent can receive at most one more item outside $A$; otherwise, after removing this extra item, her remaining cost would be strictly larger than $r$, hence strictly larger than $p_2$, violating EFX.

It follows that the unique bundle containing no large item in $A$ must contain at least $2s-t_2$ items from $B\cup C$: outside this bundle, the $t_2$ agents in $T_2$ can hold at most one non-$A$ item each, and the $T_1$ agents holding an $A$-item hold none. 
Since $2s-t_2=s+(t_2+1)$ and only $s$ items of $B$ are cheap under $P_1$, the $P_1$-cost of the no-$A$ bundle is at least
\[
    s+(t_2+1)q.
\]
This quantity is strictly larger than $r$, since
\[
    s+(t_2+1)q-r
    =(2t_2+1)+(t_2+1)(2t_2+3)-(2t_2+1)(t_2+2)
    =2t_2+2>0.
\]
Thus, the no-$A$ bundle has $P_1$-cost greater than $r$, while every bundle containing an $A$-item has $P_1$-cost at least $r$. This contradicts $p_1<r$. Hence $p_1\ge r$. 

The proof of $p_2\ge r$ is identical after interchanging the roles of $T_1$ and $T_2$, and of $B$ and $C$; in the counting step, we use $|T_1|\le t_2$, so the no-$A$ bundle contains at least $2s-|T_1|\ge 2s-t_2$ non-$A$ chores, which is enough for the same lower bound. 
\end{proof}

We can now derive the final contradiction. Let $X_i$ be the unique bundle containing no $A$-item.

First suppose $i\in T_1$. Write $y=b(X_i)$ and $z=c(X_i)$.
By the proposition above,
\begin{equation}\label{eqn:yzsymm}
    y+q z=P_1(X_i)\ge r,
    \qquad\mbox{and}\qquad
    q y+z=P_2(X_i)\ge r. 
\end{equation}
Since $0\le y,z\le s$, the inequalities in \eqref{eqn:yzsymm} imply
\[
    z\ge \left\lceil \frac{r-s}{q}\right\rceil
    \qquad\text{and}\qquad
    y\ge \left\lceil \frac{r-s}{q}\right\rceil.
\]
Moreover,
\[
    \frac{r-s}{q}
    =\frac{(2t_2+1)(t_2+2)-(2t_2+1)}{2t_2+3}
    =t_2+\frac{1}{2t_2+3},
\]
so
\begin{equation}\label{eqn:yzt_2+1}
    y\ge t_2+1,
    \qquad\mbox{and}\qquad
    z\ge t_2+1.  
\end{equation}
Because $i\in T_1$ and $y\ge t_2+1>0$, the bundle $X_i$ contains a $B$-item of cost $1$ for agent $i$. 
EFX for agent $i$ therefore implies
\[
    P_1(X_i)-1\le p_1.
\]
Using \eqref{eqn:yzt_2+1}, we get
\begin{equation}\label{eqn:p1lower}
    p_1\ge P_1(X_i)-1
    =y+q z-1
    \ge (t_2+1)(q+1)-1
    =r+t_2+1.   
\end{equation}
Now at most $s-y\le t_2$ items of $B$ remain outside $X_i$. Consider any $A$-holding bundle that contains no item from $C$. 
Its $P_1$-cost is at most $r+t_2$, because it contains one $A$-item and can contain at most all the remaining $t_2$ items from $B$. This is strictly less than $p_1$ by \eqref{eqn:p1lower}, contradicting the definition of $p_1$. Therefore, every one of the $n-1$ $A$-holding bundles must contain at least one item from $C$.
On the other hand, at most $s-z\le t_2$ items from $C$ remain outside $X_i$, whereas
$n-1>t_2$
for every $n\ge 4$ with $t_2=\lceil n/2\rceil$. 
This is impossible.

The case $i\in T_2$ is symmetric. In that case, EFX for the no-$A$ agent gives $p_2\ge r+t_2+1$; then every $A$-holding bundle must contain at least one item from $B$, but only at most $t_2$ items from $B$ remain outside $X_i$, again fewer than the $n-1$ $A$-holding bundles.

Thus, no EFX allocation exists for the constructed instance.

\section{More Details in Concatenation of \texorpdfstring{$\boldsymbol{M_{01}}$}{M01} Prefix and \texorpdfstring{$\boldsymbol{M_2}$}{M2} Allocation}
\label{append:concatenation}

\begin{paperparagraph}
If $M_2=\emptyset$, Lemma~\ref{lem:canonical}(a) already proves EFX. As in Sect.~\ref{subsec:combine}, the cases below therefore assume $M_2\neq\emptyset$.
\end{paperparagraph}

In this section, we provide details in the case study for concatenating the $M_{01}$ Prefix and the $M_2$ Allocation discussed in Sect.~\ref{subsec:combine}.

\subsection{Case \texorpdfstring{$\boldsymbol{b=0}$}{b = 0}}
This case is relatively easy to handle.
The corresponding arguments in Sect.~\ref{subsec:b=0} already give a formal proof.

\subsection{Case \texorpdfstring{$\boldsymbol{b=1}$}{b = 1}}
Recall that $|M_{01}|=4a+1$.  Thus, in a canonical allocation of $M_{01}$,
exactly one agent receives $a+1$ items, and the other three agents receive $a$
items.  Recall that the former agent is the long agent and the latter agents are the short agents. We discuss two cases, according to whether some edge type in $M_2$ has multiplicity at least three.

\subsubsection{There exists an edge type with multiplicity at least three}

Choose an edge type of maximum multiplicity, and assume without loss of
generality that this type is $(1,2)$.  Thus, $x_{12}\ge 3$ and
$x_{12}\ge x_{ij}$ for every edge type $(i,j)$.

We further split into two cases.
\paragraph{B.2.1(a) There is another edge intersecting $\boldsymbol{(1,2)}$.}
Assume without loss of generality that edge is an item of type $(2,3)$.
Choose a canonical allocation $P$ of $M_{01}$ with agent $4$ being the long
agent, i.e.,
\[
    |P_4|=a+1,\qquad |P_1|=|P_2|=|P_3|=a.
\]
Use two type-$(1,2)$ items and one type-$(2,3)$ item for gap filling: give one
type-$(1,2)$ item to agent $1$, one type-$(1,2)$ item to agent $2$, and the
type-$(2,3)$ item to agent $3$.

Let $P'$ denote the allocation after this gap filling.  We first verify that
$P'$ is envy-free.  Agents $1,2,3$ now have $a+1$ items, and the extra item
assigned to each of them is small for its owner.  For any two agents among
$1,2,3$, the $M_{01}$ parts have the same size $a$ and are canonical, hence
the usual equal-quota argument in Lemma~\ref{lem:canonical}(a) gives
\[
    c_i(P_i)\le c_i(P_j).
\]
Adding one item of cost $1$ to agent $i$ and one item of cost at least $1$ to
agent $j$ preserves the inequality.  Therefore, agents $1,2,3$ do not envy each
other.

Next, consider the comparison between a short agent $i\in\{1,2,3\}$ and the
long agent $4$.  If all items of $P_i$ are small for agent $i$, then
$c_i(P_i')=a+1$, whereas $P_4$ contains $a+1$ items and hence
$c_i(P_4)\ge a+1$.  If not all items of $P_i$ are small for agent $i$, then
by canonicity all items of $M_{01}$ that are small for $i$ must have already been assigned to $P_i$\,; in particular, the long bundle
$P_4$ is sufficiently costly from agent $i$'s perspective, and again
$c_i(P_i')\le c_i(P_4)$. 

Finally, from agent $4$'s perspective, all three gap-filling items assigned to agents $1,2,3$ are large.  Hence, the additional
cost on every short bundle is $r$, while agent $4$ receives no gap-filling
item.  Therefore, agent $4$ also does not envy any of agents $1,2,3$.
Thus, $P'$ is envy-free.

Let $M_2^-$ be the remaining multiset of $M_2$ items after the above
gap-filling step.  If $M_2^-$ is not exceptional, we allocate $M_2^-$ by
Lemma~\ref{lem:M2properties}(i).  Since $P'$ is envy-free and every residual
envy is certified by removing a small item, Lemma~\ref{lem:composition} implies
that the combined allocation is EFX.

It remains to consider the case where $M_2^-$ is exceptional.  Since
$x_{12}$ was maximum and we removed only two type-$(1,2)$ items, at least one
type-$(1,2)$ item remains in $M_2^-$.  Therefore, the exceptional agents can
only be either $\{1,2\}$ or $\{3,4\}$.

\subparagraph{The exceptional agents of $\boldsymbol{M_2^-}$ are $\boldsymbol{\{1,2\}}$. } Then $M_2^-$ must consist
of one type-$(1,2)$ item and $4q+3$ type-$(3,4)$ items.  Since the original
multiplicity of type $(1,2)$ was maximum and two such items have been removed (for gap-filling),
this is possible only when
\[
    x_{12}=x_{34}=3,
\]
and there is exactly one type-$(2,3)$ item.  Hence the whole multiset $M_2$
consists of three type-$(1,2)$ items, one type-$(2,3)$ item, and three
type-$(3,4)$ items.

In this special situation we do not use the preceding gap-filling strategy.
Instead, we take a super-canonical allocation $P$ of $M_{01}$, whose existence is
guaranteed by Lemma~\ref{lem:canonical}(b).  Exactly one agent is long.  We
allocate the items in $M_2$ directly.

If agent $1$ is long, let agent $2$ take the three type-$(1,2)$ items, let
agent $3$ take the type-$(2,3)$ item and one type-$(3,4)$ item, and let agent
$4$ take the remaining two type-$(3,4)$ items.  If agent $2$ is long, we
symmetrically let agent $1$ take the three type-$(1,2)$ items, while agents
$3$ and $4$ are treated in the same way.  The cases where agent $3$ or agent
$4$ is long are symmetric.

We verify EFX for the representative case in which agent $1$ is long.  The
long agent receives no item from $M_2$, and every other agent receives only
items that are small for herself.  For every short agent $i$, super-canonicity
gives
\[
    c_i(P_i)\le c_i(P_1)-r.
\]
After appending the above $M_2$ allocation, each short agent receives at most
three small items.  Since $r>2$, after removing any small item from a short
agent's appended part, the remaining additional cost is at most $2<r$.
Therefore, no short agent strongly envies the long agent.  Comparisons among
short agents are immediate from the equal-quota canonical property of their
prefix bundles and the fact that the appended bundles are small for their
owners.  Finally, the long agent keeps her own prefix bundle and only sees other bundles become more costly.  Hence, the final allocation is EFX.

\subparagraph{The exceptional agents of $\boldsymbol{M_2^-}$ are $\boldsymbol{\{3,4\}}$.}  In the residual
allocation of $M_2^-$ given by Lemma~\ref{lem:M2properties}(ii), choose agent
$4$ to be the disadvantageous agent.  All residual envy relations except
possibly the one from agent $4$ to agent $3$ are already certified by removing
a small item, and hence are handled by Lemma~\ref{lem:composition}, because
the gap-filled prefix $P'$ is envy-free.

It remains to check the comparison from agent $4$ to agent $3$.  If $P_4$
contains no small item for agent $4$, then all items in agent $4$'s final
bundle are large for her, and the
large-item removal allowed in Lemma~\ref{lem:M2properties}(ii) is sufficient.
If $P_4$ contains a small item for agent $4$, then the EFX threshold of the
prefix part is obtained by removing a small item.  Since $P$ is canonical,
we have
\[
    c_4(P_4)-1\le c_4(P_3).
\]
Let $h$ be the gap-filling item assigned to agent $3$.  Since $h$ is of type
$(2,3)$, it is large for agent $4$, and hence
\[
    c_4(P_3\cup\{h\})=c_4(P_3)+r.
\]
Therefore,
\[
    c_4(P_4)
    \le c_4(P_3)+1
    = c_4(P_3\cup\{h\})-(r-1).
\]
On the other hand, in the exceptional residual allocation of $M_2^-$, agent
$4$ is chosen as the disadvantageous agent.  Thus, the only possible problematic
residual comparison is from agent $4$ to agent $3$, and it is certified by
removing a large item:
\[
    c_4(B_4)-r\le c_4(B_3).
\]
Since the final bundle of agent $4$ contains a small item in the prefix, its
EFX threshold in the combined allocation is
\[
    \tau_4(X)=c_4(P_4)+c_4(B_4)-1.
\]
Combining the two displayed inequalities gives
\[
\begin{aligned}
    \tau_4(X)
    &= c_4(P_4)+c_4(B_4)-1  \\
    &\le \bigl(c_4(P_3\cup\{h\})-(r-1)\bigr)+c_4(B_4)-1 \\
    &= c_4(P_3\cup\{h\})+c_4(B_4)-r \\
    &\le c_4(P_3\cup\{h\})+c_4(B_3) \\
    &= c_4(X_3).
\end{aligned}
\]
Therefore agent $4$ does not strongly envy agent $3$ in the combined
allocation.  Hence the final allocation is EFX.

\paragraph{B.2.1(b) No other edge intersects $\boldsymbol{(1,2)}$.}
In this case, all items of $M_2$ have type $(1,2)$ or type $(3,4)$, and by our
choice of the maximum multiplicity edge,
\[
    x_{12}\ge x_{34},\qquad x_{12}\ge 3.
\]

We first choose a canonical allocation $P$ of $M_{01}$ in which one of agents
$3,4$ is long and the other is short, with the following additional property:
the short agent among $\{3,4\}$ has an advantage of at least $r$ over the long
agent among $\{3,4\}$.  That is, after relabeling agents $3$ and $4$ if
necessary, we may assume
\[
    |P_3|=a+1,\qquad |P_4|=a,
\]
and
\[
    c_4(P_4)\le c_4(P_3)-r.
\]
Such a canonical allocation always exists.  Indeed, first choose a canonical
allocation with agent $3$ long and agent $4$ short.  If the displayed
inequality fails, then the only possible reason is that agent $4$ regards all
items in $P_3$ as small.  Moving one item from $P_3$ to $P_4$ and then swapping the names of
agents $3$ and $4$ yields the desired canonical allocation.

We firstly handle two small exceptional parameter values separately.

If
\[
    x_{12}=x_{34}=3,
\]
allocate two type-$(1,2)$ items to agent $1$, one type-$(1,2)$ item to agent
$2$, no $M_2$ item to agent $3$, and all three type-$(3,4)$ items to agent
$4$.  If
\[
    x_{12}=4,\qquad x_{34}=3,
\]
we use the same allocation, except that agent $2$ receives one additional
type-$(1,2)$ item.  In both cases, agents $1,2,4$ receive only items that are
small for themselves, while agent $3$ receives no item from $M_2$.  The only
potentially delicate comparison is from agent $4$ to agent $3$, but this is
covered by the prefix advantage
\[
    c_4(P_4)\le c_4(P_3)-r.
\]
Since $r>2$, removing any one of agent $4$'s small type-$(3,4)$ items leaves
at most two additional units of cost, which is strictly less than the prefix
advantage.  Therefore, the resulting allocation is EFX.

We may therefore assume that we are not in either of the two special cases
above.  Use three type-$(1,2)$ items for gap filling: give one to agent $1$,
one to agent $2$, and one to agent $4$.  The first two agents receive small
items, while agent $4$ receives an item that is large for her.  By the choice
of $P$, the large item assigned to agent $4$ is exactly compensated by the
$r$-advantage of agent $4$ over agent $3$ in the prefix.  Hence, the
gap-filled prefix is envy-free.

Let $M_2^-$ be the residual multiset.  If $M_2^-$ is not exceptional, then
Lemma~\ref{lem:M2properties}(i) and Lemma~\ref{lem:composition} imply that the
combined allocation is EFX.

If $M_2^-$ is exceptional, then its exceptional agents cannot be $\{1,2\}$.
Indeed, for $\{1,2\}$ to be exceptional, the residual would have to contain
at most one type-$(1,2)$ item and $4q+3$ type-$(3,4)$ items.  Since three
type-$(1,2)$ items were removed and $x_{12}\ge x_{34}$ originally, this can
only happen in the two special cases $x_{12}=x_{34}=3$ or
$x_{12}=4,x_{34}=3$, which we have already handled.  Thus, the exceptional
agents must be $\{3,4\}$.

Choose agent $3$ to be the disadvantageous agent in the exceptional residual
allocation.  If $P_3$ contains no small item for agent $3$, then agent $3$'s
final bundle contains only large items in the relevant exceptional comparison,
so the large-item removal guaranteed by Lemma~\ref{lem:M2properties}(ii)
suffices.  Otherwise, $P_3$ contains a small item, and the gap-filling item
assigned to agent $4$ is of type $(1,2)$, which is large for agent $3$.
This creates the necessary $r-1$ slack to compensate for the exceptional
large-removal comparison.  Therefore, agent $3$ does not strongly envy agent
$4$, and the final allocation is EFX.

\subsubsection{All edge multiplicities are at most two}

We now assume that every edge type in $M_2$ has multiplicity at most two.
Then no residual multiset obtained by deleting items can be exceptional,
because every exceptional configuration contains an edge type of multiplicity
at least three.  Therefore, it is enough to use some items of $M_2$ for
gap-filling so that the prefix becomes envy-free; the remaining items can then
be allocated by Lemma~\ref{lem:M2properties}(i) and the resultant allocation is EFX by Lemma~\ref{lem:composition}.

\paragraph{B.2.2(a) Some agent considers at least three items of $\boldsymbol{M_2}$ as large.}
Assume without loss of generality that {agent $1$ considers at least three}
items of $M_2$ as large.  Choose three such items.  These items have edge types among
$(2,3),(2,4),(3,4)$.

Choose a canonical allocation $P$ of $M_{01}$ with agent $1$ long:
\[
    |P_1|=a+1,\qquad |P_2|=|P_3|=|P_4|=a.
\]
Under the assumption that the number of every category of items is at most $2$, then the three items which agent $1$ views as large cannot be of the same type, indicating that agents $2,3,4$ can always get an item which is small for them correspondingly.

We then verify that the current allocation after gap-filling is envy-free.
The canonical allocation before the gap-filling is EFX (Lemma~\ref{lem:canonical}(a)), so agent $1$ does not envy any of the agents $2,3,4$ upon removing one item from $P_1$.
Since, in the gap-filling, $P_1$ is unchanged while an item that is large for agent $1$ is added to each of $P_2,P_3,P_4$, the allocation after gap-filling is envy-free for agent $1$ since the extra large item added to $P_2,P_3$, or $P_4$ compensates the EFX certificate in $P_1$.
For agent $2$, before gap-filling, she has an advantage of one item (with cost either $1$ or $r$) against agent $1$, and she does not envy any of agents $3$ and $4$.
After the gap-filling, agent $2$ receives one more small item, and each of agents $3$ and $4$ receives one more item that can be small or large for agent $2$.
Envy-freeness holds for agent $2$: she does not envy agent $1$ as the advantage before gap-filling is at least $1$ (either $1$ or $r$ to be exact), so the extra small item does not surpass this advantage; she does not envy any of agents $3$ and $4$ as the new item added to $P_2$ is at most as large as the item added to $P_3$ or $P_4$.
The envy-freeness for agents $3$ and $4$ can be shown similarly.

After this gap-filling step, let $M_2^-$ be the remaining multiset.  
Since all multiplicities are at most two, $M_2^-$ is not exceptional.  
Allocate $M_2^-$ using Lemma~\ref{lem:M2properties}(i).  
By Lemma~\ref{lem:composition}, the combined allocation is EFX.

\paragraph{B.2.2(b) No agent has three large items in $\boldsymbol{M_2}$.}
Since every item of $M_2$ is large for exactly two agents, the total number of
large incidences over all agents is $2|M_2|$.  If no agent considers three
items as large, then each agent considers at most two items as large, and
\[
    2|M_2|\le 8.
\]
Thus, $|M_2|\le 4$.

If $|M_2|\le 3$, then, using the fact that every edge multiplicity is at most
two, we can assign the items of $M_2$ so that exactly $|M_2|$ agents receive
one item each, and each such item is small for its recipient.  Choose the
canonical prefix so that an agent receiving no item from $M_2$ is the long
agent.  Then the final allocation is EFX: the long agent receives no
additional item and only sees other bundles become more costly, while every
short agent who receives an item receives a small item, and comparisons among
short agents are handled by the equal-quota canonical argument.

It remains to consider $|M_2|=4$.  Under the current assumptions, there are
only two possible configurations up to relabeling:
\begin{enumerate}[label=(\arabic*)]
    \item two type-$(1,2)$ items and two type-$(3,4)$ items;
    \item one item of each of the four types
    \[
        (1,2),\ (2,3),\ (3,4),\ (1,4).
    \]
\end{enumerate}

Take a super-canonical allocation $P$ of $M_{01}$, guaranteed by
Lemma~\ref{lem:canonical}(b).  By symmetry, assume agent $1$ is the long
agent.

In configuration (1), give the two type-$(1,2)$ items to agent $2$, give one
type-$(3,4)$ item to agent $3$, and give the other type-$(3,4)$ item to agent
$4$.

In configuration (2), give the type-$(1,2)$ item and the type-$(2,3)$ item to
agent $2$, give the type-$(3,4)$ item to agent $3$, and give the type-$(1,4)$
item to agent $4$.

In both configurations, every short agent receives only items that are small
for herself, while the long agent receives no item of $M_2$.  By
super-canonicity, each short agent has an $r$-advantage over the long agent in
the prefix.  Since every short agent receives at most two small items from
$M_2$, after removing one such item, her additional residual cost is at most
$1<r$.  Thus, no short agent envies the long agent.  Comparisons among
short agents follow from the equal-quota canonical property of their prefix
bundles and the fact that items they receive in $M_2$ are all small in their perspective.  Finally, the long agent receives no $M_2$ item
and only sees the other bundles become weakly more costly.  Hence, the final
allocation is EFX.

Here we have completed the proof for the case $b=1$.

\subsection{Case \texorpdfstring{$\boldsymbol{b=2}$}{b = 2}}

We now consider the case $|M_{01}|=4a+2$.  In a canonical allocation of
$M_{01}$, exactly two agents receive $a$ items and the other two agents receive
$a+1$ items.  Recall that the former agents are short and the latter agents are long. We distinguish two cases according to whether some edge type in $M_2$ has
multiplicity at least two.

\subsubsection{Some edge type has multiplicity at least two}

Choose an edge type of maximum multiplicity, and assume without loss of
generality that it is $(1,2)$.  Thus, $x_{12}\ge 2$ and
$x_{12}\ge x_{ij}$ for every edge type $(i,j)$.

Choose a canonical allocation $P$ of $M_{01}$ in which agents $1$ and $2$ are
short and agents $3$ and $4$ are long:
\[
    |P_1|=|P_2|=a,
    \qquad
    |P_3|=|P_4|=a+1.
\]
Use two type-$(1,2)$ items for gap filling, assigning one to agent $1$ and one
to agent $2$.  Let $P'$ denote the resulting partial allocation.

We first verify that $P'$ is envy-free.  Agents $1$ and $2$ each receive a
small gap-filling item, so their bundle now contains $a+1$ items.  For each
$i\in\{1,2\}$ and each $j\in\{3,4\}$, the same argument as in
Lemma~\ref{lem:canonical}(a) gives
\[
    c_i(P_i)+1\le c_i(P_j).
\]
Indeed, if $P_i$ consists entirely of items small for $i$, then
$c_i(P_i)+1=a+1\le c_i(P_j)$; otherwise, by canonicity, all remaining
$M_{01}$ items outside $P_i$ are large for $i$, and the inequality is even
easier.  Hence, neither agent $1$ nor agent $2$ envies agents $3,4$ after
gap filling.

For agents $3$ and $4$, the two gap-filling items assigned to agents $1$ and
$2$ are large.  Since $P$ is canonical, for $j\in\{3,4\}$ and
$i\in\{1,2\}$,
\[
    c_j(P_j)-\min_{g\in P_j}c_j(g)\le c_j(P_i).
\]
Therefore,
\[
    c_j(P_j)
    \le c_j(P_i)+\min_{g\in P_j}c_j(g)
    \le c_j(P_i)+r
    = c_j(P_i'),
\]
where the last equality uses that the type-$(1,2)$ gap-filling item assigned
to $i$ is large for $j$.  Thus, agents $3,4$ also do not envy agents $1,2$.
The comparison between agents $1$ and $2$, and the comparison between agents
$3$ and $4$, follow from the equal-quota part of the canonical allocation.
Hence, $P'$ is envy-free.

Let $M_2^-$ be the remaining multiset of $M_2$ items.  If $M_2^-$ is not
exceptional, then we allocate $M_2^-$ according to
Lemma~\ref{lem:M2properties}(i).  Since the prefix $P'$ is envy-free and all
residual envy is certified by removing a small item, Lemma~\ref{lem:composition}
implies that the combined allocation is EFX.

It remains to consider the case where $M_2^-$ is exceptional.  We first observe
that the exceptional agents can only be either $\{1,2\}$ or $\{3,4\}$.
Indeed, suppose for example that the exceptional agents were $\{2,4\}$.
Then, by Lemma~\ref{lem:M2properties}(ii), the residual multiset would contain
at least three items of the disjoint type $(1,3)$.  Since we removed only two
type-$(1,2)$ items (for gap-filling), this would imply that, before gap filling, the multiplicity of type $(1,3)$ was at least three, while the relevant multiplicity of type $(1,2)$ was only two in the tight residual situation.  This contradicts the choice of $(1,2)$ as an edge type of maximum multiplicity.  
The same argument excludes all exceptional pairs other than $\{1,2\}$ and $\{3,4\}$.

If the exceptional agents are $\{3,4\}$, then both exceptional agents are long
agents in the prefix allocation.  Therefore Lemma~\ref{lem:easycombine}
applies directly, and we obtain an EFX allocation of $M_{01}\cup M_2$.

It remains to handle the case where the exceptional agents are $\{1,2\}$.
Since two type-$(1,2)$ items have already been removed, and since
type $(1,2)$ had maximum multiplicity before gap filling, this case can occur
only when
\[
    x_{12}=x_{34}=3.
\]
In this situation the residual would consist of one type-$(1,2)$ item and
three type-$(3,4)$ items.  We avoid this residual configuration by abandoning
the preceding gap-filling step and handling the original instance directly.

First assume $a>0$.  Apply Lemma~\ref{lem:canonical}(c) to the partition
\[
    \{1,2\}\mathbin{\dot\cup}\{3,4\}.
\]
We obtain a super-canonical allocation $P$ of $M_{01}$ in which one agent from
$\{1,2\}$ and one agent from $\{3,4\}$ are short.  Relabeling inside the two
pairs if necessary, assume that agents $2$ and $4$ are short and agents $1$
and $3$ are long.  Allocate all three type-$(1,2)$ items to agent $2$, and all
three type-$(3,4)$ items to agent $4$.

We verify EFX.  Agents $2$ and $4$ only receive items that are small for
themselves.  For the comparison from a short agent to a long agent,
super-canonicity gives, for every short agent $i$ and long agent $j$,
\[
    c_i(P_i)\le c_i(P_j)-r.
\]
Since each of agents $2$ and $4$ receives three small items from $M_2$, after
removing one such item the remaining additional cost is $2$.  As $r>2$, we
have
\[
    c_i(P_i)+2 < c_i(P_j)
\]
for every short agent $i$ and every long agent $j$.  Hence, no short agent
strongly envies a long agent.

The comparison between the two short agents is also safe.  For instance, from
agent $2$'s perspective, agent $4$ receives three type-$(3,4)$ items, all of
which are large for agent $2$.  Therefore,
\[
    c_2(P_2)+2
    \le c_2(P_4)+3r
    = c_2(P_4\cup B_4).
\]

The comparison from agent $4$ to agent $2$ is
symmetric.  Finally, the two long agents receive no item from $M_2$, so their
own bundles are unchanged while all comparison bundles become weakly more
costly.  Thus, the constructed allocation is EFX.

Now assume $a=0$.  Then $|M_{01}|=2$.  If there exists a super-canonical
prefix in which exactly one of agents $1,2$ and exactly one of agents $3,4$
are long, the above argument applies without change.  The only remaining
case is when the two items in $M_{01}$ are small for agents $1$ and $2$,
respectively; the symmetric case where they are small for agents $3$ and $4$
is identical.

In this hard case, choose the prefix in which agents $1$ and $2$ are long.
Thus, each of agents $1$ and $2$ receives one small $M_{01}$ item, while agents
$3$ and $4$ receive no $M_{01}$ item.  Allocate two type-$(1,2)$ items to
agent $1$ and one type-$(1,2)$ item to agent $2$; allocate two type-$(3,4)$
items to agent $3$ and one type-$(3,4)$ item to agent $4$.

Agent $1$ has total cost $3$ and
EFX threshold $2$; agent $2$ has total cost $2$ and EFX threshold $1$; agents
$3$ and $4$ have thresholds at most $1$.  Since every other nonempty bundle
has cost at least $1$, and since any bundle of type $(3,4)$ has cost at least
$r>2$ from the perspective of agents $1,2$, all EFX inequalities hold.  Hence,
this corner case is also resolved.

\subsubsection{All edge types have multiplicity at most one}

Now assume that $M_2$ is a simple graph.  Choose a super-canonical
allocation $P$ of $M_{01}$, whose existence is guaranteed by
Lemma~\ref{lem:canonical}(b).  Relabeling if necessary, suppose agents $1$
and $2$ are short and agents $3$ and $4$ are long:
\[
    |P_1|=|P_2|=a,
    \qquad
    |P_3|=|P_4|=a+1.
\]

We allocate all items in $M_2$ directly.  In each of the following cases, the
allocation is chosen so that the possible envy from long agents to short agents
is compensated by a large item placed on the corresponding short bundle, while
possible envy from short agents to long agents is compensated by
super-canonicity.

We repeatedly use the following two facts.  First, for every short agent
$i\in\{1,2\}$ and every long agent $j\in\{3,4\}$, super-canonicity gives
\begin{equation}\label{eqn:b2-super}
    c_i(P_i)\le c_i(P_j)-r.
\end{equation}
Second, for every long agent $j\in\{3,4\}$ and every short agent
$i\in\{1,2\}$, canonicity gives
\begin{equation}\label{eqn:b2-long-short}
    c_j(P_j)-\min_{g\in P_j}c_j(g)\le c_j(P_i).
\end{equation}
Thus, if the short bundle of $i$ receives an item that is large for $j$, then
\[
    c_j(P_j)\le c_j(P_i)+r\le c_j(P_i\cup B_i),
\]
where $B_i$ denotes the $M_2$ items assigned to agent $i$.

\paragraph{B.3.2(a) $\boldsymbol{x_{13}+x_{14}>0}$ and $\boldsymbol{x_{23}+x_{24}>0}$.}

Let agent $1$ take all existing items among $(1,3)$ and $(1,4)$, and let
agent $2$ take all existing items among $(2,3)$ and $(2,4)$.  These items are
small for the short agents who receive them.

\subparagraph{- Both $\boldsymbol{(1,2)}$ and $\boldsymbol{(3,4)}$ exist} Give the item $(1,2)$ to whichever of
agents $1,2$ currently receives fewer items, breaking ties arbitrarily.  For
the item $(3,4)$, choose an agent in $\{3,4\}$ who regards at least one item
assigned to agent $1$ and at least one item assigned to agent $2$ as large,
and give $(3,4)$ to this agent.  Such an agent exists by the construction:
if the item $(1,2)$ is assigned to agent $1$, then both agents $3,4$ regard
it as large, while agent $2$ receives at least one of $(2,3)$ and $(2,4)$,
which is large for at least one of agents $4$ and $3$, respectively.  The
other cases are symmetric.

We verify EFX in this subcase.  First, consider a short agent
$i\in\{1,2\}$.  Agent $i$ only receives items that are small for herself, and
after removing one such item, the remaining number of $M_2$ items in her
bundle is at most $2$.  Hence, for every long agent $j\in\{3,4\}$,
super-canonicity gives
\[
    c_i(X_i)-1
    \le c_i(P_i)+2
    < c_i(P_i)+r
    \le c_i(P_j)
    \le c_i(X_j),
\]
where we use $r>2$.  Thus, no short agent strongly envies a long agent. Moreover, obviously, the comparison between the two short agents is also safe.

Now consider a long agent.  If a long agent receives no item from $M_2$, then
her own bundle is unchanged, and her EFX inequalities follow directly from
the canonical EFX property of the prefix.  The only long agent whose threshold
may increase is the one receiving the item $(3,4)$.  By our choice of this
agent, say agent $j$, each short bundle contains an item that is large for
$j$.  Hence, for every short agent $i\in\{1,2\}$,
\[
    c_j(P_j\cup\{(3,4)\})-1
    =c_j(P_j)
    \le c_j(P_i)+r
    \le c_j(X_i).
\]

The comparison between agents $3$ and $4$ follows from the equal-quota
canonical comparison and the fact that at most one of them receives a small
item $(3,4)$.  Therefore, the final allocation is EFX in this subcase.

\subparagraph{- $\boldsymbol{(1,2)}$ exists and $\boldsymbol{(3,4)}$ does not}The item $(1,2)$ can be assigned to
either agent $1$ or agent $2$.

We verify EFX in this subcase.  Both short agents receive only items that are
small for themselves.  For each short agent $i\in\{1,2\}$ and each long agent
$j\in\{3,4\}$, after removing one item from $X_i\cap M_2$, the remaining
additional cost is at most $2$.  Therefore,
\[
    c_i(X_i)-1
    \le c_i(P_i)+2
    < c_i(P_i)+r
    \le c_i(P_j)
    \le c_i(X_j),
\]
so no short agent strongly envies a long agent.

For the comparison between agents $1$ and $2$, note that every cross edge
assigned to one of them is large for the other short agent.  Hence, even if
one short agent receives more $M_2$ items than the other, the other short
bundle is sufficiently costly from her perspective.  More explicitly, for
$\{i,k\}=\{1,2\}$,
\[
    c_i(X_i)-1
    \le c_i(P_i) + 2 < c_i(P_k) + r \le c_i(X_k),
\]

Finally, agents $3$ and $4$ receive no item from $M_2$, so their own thresholds do not increase; all their EFX inequalities follow from the prefix allocation.  Thus,
the final allocation is EFX in this subcase.

\subparagraph{- $\boldsymbol{(1,2)}$ does not exist and $\boldsymbol{(3,4)}$ exists} Then, when all four cross edges
\[
    (1,3),(1,4),(2,3),(2,4)
\]
exist, assign $(3,4)$ to either agent $3$ or agent $4$.

We verify EFX in this subcase.  Each short agent receives exactly two small
items.  Thus, for each short agent $i\in\{1,2\}$ and each long agent
$j\in\{3,4\}$,
\[
    c_i(X_i)-1
    \le c_i(P_i)+1
    < c_i(P_i)+r
    \le c_i(P_j)
    \le c_i(X_j).
\]

The two short agents do not strongly envy each other, because after removing
one of her two small items, each short agent has one remaining $M_2$ item,
while the other short agent has two $M_2$ items, each of cost at least $1$.

For long agents, suppose without loss of generality that agent $3$ receives
$(3,4)$.  Then agent $3$ regards $(1,4)$ and $(2,4)$ as large.  Since both
items are assigned to the short side, for each short agent $i\in\{1,2\}$,
the bundle $X_i$ contains an item that is large for agent $3$.  Hence,
\[
    c_3(X_3)-1
    =c_3(P_3)
    \le c_3(P_i)+r
    \le c_3(X_i).
\]

Agent $4$ receives no item from $M_2$, so her own threshold does not increase
and her EFX inequalities follow from the canonical prefix.  The case where
agent $4$ receives $(3,4)$ is symmetric.  Therefore, the final allocation is
EFX in this subcase.

Otherwise, at least one short agent $i\in\{1,2\}$ receives at most one cross
edge; assign $(3,4)$ to such a short agent.

We verify EFX in this subcase.  Let $i$ be the short agent who receives
$(3,4)$, and let $k$ be the other short agent.  By the assumption of this
subcase, agent $i$ receives at most one cross edge in addition to $(3,4)$.
Thus, from agent $i$'s own perspective, after removing one item from her
final bundle, the additional $M_2$ cost is at most $r$.  Therefore, for every
long agent $j\in\{3,4\}$,
\[
    c_i(X_i)-\min_{g\in X_i}c_i(g)
    \le c_i(P_i)+r
    \le c_i(P_j)
    \le c_i(X_j),
\]
where the second inequality follows from super-canonicity.  The other short
agent $k$ receives only small items, at most two of them, so the same
super-canonicity argument gives
\[
    c_k(X_k)-1
    \le c_k(P_k)+1
    < c_k(P_k)+r
    \le c_k(P_j)
    \le c_k(X_j).
\]

The comparison between the two short agents is also safe.  Agent $i$ receives
the type-$(3,4)$ item, which is large for agent $k$, and any cross edge
assigned to $i$ is also large for agent $k$ unless it is incident to $k$,
which it is not.  Conversely, every cross edge assigned to $k$ is large for
agent $i$.  Hence, each short agent sees the other short agent's $M_2$ bundle
as sufficiently costly to cover the remaining $M_2$ cost after one item is
removed.  Formally,
\[
    c_i(X_i)-\min_{g\in X_i}c_i(g)
    \le c_i(P_i)+r
    \le c_i(P_k)+c_i(X_k\cap M_2)
    =c_i(X_k),
\]
and the symmetric inequality for agent $k$ is analogous.

Finally, agents $3$ and $4$ receive no item from $M_2$, so their own
thresholds do not increase.  Since all comparison bundles only become more
costly, their EFX inequalities follow from the canonical prefix allocation.
Therefore, the final allocation is EFX in this subcase.

\paragraph{B.3.2(b) $\boldsymbol{x_{13}+x_{14}>0}$ and $\boldsymbol{x_{23}=x_{24}=0}$.}

We distinguish two subcases.

First, suppose
\[
    x_{13}+x_{14}=2,
\]
so both $(1,3)$ and $(1,4)$ exist.  If the item $(1,2)$ exists, let agent $1$
take $(1,3)$ and $(1,4)$, let agent $2$ take $(1,2)$, and assign $(3,4)$, if
it exists, to either agent $3$ or agent $4$.  If $(1,2)$ does not exist but
$(3,4)$ exists, let agent $1$ take $(1,3)$ and $(1,4)$, and let agent $2$
take $(3,4)$.  If neither $(1,2)$ nor $(3,4)$ exists, let agent $1$ take
$(1,3)$ and let agent $2$ take $(1,4)$.

In all three allocations, agent $1$ receives only small items.  Agent $2$
receives either a small item $(1,2)$, a large item $(3,4)$, or a large item
among $(1,3),(1,4)$.

We now verify EFX carefully.

For agent $1$, after removing one of her at most two $M_2$ items, we have
\[
c_1(X_1 \setminus \{g\}) \le c_1(P_1) + 1.
\]
Since $r>2$ and by \eqref{eqn:b2-super},
\[
c_1(P_1)+1 < c_1(P_1)+r \le c_1(P_j), \quad \forall j\in\{3,4\}.
\]

In addition, agent $2$ takes at least one item in $M_2$, so 
\[
c_1(P_1)+1 \le c_1(P_2)+1 \le c_1(X_2);
\]
therefore, agent $1$ does not strongly envy any other agents.

For agent $2$, since she receives at most one large item,
\[
c_2(X_2) - 1 \le c_2(P_2) + r - 1 \le c_2(P_j),\quad \forall j\in\{3,4\}.
\]

And again by reason that $X_1$ must contain at least one item which is {large for agent $2$},
\[
c_2(X_2) \le c_2(X_1),
\]
Thus, agent $2$ does not strongly envy any other agent.

Now consider the long agents. We find that except for the first case, agents $3,4$ do not receive any more items in the other cases. Therefore, we just focus on the first case when agent $1$ takes $(1,3)$ and $(1,4)$, agent $2$ takes $(1,2)$, and agent $3$ takes $(3,4)$ without loss of generality. 
\[
c_3(X_3) - 1 = c_3(P_3) \le c_3(P_i) + r \le c_3(X_i), \quad \forall i\in\{1,2\}.
\]

This completes the EFX verification for this subcase.

Next suppose
\[
    x_{13}+x_{14}=1.
\]

Let agent $1$ take the unique item among $(1,3)$ and $(1,4)$.  If at most one
of $(1,2)$ and $(3,4)$ exists, give this item, if it exists, to agent $2$.
If both $(1,2)$ and $(3,4)$ exist, give $(1,2)$ to agent $2$; for the item
$(3,4)$, give it to agent $3$ if $x_{14}=1$, and give it to agent $4$ if
$x_{13}=1$.

The reason for this last choice is that the long agent receiving $(3,4)$ is
the one who regards agent $1$'s unique cross edge as large.  For example, if
$x_{14}=1$, then agent $3$ regards the item $(1,4)$ assigned to agent $1$ as
large, and hence
\[
{c_3(P_3)\le c_3(P_1)+r\le c_3(X_1).}
\]
Agent $3$ also receives the small item $(3,4)$, so her EFX threshold does not
increase beyond $c_3(P_3)$.  The case $x_{13}=1$ is symmetric.  The remaining
comparisons are handled by \eqref{eqn:b2-super} and by the equal-quota
canonical comparisons.

\paragraph{B.3.2(c) $\boldsymbol{x_{23}+x_{24}>0}$ and $\boldsymbol{x_{13}=x_{14}=0}$.}

{This case is symmetric to case~(b),} with the roles of agents $1$ and $2$
interchanged.  We use the same allocation rule and the same verification.

\paragraph{B.3.2(d) $\boldsymbol{x_{13}=x_{14}=x_{23}=x_{24}=0}$.}

Only items of type $(1,2)$ and $(3,4)$ may exist.  Since the graph is simple, there is at most one item of each type.

If the item $(3,4)$ exists, allocate it to a short agent whose prefix bundle
contains no small item if such a short agent exists.  If both short agents
contain small prefix items, allocate $(3,4)$ to either short agent.
Skip these if the $(3,4)$-item does not exist.
If the item $(1,2)$ also exists, allocate it to the other short agent (in the case the $(3,4)$-item does not exist, this $(1,2)$-item can be given to either short agent).

Let us verify EFX.  Suppose first that the short agent receiving $(3,4)$,
say agent $1$, has no small prefix item.  Then all items in her final bundle
are large for her.  Thus, after removing
the large item $(3,4)$ we can get
\[
    \tau_1(X)\le c_1(P_1) \le c_1(P_h) \le c_1(X_h), \quad \forall h \in \{2,3,4\}.
\]
Hence, agent $1$ is EFX. The EFX property for the remaining agents is trivial. 

Now suppose both short agents have small prefix items. Hence,
\[
    c_1(P_2)\ge c_1(P_1)+r-1,
\]
and symmetrically
\[
    c_2(P_1)\ge c_2(P_2)+r-1.
\]
Suppose the agent who is assigned item $(3,4)$ is $1$ {without loss of generality, so} 
\[
    c_1(X_1) - 1 =  c_1(P_1)+r-1 \le c_1(X_h), \quad \forall h \in \{2,3,4\},
\]
where the inequality holds because both short agents have small prefix items and the prefix allocation is super-canonical. Also, the EFX for the other agents is obvious.

This completes the proof for the case $b=2$.

\subsection{Case \texorpdfstring{$\boldsymbol{b=3}$}{b = 3}}

We now consider the case $|M_{01}|=4a+3$.  In a canonical allocation of
$M_{01}$, exactly one agent receives $a$ items and the remaining three agents
receive $a+1$ items.  Recall that the former agent is the short agent and the latter agents are long agents.

We discuss three subcases according to the structure of $M_2$.

\subsubsection{There exist two intersecting edge types}

Assume without loss of generality that $M_2$ contains items of types $(1,2)$
and $(1,3)$.

Consider a canonical allocation $P$ of $M_{01}$ in which agent $1$ is the
short agent.
\paragraph{B.4.1(a) The prefix with agent $\boldsymbol{1}$ short is super-canonical.}
Suppose first that $P$ is super-canonical.  We use one type-$(1,2)$ item and
one type-$(1,3)$ item for gap filling, both assigned to agent $1$.  Let $P'$
denote the resulting partial allocation.  We verify that $P'$ is envy-free.

For agent $1$, both gap-filling items are small.  Thus,
\[
    c_1(P'_1)=c_1(P_1)+2.
\]

Since $P$ is super-canonical and agent $1$ is short, for every
$j\in\{2,3,4\}$ we have
\[
    c_1(P_1)\le c_1(P_j)-r.
\]

As $r>2$, it follows that
\[
    c_1(P'_1)=c_1(P_1)+2 < c_1(P_1)+r \le c_1(P_j)=c_1(P'_j).
\]

Hence, agent $1$ does not envy any long agent.

Now consider a long agent $j\in\{2,3,4\}$.  Her own bundle is unchanged in
the gap-filling step.  By the canonical EFX property of $P$,
\[
    c_j(P_j)-\min_{g\in P_j}c_j(g)\le c_j(P_1).
\]

Since $\min_{g\in P_j}c_j(g)\le r$, we have
\[
    c_j(P_j)\le c_j(P_1)+r.
\]

Moreover, among the two gap-filling items assigned to agent $1$, at least one
is large for agent $j$: agent $2$ regards the type-$(1,3)$ item as large,
agent $3$ regards the type-$(1,2)$ item as large, and agent $4$ regards both
items as large.  Therefore,
\[
    c_j(P'_1)\ge c_j(P_1)+r\ge c_j(P_j)=c_j(P'_j).
\]

Thus, no long agent envies agent $1$.  Comparisons among long agents are
unchanged from the canonical prefix and have equal quota $a+1$, so they are
envy-free by the equal-quota part of Lemma~\ref{lem:canonical}(a).  Therefore,
$P'$ is envy-free.

Let $M_2^-$ be the remaining multiset of $M_2$ items.  If $M_2^-$ is not
exceptional, allocate $M_2^-$ according to Lemma~\ref{lem:M2properties}(i).
Since $P'$ is envy-free and all residual envy is certified by removing a small
item, Lemma~\ref{lem:composition} implies that the combined allocation is EFX.

Suppose now that $M_2^-$ is exceptional.  If the exceptional pair avoids
agent $1$, then both exceptional agents are long agents in the prefix $P$.
Hence, Lemma~\ref{lem:easycombine} applies directly, and the combined allocation
is EFX.

It remains to consider the case where the exceptional pair is $(1,i)$ for
some $i\in\{2,3,4\}$.  In the exceptional residual allocation, choose agent
$i$ to be the disadvantageous agent.  We only need to verify the possible
problematic comparison from agent $i$ to agent $1$.

By Lemma~\ref{lem:M2properties}(ii), the residual comparison is safe after
removing a large item:
\[
    c_i(B_i)-r\le c_i(B_1).
\]

If $P_i$ contains no small item for agent $i$, then the final bundle of agent
$i$ contains only large items in the relevant exceptional comparison.  Using
the canonical EFX inequality for the prefix,
\[
    c_i(P_i)-r\le c_i(P_1),
\]
we get
\[
\begin{aligned}
    \tau_i(X)
    &\le c_i(P_i)+c_i(B_i)-r  \\
    &\le c_i(P_1)+c_i(B_i) \\
    &\le c_i(P_1)+c_i(B_1)+r
    \le c_i(P'_1)+c_i(B_1)=c_i(X_1),
\end{aligned}
\]
where the last inequality uses the fact that at least one of the two
gap-filling items assigned to agent $1$ is large for agent $i$.

If $P_i$ contains a small item for agent $i$, then
\[
    c_i(P_i)-1\le c_i(P_1).
\]
Moreover, among the type-$(1,2)$ and type-$(1,3)$ gap-filling items assigned
to agent $1$, one is large and the other has cost at least $1$ for agent $i$.
Therefore,
\[
    c_i(P'_1)\ge c_i(P_1)+r+1.
\]
Using again $c_i(B_i)-r\le c_i(B_1)$, we obtain
\[
\begin{aligned}
    \tau_i(X)
    &=c_i(P_i)+c_i(B_i)-1 \\
    &\le c_i(P_1)+c_i(B_i) \\
    &\le c_i(P_1)+c_i(B_1)+r \\
    &\le c_i(P'_1)+c_i(B_1)
    =c_i(X_1).
\end{aligned}
\]
Thus, agent $i$ does not strongly envy agent $1$.  All other comparisons are
handled by Lemma~\ref{lem:composition} or Lemma~\ref{lem:easycombine}, and the
combined allocation is EFX.

\paragraph{B.4.1(b) The prefix with agent $\boldsymbol{1}$ short is not super-canonical.}
Now suppose that the canonical allocation with agent $1$ short is not
super-canonical.  Among the two intersecting edge types $(1,2)$ and $(1,3)$,
choose one with weakly larger multiplicity.  Without loss of generality assume
\[
    x_{12}\ge x_{13}.
\]

We now make agent $2$, rather than agent $1$, the short agent in the
$M_{01}$ prefix.  Use one type-$(1,2)$ item for gap filling, assigned to
agent $2$.  Let $P'$ denote the resulting partial allocation.

We first explain why $P'$ is envy-free.  The item assigned to agent $2$ is
small for agent $2$ and large for agents $3$ and $4$.  Thus, agents $3$ and
$4$ do not envy agent $2$ after the gap filling.  Agent $2$ receives one small
item, so her cost increases by only $1$, while the canonical prefix with
agent $2$ short gives the required one-item balance against the long agents.
Finally, because the original canonical allocation with agent $1$ short was
not super-canonical, agent $1$'s new long bundle consists only of items that
are small for agent $1$.  Therefore, agent $1$ does not envy the short bundle
of agent $2$.  Hence, $P'$ is envy-free.

Let $M_2^-$ be the residual multiset.  Since we removed one type-$(1,2)$ item
and did not remove any type-$(1,3)$ item, the residual still contains at least
one type-$(1,3)$ item.

If $M_2^-$ is not exceptional, allocate it using
Lemma~\ref{lem:M2properties}(i).  Since $P'$ is envy-free, Lemma~\ref{lem:composition}
gives an EFX allocation.

Suppose $M_2^-$ is exceptional.  Since $x_{12}\ge x_{13}$ before the removal
and $M_2^-$ still contains a type-$(1,3)$ item, the only possible exceptional
form is that $M_2^-$ contains exactly one type-$(1,3)$ item and many
type-$(2,4)$ items.  Thus, the exceptional pair is $(1,3)$.

Choose agent $3$ as the disadvantageous agent.  To verify that the allocation is EFX, we only need to verify the
possible problematic comparison from agent $3$ to agent $1$.

By Lemma~\ref{lem:M2properties}(ii),
\[
    c_3(B_3)-r\le c_3(B_1).
\]

If the current prefix bundle of agent $3$ contains no small item for agent
$3$, then the large-removal comparison is compatible with the final EFX
threshold, therefore
\[
\begin{aligned}
    \tau_3(X)
    &\le c_3(P_3)+c_3(B_3)-r \\
    &= c_3(P_1)+c_3(B_3)-r \\
    &\le c_3(P_1)+c_3(B_1)
    = c_3(X_1).
\end{aligned}
\]

If the prefix bundle of agent $3$ contains a small item, then agent $1$'s
long prefix bundle, which consists only of items small for agent $1$, is large
for agent $3$ item by item.  Hence, agent $3$ has a prefix advantage of at
least $r-1$ toward agent $1$:
\[
    c_3(P_1)\ge c_3(P_3)+r-1.
\]
Thus,
\[
\begin{aligned}
    \tau_3(X)
    &=c_3(P_3)+c_3(B_3)-1 \\
    &\le c_3(P_1)-(r-1)+c_3(B_3)-1 \\
    &=c_3(P_1)+c_3(B_3)-r \\
    &\le c_3(P_1)+c_3(B_1)
    =c_3(X_1).
\end{aligned}
\]
Therefore, agent $3$ does not strongly envy agent $1$, and the combined
allocation is EFX.

\subsubsection{No two edges intersect, and there are two edge types}

Assume without loss of generality that the two edge types are $(1,2)$ and
$(3,4)$, and that
\[
    x_{34}\ge x_{12}.
\]
Consider all super-canonical allocations of $M_{01}$.

\paragraph{B.4.2(a) Some super-canonical allocation has agent $\boldsymbol{1}$ or $\boldsymbol{2}$ short.}
Suppose first that there exists a super-canonical allocation $P$ in which one
of agents $1,2$ is short.  Without loss of generality, assume agent $1$ is
short.  We use one type-$(1,2)$ item and one type-$(3,4)$ item for
gap filling, both assigned to agent $1$.  Let $P'$ denote the resulting
partial allocation.

We first record the relevant inequalities.  Since $P$ is super-canonical and
agent $1$ is short, for every long agent $j\in\{2,3,4\}$,
\[
    c_1(P_1)\le c_1(P_j)-r.
\]

From agent $1$'s perspective, the type-$(1,2)$ item is small and the
type-$(3,4)$ item is large.  Hence,
\[
    \tau_1(P')=c_1(P_1)+r\le c_1(P_j)=c_1(P'_j)
    \qquad\forall j\in\{2,3,4\}.
\]

Thus, agent $1$ may envy others in the gap-filled prefix, but only up to a
small item.

For every agent $j\in\{2,3,4\}$, at least one of the two gap-filling items
assigned to agent $1$ is large for agent $j$: agent $2$ regards the
type-$(3,4)$ item as large, while agents $3$ and $4$ regard the type-$(1,2)$
item as large.  Hence, the same argument as above shows that none of agents
$2,3,4$ envies agent $1$.  Comparisons among agents $2,3,4$ are unchanged
from the equal-quota canonical prefix.  

Let $M_2^-$ be the residual multiset.  We apply Lemma~\ref{lem:M2properties}(iii)
to allocate $M_2^-$ so that agent $1$ receives a residual-favorite bundle:
\[
    c_1(B_1)\le c_1(B_j)\qquad\forall j.
\]

If $M_2^-$ is not exceptional, then all other residual envy is certified by
removing a small item.  Combining the displayed inequality with
\[
    \tau_1(P')\le c_1(P'_j)
\]
shows that agent $1$ remains EFX after concatenation, while all other agents
are handled by Lemma~\ref{lem:composition}.  Hence, the combined allocation is
EFX.

Suppose $M_2^-$ is exceptional.  Since the residual allocation was chosen so
that agent $1$ is residual-favorite, the only possible exceptional pair is
$(1,2)$, and agent $2$ is chosen as the disadvantageous agent.  We only need
to check that agent $2$ does not strongly envy agent $1$ in the final
allocation.

If $P_2$ contains no small item for agent $2$, then the final comparison is
compatible with removing a large item, and the exceptional residual inequality
\[
    c_2(B_2)-r\le c_2(B_1)
\]
together with the canonical prefix inequality gives EFX.

If $P_2$ contains a small item for agent $2$, then
\[
    c_2(P_2)-1\le c_2(P_1).
\]

Moreover, among the two gap-filling items assigned to agent $1$, the
type-$(1,2)$ item is small for agent $2$, while the type-$(3,4)$ item is large
for agent $2$.  Therefore,
\[
    c_2(P'_1)=c_2(P_1)+r+1.
\]

Using $c_2(B_2)-r\le c_2(B_1)$, we get
\[
\begin{aligned}
    \tau_2(X)
    &=c_2(P_2)+c_2(B_2)-1 \\
    &\le c_2(P_1)+c_2(B_2) \\
    &\le c_2(P_1)+c_2(B_1)+r \\
    &\le c_2(P'_1)+c_2(B_1)
    =c_2(X_1).
\end{aligned}
\]
Thus, agent $2$ does not strongly envy agent $1$, and the combined allocation
is EFX.

\paragraph{B.4.2(b) Agents $\boldsymbol{1}$ and $\boldsymbol{2}$ are long in every super-canonical allocation.}
Now suppose both agents $1$ and $2$ are long in every super-canonical
allocation.  Choose any super-canonical allocation $P$.  Then the unique short
agent belongs to $\{3,4\}$; assume without loss of generality that agent $4$
is short.  Assign one type-$(1,2)$ item to agent $4$ for gap filling, and let
$P'$ be the resulting partial allocation.

We first verify that $P'$ is envy-free.  Since agents $1$ and $2$ are long in
every super-canonical allocation, neither of them can be chosen as the short
agent in the construction of Lemma~\ref{lem:canonical}(b).  Equivalently,
agents $1$ and $2$ are heavy.  Thus, in the present canonical allocation,
their long bundles contain only items that are small for themselves.  In
particular,
\[
    c_i(P_i)=a+1
    \qquad\text{for } i=1,2.
\]

Let $h$ be the type-$(1,2)$ item assigned to agent $4$.  Since $h$ is small
for agents $1$ and $2$, we have
\[
    c_i(P'_4)=c_i(P_4)+1\ge a+1=c_i(P_i)
    \qquad\text{for } i=1,2.
\]

Hence, neither agent $1$ nor agent $2$ envies agent $4$.

Next consider agent $3$.  The item $h$ is large for agent $3$.  By the
canonical EFX property of $P$,
\[
    c_3(P_3)-\min_{g\in P_3}c_3(g)\le c_3(P_4).
\]

Since $\min_{g\in P_3}c_3(g)\le r$, we obtain
\[
    c_3(P_3)\le c_3(P_4)+r=c_3(P'_4).
\]

Thus, agent $3$ does not envy agent $4$.

Now consider the short agent $4$.  Since $P$ is super-canonical and agent $4$
is short, for every long agent $j\in\{1,2,3\}$,
\[
    c_4(P_4)\le c_4(P_j)-r.
\]

The gap-filling item $h$ is large for agent $4$, so
\[
    c_4(P'_4)=c_4(P_4)+r\le c_4(P_j)=c_4(P'_j)
    \qquad\forall j\in\{1,2,3\}.
\]

Hence, agent $4$ does not envy any long agent.  Finally, comparisons among
agents $1,2,3$ are unchanged, and these agents have the same prefix quota
$a+1$, so the equal-quota canonical argument gives envy-freeness among them.
Therefore, $P'$ is envy-free.

Let $M_2^-$ be the residual multiset after removing the gap-filling item
$h$.  If $M_2^-$ is not exceptional, allocate it using
Lemma~\ref{lem:M2properties}(i).  Since $P'$ is envy-free and all residual
envy is certified by removing a small item, Lemma~\ref{lem:composition}
implies that the combined allocation is EFX.

It remains to consider the case where $M_2^-$ is exceptional.  Since before
the gap-filling step we assumed
\[
    x_{34}\ge x_{12},
\]
and we removed one type-$(1,2)$ item, the residual multiplicities satisfy
\[
    x_{34}>x_{12}^-,
\]
where $x_{12}^-=x_{12}-1$.  Thus, if $M_2^-$ is exceptional, the dominant
edge type must be $(3,4)$, and the exceptional agents are $\{1,2\}$.

Both exceptional agents $1$ and $2$ are long agents in the prefix.  Therefore,
we may apply Lemma~\ref{lem:easycombine} to the envy-free prefix $P'$ and the
exceptional residual allocation of $M_2^-$.  Hence, the combined allocation is
EFX.

\subsubsection{Only one edge type exists}
\label{append:single-type}

Assume without loss of generality that all items of $M_2$ have type $(1,2)$. Again consider all super-canonical allocations of $M_{01}$.

\paragraph{B.4.3(a) Some super-canonical allocation has agent $\boldsymbol{1}$ or $\boldsymbol{2}$ short.}
Suppose there exists a super-canonical allocation $P$ in which one of agents
$1,2$ is short.  Without loss of generality, assume agent $1$ is short.

If $|M_2|\le r$, assign all items of $M_2$ to agent $1$.  Since all these
items are small for agent $1$, for every long agent $j$,
\[
    \tau_1(X)\le c_1(P_1)+|M_2|-1
    < c_1(P_1)+r
    \le c_1(P_j)\le c_1(X_j).
\]

Thus, agent $1$ does not strongly envy any long agent.  Every long agent sees
agent $1$'s added items as either small or large; in particular, the comparison
bundle of agent $1$ only becomes more costly.  Since the prefix was
super-canonical, the resulting allocation is EFX.

Now assume $|M_2|>r$.  Assign $\lfloor r\rfloor$ type-$(1,2)$ items to agent
$1$ as gap-filling items.  Let $M_2^-$ be the remaining type-$(1,2)$ items.
We allocate $M_2^-$ in round-robin order.  If $P_3$ contains no small item for
agent $3$, use the order
\[
    1,2,3,4.
\]
Otherwise, use the order
\[
    1,2,4,3.
\]

The first $\lfloor r\rfloor$ gap-filling items ensure that agent $1$ is
protected against all long agents by super-canonicity:
\[
    \tau_1(P')\le c_1(P_1)+\lfloor r\rfloor-1
    < c_1(P_1)+r
    \le c_1(P_j)
    \qquad\forall j\in\{2,3,4\}.
\]

The round-robin allocation of the residual type-$(1,2)$ items makes the
residual bundle sizes differ by at most one.  Agents $1$ and $2$ receive
small residual items, while agents $3$ and $4$ receive large residual items.

If the residual size is not congruent to $3$ modulo $4$, the round-robin
residual allocation is non-exceptional, and Lemma~\ref{lem:composition}
applies straightly because the non-exceptional round-robin allocation must be EFX (notice that proving agent $2$ does not strongly envy agent $1$ requires the second part of Lemma~\ref{lem:composition}: we have $c_2(P_2')-c_2(P_1')=c_2(P_2)-c_2(P_1)-\lfloor r\rfloor\leq r-\lfloor r\rfloor\leq1$ after the gap-filling and $c_2(B_1)-c_2(B_2)\geq0$ due to the round-robin order).  If the residual size is congruent to $3$ modulo $4$, exactly one of
agents $3,4$ is the disadvantageous agent in the residual comparison.  The
chosen order determines which of agents $3,4$ this is.

If $P_3$ contains no small item for agent $3$, we use the order $1,2,3,4$,
so agent $3$ is allowed to be the disadvantageous agent; her final bundle then
contains only large items in the relevant comparison, so large-item removal is
compatible with EFX.  If $P_3$ contains a small item, we use the order
$1,2,4,3$, so agent $4$ is the disadvantageous agent. Now if $P_4$ contains no small item for agent $4$, the overall allocation is EFX as we just stated for agent $3$ and $P_3$. Otherwise, if $P_4$ contains at least one small item for agent $4$, and we know the bundle $P_3$ contains an item that is large for agent $4$ (since $P_3$ contains an item in $M_{01}$ that is small for agent $3$, this item must be large for agent $4$ by the property of $M_{01}$). Hence,
\[
    c_4(P_3)\ge c_4(P_4)+r-1.
\]
This $r-1$ prefix slack compensates exactly for the difference between
large-item removal in the residual exceptional comparison and small-item
removal in the final EFX threshold.  Thus, the same calculation as in
Lemma~\ref{lem:easycombine} shows that the combined allocation is EFX.

\paragraph{B.4.3(b) Agents $\boldsymbol{1}$ and $\boldsymbol{2}$ are long in every super-canonical allocation.}
Finally, suppose both agents $1$ and $2$ are long in every super-canonical
allocation.  Choose any super-canonical allocation $P$.  Then the unique short
agent lies in $\{3,4\}$.  Assume without loss of generality that agent $4$ is
short.

Use one type-$(1,2)$ item as a gap-filling item for agent $4$.  This item is
large for agent $4$ and also large for agent $3$, while it is small for
agents $1$ and $2$.  By super-canonicity, the short agent remains protected
against all long agents, and the long agents do not envy the short agent after
the gap filling. The proof here is similar to \textbf{B.4.2(b)}. Thus, the current partial allocation is envy-free.

Let $M_2^-$ be the residual multiset.  If $M_2^-$ is not exceptional, which
is exactly the case $|M_2^-|\not\equiv 3\pmod 4$, allocate it using
Lemma~\ref{lem:M2properties}(i) and apply Lemma~\ref{lem:composition}.

If $M_2^-$ is exceptional, then the only possible exceptional pair is
$(3,4)$.  Choose agent $3$ as the disadvantageous agent.  We verify the only
potentially problematic comparison from agent $3$ to agent $4$.

If $P_3$ contains no small item for agent $3$, then agent $3$'s final bundle
contains only large items in the relevant comparison, and the large-item
removal from the exceptional residual allocation is valid in the combined
allocation.

If $P_3$ contains a small item for agent $3$, then the gap-filling item
assigned to agent $4$ is of type $(1,2)$, which is large for agent $3$.
Hence,
\[
    c_3(P'_4)\ge c_3(P_4)+r.
\]

Moreover, the prefix EFX inequality gives
\[
    c_3(P_3)-1\le c_3(P_4).
\]

The exceptional residual comparison gives
\[
    c_3(B_3)-r\le c_3(B_4).
\]

Therefore,
\[
\begin{aligned}
    \tau_3(X)
    &=c_3(P_3)+c_3(B_3)-1 \\
    &\le c_3(P_4)+c_3(B_3) \\
    &\le c_3(P_4)+c_3(B_4)+r \\
    &\le c_3(P'_4)+c_3(B_4)
    =c_3(X_4).
\end{aligned}
\]

Thus, agent $3$ does not strongly envy agent $4$.  All other comparisons are
handled by Lemma~\ref{lem:composition}.  Hence, the combined allocation is EFX.

This completes the proof for the case $b=3$.

\end{document}